\pdfoutput=1
\documentclass[12pt,reqno]{amsart}
\usepackage[margin=1in]{geometry}
\usepackage{xcolor,graphicx,subcaption,amsmath,amssymb,amsthm,amsfonts,bm,mathtools}
\usepackage{array,booktabs,tabularx}
\usepackage{algorithm}
\usepackage{algpseudocode}
\usepackage[longnamesfirst]{natbib}
\usepackage{bibunits}
\usepackage[hidelinks,hypertexnames=false]{hyperref}
\newcommand{\piPoLeCe}{\widehat{\pi}_{\mathrm{PoLeCe}}}
\newcommand{\piEWM}{\widehat{\pi}_{\mathrm{EWM}}}

\newcommand{\E}{\mathrm{E}}
\DeclareMathOperator*{\argmax}{argmax}
\DeclareMathOperator*{\argmin}{argmin}
\newtheorem{lemma}{Lemma}
\newtheorem{prop}{Proposition}
\theoremstyle{remark}

\defaultbibliography{PoLeCe_references}
\defaultbibliographystyle{econometrica}
\theoremstyle{plain}
\newtheorem{aprop}{Proposition}[section]
\newtheorem{alemma}[aprop]{Lemma}
\newtheorem{acorollary}[aprop]{Corollary}

\theoremstyle{remark}
\newtheorem{aremark}[aprop]{Remark}
\makeatletter
\renewcommand\paragraph{\@startsection{paragraph}{4}{\parindent}{\z@}{-\fontdimen2\font}{\normalfont\itshape}}
\makeatother

\begin{document}
\begin{bibunit}
\title[Policy Learning with Confidence]{Policy Learning with Confidence{$^\ast$}}

\author[Victor Chernozhukov]{Victor Chernozhukov{$^\dagger$}}
\thanks{{$^\ast$}This is a revised version of arXiv working paper 2502.10653 originally posted in February 2025. We thank Isaiah Andrews, Jiafeng (Kevin) Chen, Nathan Hendren, Toru Kitagawa, Matt Masten, Adam McCloskey, Whitney Newey, Vira Semenova, Jesse Shapiro, Aleksey Tetenov and conference and seminar participants at the Munich Econometrics Workshop, Duke University,  CUNY, University of Geneva, Tilburg University, University of Chicago, the Encounters in Econometric Theory workshop, CUHK Shenzhen, IAAE, 
Goethe University Frankfurt, the 2025 ESWC, and the 2026 ASSA meeting for helpful comments and discussion.  Tian Xie provided excellent research assistance. Generative AI was used to proofread exposition and check results in the initial draft and the revision, and for assistance in structuring the revision and developing new material for the Online Appendices.  The authors are responsible for any errors. Liyang Sun acknowledges generous support from the Economic and Social Research Council (new investigator grant UKRI607).%
\\ 
\indent{$^\dagger$}Department of Economics, MIT, and University College London. Email: vchern@mit.edu}

\author[Sokbae Lee]{Sokbae Lee{$^\ddagger$}}
\thanks{{$^\ddagger$}Department of Economics, Columbia University. Email: sl3841@columbia.edu}

\author[Adam M. Rosen]{Adam M. Rosen{$^\S$}}
\thanks{{$^\S$}Department of Economics, Duke University. Email: adam.rosen@duke.edu}

\author[Liyang Sun]{Liyang Sun{$^\P$}}
\thanks{{$^\P$}Department of Economics, University College London. Email: liyang.sun@ucl.ac.uk}
\date{October 2026}  

\begin{abstract}

This paper introduces a rule for policy selection in the presence of estimation uncertainty, explicitly accounting for estimation risk.  The rule belongs to the class of risk-aware rules on the efficient decision frontier, characterized as policies offering maximal estimated welfare for a given level of estimation risk. Among this class, the proposed rule is chosen to provide a reporting guarantee, ensuring that the welfare delivered exceeds a threshold with a pre-specified confidence level. We apply this approach to the allocation of a limited budget among social programs using estimates of their marginal value of public funds and associated standard errors.
\end{abstract}

\maketitle

\noindent\texttt{Keywords:}  budget allocation, risk-aware policy learning, statistical decision theory
\smallskip

\noindent\texttt{JEL classification codes:} C14, C44, C52.

\newpage

\section{Introduction}

Consider a decision maker (DM) faced with choosing from a menu of policies to maximize expected welfare. The DM could be a planner deciding how to allocate treatments to heterogeneous individuals, a firm choosing different potential innovations in which to invest resources, an auctioneer choosing an optimal auction design or reserve price, or a policy-maker choosing expenditure shares to allocate to government programs to maximize a measure of public welfare. If the DM knew the welfare that would be obtained from implementing each policy, the choice would be easy; she would then simply choose the policy that would yield the highest welfare. However, the DM lacks such precise knowledge, but instead has data from one or possibly multiple studies that can be used to consistently estimate the welfare that would be achieved by each policy. The estimates are measured with varying degrees of precision, as reflected by their standard errors. How should the DM use the available information for policy selection?

An intuitively appealing choice is to select the policy with the highest estimated welfare, which simply replaces the population objective with its sample analog in determining policy choice. In the literature on policy learning, this is referred to as empirical welfare maximization (EWM).\footnote{See for example \cite{KT2018} and \cite{AW2021}. \cite{andrews2024inference} alternatively refer to this as the ``natural rule'' or ``picking the winners''.} The same logical task applies more broadly to a variety of economic contexts, for example, when researchers use structural models, they obtain estimates of the performance of different policies that can then be used to inform policy choice.

Unfortunately, the available estimates for the welfare of different choices are generally not identical to their population values. Estimation error could result in the EWM rule delivering suboptimal policy choice. It may be observed, for example, that among a menu of options policy A has the highest estimated welfare with policy B coming in a close second, but that the standard error of policy B is much smaller than that of policy A, reflecting that the welfare of policy B is estimated much more precisely. Could it be that policy B is actually better than policy A, and that the higher estimated welfare of policy A is simply down to a lack of precision in the estimates? Should the DM consider choosing policy B since it is estimated more precisely to hedge against estimation error?

We propose a risk-aware (RW) rule for selecting policies with the goal of maximizing
expected welfare in the presence of estimation uncertainty, explicitly accounting for
estimation risk in the decision rule. As inputs to her decision the DM has consistent estimates
for the welfare of each available policy, and their corresponding sample variance
matrix, with sample estimates approximately normally distributed around their population
means.\footnote{For estimates from independent samples, the covariance matrix is diagonal.} RW rules account for uncertainty by penalizing the estimated welfare of each policy by an amount proportional to its standard error, thereby selecting a policy
from the (Pareto) efficient decision frontier that balances performance and precision as
measured by trading off higher sample estimates against smaller sample variance. The
exact tradeoff is governed by the DM’s choice of 
the penalty factor.

Our proposed rule chooses the penalty factor to give a reporting guarantee. Following the confidence band construction of \cite{CLR2013}, we form a lower confidence bound for each policy's welfare and select the policy with the highest bound. The maximized bound then gives a welfare level that the selected policy achieves with the stated confidence level. We call this rule \emph{Policy Learning with Confidence}, or PoLeCe. The analysis below establishes both its reporting guarantee and its regret bound.

To illustrate, consider a setting in which a budget-conscious DM needs to allocate funds across several social programs that affect the same group of individuals, either in response to a slight budget increase or for incremental cost-cutting.
Suppose the DM wishes to use the Marginal Value of Public Funds (MVPF), spearheaded by \cite{Hendren/Sprung-Keyser:2020}, for this purpose. The MVPF of each social program is the ratio of its marginal benefit to its net marginal cost to the government, inclusive of the impact of any behavioral responses on the government budget. For instance, an MVPF of 1.5 indicates that each \$1 of net government spending generates \$1.50 in benefits for beneficiaries. For a utilitarian DM, the MVPF provides a  metric to compare the ``bang for the buck'' of different policies that target the same beneficiaries.  A DM interested in allocating additional marginal expenditure will prioritize programs with the highest MVPF to maximize welfare.\footnote{Here we abstract from distributional incidence and assume the DM values dollars equally across the beneficiaries of each policy.  If welfare weights vary across beneficiaries, as in the original framework of \cite{Hendren/Sprung-Keyser:2020}, our analysis can accommodate this by adjusting each MVPF in the value function using the DM's designated weights. Example welfare weights have been derived in \cite{hendren2020measuring}. Other metrics for measuring programs' marginal benefits and costs have also been discussed in \cite{garcia2022three}.}

In practice the DM must make decisions based on estimated MVPFs. The Policy Impacts Library, \cite{policy_impacts_library}, provides easily accessible MVPF estimates online while also noting significant estimation uncertainty in some cases.  Using the EWM rule to allocate funds would direct the entire budget to the program with the highest MVPF estimate. However, the DM may  be reluctant to direct the entire budget to a single social program if this MVPF estimate has a large standard error relative to other programs whose MVPF estimates are nearly as high, and which have smaller standard errors.  How should the DM balance the tradeoff between the size of the estimated MVPF and the precision with which it is estimated?
\begin{center}
    [FIGURE~\ref{fig:frontier_MVPF} AROUND HERE]
\end{center}

This tradeoff is illustrated in Figure~\ref{fig:frontier_MVPF} using the MVPF estimates from the Policy Impacts Library, which comprises 14 US programs after we focus on estimates based on randomized controlled trials (RCTs) with finite standard errors. To mimic a DM who applies constant welfare weights within certain age groups but not necessarily across them due to distributional concerns, we use the Policy Impacts Library classification of programs based on intended beneficiaries' age. We then consider budget allocations separately for two distinct age groups. Six of these programs target beneficiaries age 25 and under, covering early childhood education programs, college financial aid, and job training.\footnote{As explained in the Policy Impacts Library, alternative specifications have led to larger MVPF estimates for early childhood education depending on how lifetime earnings are forecasted and whether one accounts for the transfer value of preschool subsidies to  parents.} Eight of these programs target beneficiaries over age 25, covering federal social assistance, health insurance, and housing vouchers. In both cases, the MVPF of the programs selected by EWM is not estimated as precisely as the welfare delivered by diversifying the budget across multiple programs.  As demonstrated in Figure~\ref{fig:frontier_MVPF}, our proposal selects a point on the efficient decision frontier that balances estimated welfare and estimation risk. In each panel of Figure~\ref{fig:frontier_MVPF}, the PoLeCe allocation is the point at which the green dashed line, whose slope is the critical value used to construct the lower confidence bound, is tangent to the frontier. 

In the above example, all that was required to map out the efficient decision frontier, and indeed all that is needed to determine the optimal allocation following our proposal, are point estimates for the welfare of each available policy and their joint variance-covariance matrix. This can be obtained in applications that feature a variety of different models and sampling processes. An important special case to which our framework applies is the analysis of optimal treatment assignment using data on individual-specific treatments and allocations, such as from an RCT. This is demonstrated in Online Appendix \ref{sec:example treatment policies} with three additional applications spanning the study of treatments that encourage immunization, enable informal savings technologies, and encourage sobriety among low-wage workers.

Our paper relates to three strands of literature, which we review in more detail in Online Appendix~\ref{sec:literature}. The first is statistical decision theory, which offers Bayes, maximin, and minimax-regret rules for decisions based on sample data \citep{Wald:50,Manski2004,stoye_minimax_2009,hirano_asymptotics_2009,Manski2021}. The PoLeCe rule, in the terminology of \cite{Manski2021}, is a form of ``as-if optimization with set estimates.'' The second strand is the literature on policy learning, which establishes regret bounds for EWM and related rules \citep{KT2018,AW2021}. We show how a penalty on estimation risk changes these bounds and when it improves them. The third strand concerns decisions that come with inferential guarantees, deemed P-certified decisions in \cite{andrews2025certifieddecisions}. What distinguishes PoLeCe within this class is that its certificate is delivered directly from maximizing confidence bands from \cite{CLR2013} in a single stage without explicitly computing a maximin rule over a first-stage set estimate.

The rest of the paper proceeds as follows. Section~\ref{sec:motivation and risk-aware policy choice} introduces risk-aware rules and the PoLeCe rule in an idealized setting, in which the welfare estimates are exactly Gaussian and their covariance matrix is known. This setting allows us to explain the main ideas and to state the main guarantees with little technical apparatus. Section~\ref{Sec: Reg Guarantees of RW policies} states precise conditions under which the guarantees continue to hold approximately, gives the formal results, and discusses the role of the policy class. Section~\ref{sec: empirical examples} describes how PoLeCe is computed and returns to the empirical application of allocation of public funds across social programs. Section~\ref{sec: conclusion} concludes. Proofs are collected in the appendix. The Online Appendix contains a more detailed review of the literature, computational details, alternative motivations for risk-aware rules, three applications to treatment choice with calibrated simulations, and further extensions.

\section{Risk-Aware Policy Choice}\label{sec:motivation and risk-aware policy choice}

\subsection{Setting}
Throughout the paper, a decision maker (DM) chooses a policy from a menu $\Pi$ of $p$ policies. Each policy $\pi$ yields welfare $V(\pi)$, which is unknown to the DM. Instead, she has an estimate $\widehat V(\pi)$ of the welfare of each policy, together with its standard error $\widehat{s}(\pi)>0$. In this section, to convey the main ideas without technical complications, we assume the estimates are exactly Gaussian,
\begin{equation}\label{eq:ideal}
 \widehat V(\pi)=V(\pi)+\widehat{s}(\pi)Z_\pi,\qquad (Z_\pi)_{\pi\in\Pi}\sim N(0,\widehat{C}),
\end{equation}
where $\widehat{s}(\pi)$ and $\widehat{C}$ are known standard errors and a known correlation matrix, and are treated as fixed. In Section~\ref{Sec: Reg Guarantees of RW policies} we drop this simplification and invoke approximate Gaussianity, and show that the conclusions of this section continue to hold up to vanishing approximation errors.

As a motivating example, a public or private agency considers programs $\{1,\ldots,J\}$ and chooses an investment allocation from a nonempty menu
\begin{equation}\label{Pi Fraction Allocation Problem}
 \Pi\subseteq
 \left\{\pi\in\mathbb R^J:\sum_{j=1}^J\pi_j=1,
 \quad 0\leq a_j\leq\pi_j\leq b_j\leq1\right\}.
\end{equation}
The bounds may reflect diversity or other requirements. The welfare of allocation \(\pi\) is $V(\pi)=\pi'R$, where \(R\) is a $J$-vector of the rate of return measures for each program, for instance the ratios of marginal benefit to net cost of funds used in public finance. When the estimate $\widehat R$ of $R$ is distributed as $N(R,\Omega/n)$ with $\widehat{\Omega}=\Omega$ known in this idealized setting, the welfare estimate $\widehat V(\pi)=\pi'\widehat R$ takes the form~\eqref{eq:ideal} with standard error
\begin{equation}\label{eq:budget value}
 \widehat{s}(\pi)=\sqrt{\pi'\widehat{\Omega}\pi/n},
\end{equation}
and correlation matrix such that for any two policies $\pi^{(i)}$ and $\pi^{(j)}$
\begin{equation}\label{C_hat_allocation_problem_finite}
\widehat{C}_{ij} = \frac{\pi^{(i)\prime}\widehat{\Omega} \pi^{(j)} }{n \widehat{s}(\pi^{(i)})\widehat{s}(\pi^{(j)})}\text{.}
\end{equation}
Here $n$ is a notional sample size, which we use to study how the rules behave as more information is acquired.\footnote{In the empirical application, $\widehat{\Omega}$ is diagonal, because the estimates come from independent studies.} For simplicity, we present the results for a finite menu of allocations, but as explained in Section~\ref{sec:policy class}, the restriction to a finite menu is not necessary.  In our application we allow the agency to choose any allocation in the probability simplex in $\mathbb R^J$, which is a continuum menu of the form~\eqref{Pi Fraction Allocation Problem}.

The same criterion applies to treatment choice. If one of $J$ treatments must be assigned, without allowing fractional allocations, the menu is
\begin{equation}\label{Pi discrete}
 \Pi=\left\{\pi\in\{0,1\}^J:\sum_{j=1}^J\pi_j=1\right\}.
\end{equation}
Allowing the assignment to depend on a finite set of covariate values gives a larger finite menu. Online Appendix~\ref{sec:example treatment policies} illustrates this setting with data from randomized experiments.

\subsection{Risk-Aware Rules}
The agency faces a tradeoff familiar from \citet{Markowitz:52}, as a higher estimated return may come with greater uncertainty. We account for this uncertainty by subtracting a multiple of the standard error from estimated welfare. For $k\geq0$, define
\begin{equation}\label{eq:polece}
 \widehat\pi_{\mathrm{RW}}(k)
 \in\argmax_{\pi\in\Pi}\{\widehat V(\pi)-k \widehat{s}(\pi)\}.
\end{equation}
The penalty $k$ determines how much estimated welfare the agency will give up for greater precision. The choice $k=0$ corresponds to the EWM rule, $\piEWM\in\argmax_{\pi\in\Pi}\widehat V(\pi)$, which ignores estimation risk. For $k>0$ each selected policy is efficient, meaning that another policy cannot offer both higher estimated welfare and lower estimation risk.\footnote{More precisely, no policy can weakly improve both criteria and strictly improve at least one. When the menu is finite, there can be efficient policies that do not maximize a linear criterion of the form~\eqref{eq:polece}. Online Appendix~\ref{sec:frontier algorithm} gives an algorithm for the frontier selected by these criteria.}

The risk-aware rules defined  in~\eqref{eq:polece} have a simple geometric interpretation, which is illustrated in Figure~\ref{fig:frontier_MVPF} for the MVPF allocation problem. A single point is plotted for each policy, with $\widehat{s}(\pi)$ on the horizontal axis and its estimated welfare $\widehat V(\pi)$ on the vertical axis. The nondominated part of the upper boundary of this set of points is the efficient decision frontier. Since the criterion $\widehat V(\pi)-k\,\widehat{s}(\pi)$ is constant along straight lines with slope $k$, the rule $\widehat\pi_{\mathrm{RW}}(k)$ selects the point at which the highest such line touches the frontier, and the intercept of this line at $s=0$ equals the maximized value of the criterion. A small penalty gives a flat line, which touches the frontier at a policy with high estimated welfare and possibly a large standard error. A large penalty gives a steep line, which touches the frontier at a more precisely estimated policy.

Thus risk-aware rules balance estimated  returns against estimation risk.  The balance is analogous to Markowitz portfolio selection, although here the standard error $\widehat{s}(\pi)$ measures imprecision of the estimated return rather than the variance of realized returns. Proposition \ref{prop:DM1} in Online Appendix \ref{sec: RW Decision} formalizes the link to Markowitz portfolio allocation by showing that a risk-averse or risk-neutral DM with Gaussian posterior beliefs for welfare with mean and variance matching the empirical estimates will always select an allocation $\pi$ from the set of risk aware rules characterized in \eqref{eq:polece} for some $k$.

\subsection{Regret of RW Rules}
RW rules have favorable regret properties when the precision of welfare estimates varies across policies. Regret is the shortfall in welfare of the selected policy relative to the
best policy on the menu, formally defined as \(V_{\max}-V(\widehat\pi_{\mathrm{RW}}(k))\), where \(V_{\max}=\max_{\pi\in\Pi}V(\pi)\). To preview the formal results in Section~\ref{Sec: Reg Guarantees of RW policies}, let \(\Pi_0\) denote the set of welfare-maximizing policies, and define
\begin{equation}\label{def:sigma_bars}
\bar{\sigma}_{\Pi} \;:=\; \sqrt{n}\,\max_{\pi \in \Pi} \widehat s(\pi),
\quad
\underline{\sigma}_{\Pi_0} \;:=\; \sqrt{n}\,\min_{\pi \in \Pi_0} \widehat s(\pi),
\end{equation}
Let \(q_{1-\beta,\Pi}\) denote the \((1-\beta)\)-quantile of the maximum of $Z_{\pi}$ from \eqref{eq:ideal} over the policy menu. Proposition~\ref{thm:RWregret} establishes high-probability regret bounds for risk-aware rules. When \(k=0\), corresponding to EWM, the bound is driven by \(\bar \sigma_{\Pi} \), the largest normalized standard error across all policies. In contrast, whenever \(k\ge q_{1-\beta,\Pi}\), the bound depends on standard errors only through \(\underline{\sigma}_{\Pi_0}\), the smallest normalized standard error among the optimal policies. Thus, sufficiently penalizing estimation risk can yield substantially sharper regret guarantees when precision varies considerably across policies.

\subsection{Policy Learning with Confidence}\label{sec:polece ideal}
PoLeCe is the risk-aware rule with penalty $k=q_{1-\alpha,\Pi}$. For any constant $k$,
\[
 V(\pi)\geq\widehat V(\pi)-k\,\widehat{s}(\pi)\ \text{for every }\pi\in\Pi
 \quad\Longleftrightarrow\quad
 \max_{\pi\in\Pi}Z_\pi\leq k .
\]
That is, the penalized estimates lie below the true welfare of every policy exactly when the largest standardized estimation error does not exceed $k$. This event has probability $1-\alpha$ when $k=q_{1-\alpha,\Pi}$. Let
\[
 LV_{1-\alpha}(\pi)=\widehat V(\pi)-q_{1-\alpha,\Pi}\,\widehat{s}(\pi)
\]
denote a lower confidence bound for $V(\pi)$ at level $1-\alpha$ that holds simultaneously for all policies on the menu. Because it holds for all policies at once, it also holds for the policy that the DM selects, and no post-estimation correction for selection is needed. PoLeCe selects a policy that attains the largest lower bound,
\[
 \piPoLeCe\in\argmax_{\pi\in\Pi}LV_{1-\alpha}(\pi),
\]
and reports the maximized bound $LV_{1-\alpha,\Pi}=\max_{\pi\in\Pi}LV_{1-\alpha}(\pi)$. PoLeCe achieves this level of welfare or higher with probability at least $1-\alpha$. Because PoLeCe corresponds to the rule that maximizes worst-case welfare within this confidence band, it falls within the class of rules \citet{Manski2021} calls ``as-if optimization with set estimates''.

\section{Formal Guarantees}\label{Sec: Reg Guarantees of RW policies}

We now drop the simplifications of Section~\ref{sec:motivation and risk-aware policy choice} and analyze a deterministic finite policy menu $\Pi=\{\pi^{(1)},\ldots,\pi^{(p)}\}$, with $p=p_n$ potentially much larger than $n$.\footnote{If $\tilde\pi$ is a control policy  with known welfare $\widehat{V}(\tilde\pi)$, then $\widehat{s}(\tilde\pi)$ and $Z_{\tilde\pi}$ below are both set to zero. Alternatively, the welfare of all policies can be redefined relative to that of the known control policy.} The policies and their number may vary with $n$. The same analysis applies conditionally on an independent sample used to construct the menu, provided the conditions below hold conditionally. Section~\ref{sec:policy class} discusses infinite classes.

For a measurable, possibly data-dependent penalty $\widehat k\geq0$, choose
\[
 \widehat\pi_{\mathrm{RW}}(\widehat k)\in
 \argmax_{\pi\in\Pi}\{\widehat V(\pi)-\widehat k\widehat s(\pi)\}.
\]
 In the allocation problem, for example, $\widehat s(\pi)$ is as defined in \eqref{eq:budget value}. Define the $p$-dimensional vector of studentized estimation errors by
\[
 \widehat Z_\pi=\frac{\widehat V(\pi)-V(\pi)}{\widehat s(\pi)},
 \qquad \pi\in\Pi.
\]
We compare its distribution with that of a centered Gaussian vector $(Z_\pi)_{\pi\in\Pi}$ with standard normal marginals. Let $\mathcal A$ be the collection of axis-aligned rectangles in $\mathbb R^p$, allowing unbounded endpoints. The following condition controls joint estimation error across all policies, including highly correlated policies.

\begin{itemize}
\item[(G)] There exists a sequence of non-negative constants \(r_n\) with \(r_n \searrow 0\) such that,
\begin{equation}
\Bigl\lvert
\Pr\!\bigl((\widehat{Z}_\pi)_{\pi \in \Pi} \in A\bigr)
\;-\;
\Pr\!\bigl((Z_\pi)_{\pi \in \Pi} \in A\bigr)
\Bigr\rvert
\;\le\; r_n, \quad
\text{for all } A \in \mathcal{A}.\label{eq:gaussian approx}\tag{G}
\end{equation}
\end{itemize}

High-dimensional Gaussian approximation results verify this condition under suitable moment, nondegeneracy, and uniform linearization assumptions.\footnote{See \cite{CCKK-AOS,CCK-AOP} for Gaussian approximation results and \cite{CCKK-Review} for a review. For estimates based on debiased machine learning, see also \cite{belloni2018uniformly,quintas2022finite}. The finite treatment menus in Online Appendix~\ref{sec:example treatment policies} provide examples based on doubly robust welfare scores.} The bounds can depend on $p$ through logarithmic factors, allowing very large menus.\footnote{For instance, if the leading term of each $\widehat Z_\pi$ is a normalized sum of independent bounded scores, $r_n$ can be taken of order $(\log^5(pn)/n)^{1/4}$, also for condition~\eqref{eq:bootstrap approx} below \citep{CCKK-AOS}; linearization and studentization errors add further terms. In the setting~\eqref{eq:ideal}, $r_n=\delta_n=0$.}  We state (G) as a high-level condition so that the analysis does not depend on a particular estimation method.

\medskip
\noindent
\subsection{Regret Bounds}
Fix $\alpha,\beta\in(0,1/2)$, and take $n$ large enough that $r_n<\min(\alpha,\beta{,(1-\alpha)/2})$. As in Section~\ref{sec:motivation and risk-aware policy choice}, $q_{1-\beta,K}$ denotes the $(1-\beta)$-quantile of $\max_{\pi\in K}Z_\pi$, where $(Z_\pi)$ is now the Gaussian vector appearing in~\eqref{eq:gaussian approx}. Under~\eqref{eq:gaussian approx}, with probability at least \(1 - \beta - r_n\),
\(\max_{\pi \in K}\widehat{Z}_\pi \le q_{1-\beta,K}\). As in Section~\ref{sec:motivation and risk-aware policy choice} let $V_{\max}=\max_{\pi\in\Pi}V(\pi)$ and let $\Pi_0$ denote the true welfare-maximizing policies, with $\bar{\sigma}_{\Pi}$ and $\underline{\sigma}_{\Pi_0}$ as defined in \eqref{def:sigma_bars}. We write $a\lesssim b$ when $a\leq Cb$ for a constant $C$ independent of $n$ and $p$.

\begin{prop}[Regret Bounds for Risk-Aware Decisions]\label{thm:RWregret}
Under condition~\eqref{eq:gaussian approx}, the regret of any RW policy rule  \(\widehat\pi_{\text{RW}}(\widehat{k})\) is bounded with probability at least \(1 - 2\beta - 2r_n\) as:
\begin{equation}\label{Th 1 regret bound}
V_{\max} - V(\widehat\pi_{\text{RW}}(\widehat{k}))
\;\;\le\;\;
\frac{\underline{\sigma}_{\Pi_0}}{\sqrt{n}}
\Bigl\{ q_{1-\beta,\Pi_0} + \widehat{k}\Bigr\}
\;+\;
\frac{\bar{\sigma}_{\Pi}}{\sqrt{n}}
\Bigl\{ \bigl(q_{1-\beta,\Pi} - \widehat{k}\bigr)_{+}\Bigr\}.
\end{equation}
In particular, if \(0 \le \widehat{k} \lesssim q_{1-\beta,\Pi}\), then with the same probability,
\[
V_{\max} - V(\widehat\pi_{\text{RW}}(\widehat{k}))
\;\;\lesssim\;\;
\bigl(\bar{\sigma}_{\Pi}/\sqrt{n}\bigr)\,q_{1-\beta,\Pi}.
\]
Moreover, if \(q_{1-\beta,\Pi} \le \widehat{k} \lesssim q_{1-\beta,\Pi}\), then with the same probability,
\[
V_{\max} - V(\widehat\pi_{\text{RW}}(\widehat{k}))
\;\;\lesssim\;\;
\bigl(\underline{\sigma}_{\Pi_0}/\sqrt{n}\bigr)\,q_{1-\beta,\Pi}.
\]
\end{prop}

Gaussian concentration gives, for every nonempty $K\subseteq\Pi$,\footnote{See Lemma A.6 in \cite{BCCHK:18}. A short derivation is given in the proof of Lemma~\ref{lemma:quantile comparison}.}
\[
 q_{1-\beta,K}\leq\sqrt{2\log|K|}+\sqrt{2\log(1/\beta)}\leq \sqrt{2\log {p}}+\sqrt{2\log(1/\beta)}.
\]
Thus, for fixed $\beta$, $p\geq2$, and $0\leq\widehat k\lesssim q_{1-\beta,\Pi}$, Proposition~\ref{thm:RWregret} gives a regret bound of order $\bar\sigma_\Pi\sqrt{\log(p)/n}$ with probability at least $1-2\beta-2r_n$. This includes EWM. If the penalty also satisfies $\widehat k\geq q_{1-\beta,\Pi}$, the bound improves to $\underline\sigma_{\Pi_0}\sqrt{\log(p)/n}$. Similar uses of variance penalties appear in \cite{maurer2009empirical}, \cite{swaminathan2015counterfactual}, and \cite{foster2023orthogonal}.

\subsection{Reporting Guarantees}\label{Sec: PoLeCe and properties}
In practice, to obtain reporting guarantees for RW decisions, one needs to approximate \(q_{1-\alpha,\Pi}\). We do so via the bootstrap, where we set:
\[
\widehat{q}_{1-\alpha,\Pi}
\;:=\;
(1-\alpha)\text{-Quantile}
\Bigl(\max_{\pi \in \Pi} \widehat{Z}_\pi^*\Bigr),
\]
where, conditional on the data, $(\widehat Z_\pi^*)_{\pi\in\Pi}\sim N(0,\widehat C)$ and $\widehat C$ estimates $\mathrm{Cov}((Z_\pi)_{\pi\in\Pi})$. Denote by ${\Pr}^*$ the bootstrap probability conditional on the data. This construction relies on the following condition:

\begin{itemize}
\item[(B)] With probability at least \(1-\delta_n\) (where \(\delta_n \searrow 0\)),
\begin{equation}
\Bigl\lvert
{\Pr}^*\!\bigl((\widehat{Z}_\pi^*)_{\pi \in \Pi} \in A\bigr)
\;-\;
\Pr\!\bigl((Z_\pi)_{\pi \in \Pi} \in A\bigr)
\Bigr\rvert
\;\le\;
r_n,
\quad
\text{for all } A \in \mathcal{A}. \label{eq:bootstrap approx}\tag{B}
\end{equation}
\end{itemize}

Like~\eqref{eq:gaussian approx}, condition~\eqref{eq:bootstrap approx} is satisfied under mild assumptions even when \(p\) is much larger than \(n\), and has been verified for a variety of estimation methods, including debiased machine learning.\footnote{See \cite{CCKK-AOS,CCK-AOP} and \cite{CCKK-Review} for more details, along with other references on ``approximate means'' settings such as debiased ML (\cite{belloni2018uniformly,quintas2022finite}). By Gaussian comparison inequalities \citep{CCK2015Comparison,CCKK-AOS}, (B) holds if the entries of $\widehat C- \mathrm{Cov}((Z_\pi)_{\pi\in\Pi})$ are uniformly $o_P(1/\log^2 p)$.}

A key consequence, proven in Lemma~\ref{lemma:quantile comparison}, is that letting \(r_n' := 2r_n + \delta_n\) we have
\begin{equation}\label{eq:LCB}
\max_{\pi \in \Pi} \widehat{Z}_\pi
\;\;\le\;\;
\widehat{q}_{1-\alpha,\Pi}
\quad
\text{with probability at least } 1-\alpha-r_n'\text{.}
\end{equation} Inequality \eqref{eq:LCB} then directly implies a uniform lower confidence bound (LCB) on the performance of all RW policy rules.

\begin{prop}[Uniform LCB Guarantees for RW Decisions]\label{prop:report}
With probability at least \(1-\alpha-r_n'\), the true value of any policy \(\pi \in \Pi\) is bounded below as follows:
\[
V(\pi)
\;\;\ge\;\;
LV_{1-\alpha}(\pi)
\;:=\;
\widehat{V}(\pi) \;-\; \widehat{q}_{1-\alpha,\Pi}\,\widehat{s}(\pi) \text{ simultaneously for all } {\pi \in \Pi}.
\]
\end{prop}

When the standard errors are estimated, PoLeCe maximizes this bound,
\begin{equation}
\piPoLeCe
\;\;\in\;\;
\argmax_{\pi \in \Pi}
LV_{1-\alpha}(\pi),\label{PoLeC2}
\end{equation}
and reports the maximized bound $LV_{1-\alpha,\Pi}=\max_{\pi\in\Pi}LV_{1-\alpha}(\pi)$. As in Section~\ref{sec:polece ideal}, it is the RW rule whose penalty is the critical value, $\widehat k=\widehat q_{1-\alpha,\Pi}$.

\begin{prop}[Reporting Guarantees and Regret Bounds for PoLeCe]\label{thm:Polece}
With probability at least \(1-\alpha-r_n'\),
\[
V_{\max}
\;\;\ge\;\;
V(\piPoLeCe)
\;\;\ge\;\;
LV_{1-\alpha,\Pi}.
\]
Moreover, with probability at least $1-\alpha - \beta - r_n - r_n' $, the regret of \(\piPoLeCe\) satisfies
\[
V_{\max} - V(\piPoLeCe)
\;\;\le\;\;
V_{\max} - LV_{1-\alpha,\Pi}
\;\;\le\;\;
\bigl(\underline{\sigma}_{\Pi_0}/\sqrt{n}\bigr)
\Bigl\{
q_{1-\beta,\Pi_0}
\;+\;
q_{1-\alpha+r_n,\Pi}
\Bigr\}.
\]
\end{prop}

The first result in Proposition~\ref{thm:Polece} shows that $LV_{1-\alpha, \Pi}$ provides a valid high-confidence lower bound on the true welfare of the PoLeCe rule. The second result sharpens the general regret bounds presented in Proposition~\ref{thm:RWregret}. The constant $\alpha$, chosen by the researcher, determines the rule and the confidence level of $LV_{1-\alpha,\Pi}$; $\beta$ only indexes the probability with which the regret bound holds. When $\alpha \leq \beta$, the two propositions give qualitatively identical bounds, and the bound in Proposition~\ref{thm:Polece} holds with at least as high a probability. When $\alpha > \beta$, the bound in Proposition~\ref{thm:Polece} is tighter, but it may hold with lower probability; the general bound of Proposition~\ref{thm:RWregret} then remains valid.

\subsection{The Role of the Policy Class}\label{sec:policy class}
So far the menu $\Pi$ has been finite. This assumption is convenient, but it is not essential. The reporting guarantee of Proposition~\ref{prop:report} relies only on the distribution of the largest studentized estimation error, through conditions~\eqref{eq:gaussian approx} and~\eqref{eq:bootstrap approx}, and on the existence of the maximizers. None of this requires the number of policies to be finite. Online Appendix~\ref{sec:continuum} states both the reporting guarantee and the regret bounds for infinite policy classes, with~\eqref{eq:gaussian approx} and~\eqref{eq:bootstrap approx} imposed on suprema over the class. The size of the class enters only the regret bounds, and only through the critical values. For a finite menu, $q_{1-\beta,\Pi}\leq\sqrt{2\log p}+\sqrt{2\log(1/\beta)}$. For an infinite class, the role of $\sqrt{2\log p}$ is played by the expected supremum of the Gaussian process $(Z_\pi)_{\pi\in\Pi}$, which can be bounded by an entropy integral in the usual way.

The allocation problem is a particularly simple case, because welfare is linear in the $J$-vector of returns $R$. The estimation noise of the whole continuum of allocations is driven by the estimation error of the $J$-vector of estimated returns $\widehat{R}$, and $q_{1-\beta,\Pi}\leq\sqrt J+\sqrt{2\log(1/\beta)}$ however many allocations the menu contains. Following the strong approximation argument of \cite{CLR2013} for bounding functions that are linear in a finite-dimensional parameter, Online Appendix~\ref{sec:continuum-allocation} derives the guarantees for the set of all allocations from asymptotic normality of $\widehat R$ and consistency of $\widehat\Omega$. No grid over allocations is needed. More general classes, such as treatment rules that depend on continuous covariates, require more work; for these, Online Appendix~\ref{sec:continuum} verifies the conditions by approximating the class with a finite net.

\section{Implementation and Application}\label{sec: empirical examples}

\subsection{Computing PoLeCe}
PoLeCe is straightforward to compute. For the allocation problem it requires the following three steps. Further details are provided in Online Appendix~\ref{sec:bootstrap conic root}.
\begin{enumerate}
\item \emph{Tracing the frontier.} For a given penalty $k$, the risk-aware rule solves
\[
 \max_{\pi\in\Pi}\Bigl\{\pi'\widehat R-k\,\|\widehat\Omega^{1/2}\pi\|/\sqrt n\Bigr\}.
\]
With $\Pi$ given by~\eqref{Pi Fraction Allocation Problem}, the objective is concave and the constraints are linear, so this is a second-order cone program, which standard convex optimization software solves quickly. Solving it for a range of values of $k$ traces out the frontier shown in Figure~\ref{fig:frontier_MVPF}.
\item \emph{Computing the critical value.} The bootstrap draws $N^*\sim N(0,I_J)$ repeatedly. For each draw, it computes the largest standardized error
\[
 \bar{Z}^{\ast}=\max_{\pi\in\Pi}\frac{\pi'\widehat\Omega^{1/2}N^*}{\|\widehat\Omega^{1/2}\pi\|},
\]
and $\widehat q_{1-\alpha,\Pi}$ is the $(1-\alpha)$-quantile of $\bar{Z}^{\ast}$ across draws. This is not a convex problem, but it reduces to a sequence of convex problems. For any $c\geq0$, $\bar{Z}^{\ast}\geq c$ if and only if $\max_{\pi\in\Pi}\{\pi'\widehat\Omega^{1/2}N^*-c\,\|\widehat\Omega^{1/2}\pi\|\}\geq0$, so when $\bar{Z}^{\ast}\geq0$, it can be found by a root-finding search over $c\geq0$, each step of which solves a second-order cone program. If $\bar Z^\ast<0$, we record zero instead; this preserves the conditional bootstrap $(1-\alpha)$-quantile for $\alpha<1/2$ (Online Appendix~\ref{sec:bootstrap conic root}). As shown in Online Appendix~\ref{sec:full simplex}, if $\Pi$ is the whole simplex and $\widehat{\Omega}$ is diagonal, the $1-\alpha$ quantile of $\bar{Z}^{\ast}$ can be solved analytically and bootstrapping is not needed.
\item \emph{Selecting the allocation.} Solve the problem in step~1 with $k=\widehat q_{1-\alpha,\Pi}$. The solution is the PoLeCe allocation, and the optimal value is the reported lower bound $LV_{1-\alpha,\Pi}$.
\end{enumerate}
For a finite menu, as in the applications to treatment choice in Online Appendix~\ref{sec:example treatment policies}, the computations are even simpler. Step~1 amounts to comparing $p$ numbers, and step~2 to simulating the maximum of $p$ correlated Gaussian variables.

\subsection{Data and Results}
We obtain MVPF estimates $\widehat{R}$ from the Policy Impacts Library, \cite{policy_impacts_library}. There are 172 programs in total, and we focus on those in the United States whose MVPF is estimated by an RCT. We assume the reported upper  bounds correspond to the upper bound of  95\% confidence intervals and infer the standard errors for each $\widehat{R}$ from these bounds, yielding $14$ programs.\footnote{We exclude programs for which the upper and lower bounds are infinite. We recognize that some of the reported intervals are asymmetric or do not correspond exactly to confidence intervals reflecting sampling uncertainty alone, because of necessary adjustments made by the original authors. We abstract from these details and use them simply as a measure of sampling uncertainty for illustrative purposes.}
We also assume the estimation errors across programs are independent, so their covariance matrix is diagonal, with the squared standard errors on the diagonal. Note that here we only account for estimation error, holding the specification in the original papers that produced these MVPF estimates fixed.
 Alternative specifications can lead to different MVPF estimates for the same program. We impose no bounds on the shares, so that the set of allocations is the whole simplex, and we set $\alpha=0.05$ throughout.

Table~\ref{tab:example mvpf adult} reports the allocations selected by EWM and PoLeCe among programs targeting individuals older than 25. EWM places the entire budget in the program with the highest MVPF estimate, Holistic Wrap-around Services, whose estimate of 1.39 comes with a standard error of 1.11. The lower confidence bound for this allocation is $-2.14$, so the data are consistent with a very low return. PoLeCe instead places 97 percent of the budget in Medicaid for single adults (OHIE), with an estimate of 1.16 and a standard error of 0.04, and nearly all of the remainder, about 3 percent, in job training for adults (JTPA). The geometric interpretation of Section~\ref{sec:motivation and risk-aware policy choice} explains this choice. The estimate for Holistic Wrap-around Services exceeds that for OHIE by only 0.23, while its standard error is larger by 1.07. At the PoLeCe slope of about 3.2, the additional estimation risk therefore costs more than 3 units of the criterion, far more than the gain in estimated welfare. The small share of JTPA has a different explanation. Since the estimates are independent, shifting a small part of the budget from OHIE to JTPA, whose estimate is higher by 0.22, raises estimated welfare without increasing the standard error of the allocation; to first order, the standard error actually falls. The resulting lower bound for the PoLeCe allocation is 1.04, far above that for EWM.
\begin{center}
    [TABLE~\ref{tab:example mvpf adult} AROUND HERE]
\end{center}

Table~\ref{tab:example mvpf youth} reports the corresponding results for programs targeting individuals aged 25 and under. In contrast to the adult programs, PoLeCe now allocates the budget to programs whose estimates are relatively imprecise. The reason is that both programs with high estimates are imprecisely estimated: Head Start has an estimate of 1.84 with a standard error of 0.47, and the Wisconsin Scholars Grant has an estimate of 1.43 with a standard error of 0.38. The precisely estimated programs, such as Year Up, have much lower estimates, so that even after accounting for their precision, including them would lower the bound that PoLeCe attains by combining the two leading programs. EWM chooses Head Start alone, whereas PoLeCe splits the budget roughly equally between the two programs. This lowers the estimated MVPF from 1.84 to 1.64. However, it also lowers the standard error from 0.47 to about 0.30. The gain in precision outweighs the loss in estimated welfare, and the lower bound rises from 0.48 to 0.76.
\begin{center}
    [TABLE~\ref{tab:example mvpf youth} AROUND HERE]
\end{center}

\section{Conclusion}\label{sec: conclusion}

In this paper we have focused on the problem faced by a DM who wishes to select from a menu of policies to maximize welfare using imperfect sample estimates. In order to balance the estimated performance of each policy with its associated statistical uncertainty, we proposed and analyzed the properties of a class of risk-aware policies that make the precision/performance tradeoff explicit.
Such policies achieve favorable regret properties. We proposed a specific rule from the class of risk-aware policies, namely the PoLeCe rule, which uses a data-dependent construction for balancing the inherent tradeoff between estimated performance and sample uncertainty, such that a lower confidence bound on the welfare obtained by the chosen policy is automatically provided.

A large body of work synthesized in Manski (2013, 2019, 2024) has advocated for greater acknowledgement and incorporation of uncertainty in planning problems.\nocite{Manski:PubPolUncertain}\nocite{Manski:2019}\nocite{Manski:Discourse}  This paper contributes by proposing a principled way to acknowledge and incorporate statistical uncertainty into decision making. Statistical uncertainty is, however, only one of the many different types of uncertainty that may be present.\footnote{For example, \cite{Manski:2019} discusses transitory uncertainty, permanent uncertainty, and conceptual uncertainty. The statistical uncertainty considered here is one type of permanent uncertainty.} We have focused on settings in which the welfare of each policy is point identified, such that consistent estimates of the welfare and sample variance of the policies are available. Application of the concepts developed here to settings that feature so-called ``deep uncertainty'', or ambiguity, would seem an important direction for further development given the large number of settings in which the mean performance of various policies is only credibly partially identified.\footnote{The literature on treatment choice when mean treatment performance is partially identified goes back to \cite{Manski:2000}. Recent research with references to the broader literature includes \cite{Russell:20}, \cite{Ishihara/Kitagawa:21}, \cite{Yata2021}, \cite{Christensen/Moon/Schorfheide:23}, \cite{Kido:23}, \cite{Olea/Qiu/Stoye:23}, and \cite{Ben-Michael/Greiner/Imai/Jiang:25}.}
This would require balancing statistical uncertainty with a collection of interval estimates for each policy's performance.

While the nuance required to formally develop such an approach is beyond the scope of the present paper, a rough prescription could be made based on an extension of the reporting guarantee established here. To see how, consider the definition of the PoLeCe rule in \eqref{PoLeC2}. The rule selects the policy that maximizes a $1-\alpha$ lower confidence band for policy welfare. Construction of such a confidence band does not require that the optimal level of welfare or the optimal rule  be point identified, and could be implemented using techniques developed in, e.g., \cite{CLR2013} under partial identification. This is of course not the only way to balance statistical uncertainty and ambiguity. A more thorough study of the performance of such an approach, and the implicit tradeoff between statistical uncertainty and ambiguity due to partial identification would seem a useful direction for future research.

\clearpage
\appendix

\section{Proofs}\label{app:proofs}

In what follows let
$$
\kappa_{1-\alpha, K}:= (1-\alpha)\text{-Quantile} \left(\sup_{\pi \in {K}} {\widehat Z}_{\pi}\right),
\qquad
q_{1-\alpha, K}:= (1-\alpha)\text{-Quantile} \left(\sup_{\pi \in {K}} {Z}_{\pi}\right),
$$
$$
\widehat q_{1-\alpha, K}:= (1-\alpha)\text{-Quantile}^*\left(\sup_{\pi \in {K}} \widehat Z^*_{\pi}\right).
$$

\begin{lemma}[Quantile Comparison]\label{lemma:quantile comparison}With probability at least $1-\delta_n$, critical values obey the following inequalities:
$$
\kappa_{1-\alpha- 2 r_n,K } \leq
q_{1-\alpha-  r_n,K } \leq \widehat q_{1-\alpha, K}
\leq q_{1-\alpha +  r_n,K }.
$$
Furthermore,
$$
q_{1-\alpha,K }
\leq \mathrm{E} \sup_{\pi \in {K}} Z_{\pi} + \sqrt{2 \log {1/\alpha}}
\leq \sqrt{2 \log |K|} + \sqrt{2 \log {1/\alpha}}.
$$
\end{lemma}
We also note that a lower bound of the form $\sqrt{ \log |K|}$ holds, provided that the $Z_{\pi}$'s are not too correlated, using Sudakov's minoration. The upper bound here follows from standard Gaussian concentration bounds.

\subsection*{Proof of Lemma \ref{lemma:quantile comparison}} By Condition (G) we have
$$
\sup_{{K} \subset \Pi}\sup_{x \in \mathbb{R}} \left| \Pr \left(\sup_{\pi \in {K}} \widehat {Z}_{\pi} \leq x \right) - \Pr \left(\sup_{\pi \in {K}} {Z}_{\pi} \leq x \right) \right| \leq r_n.
$$
By Condition (B) we have with probability at least $1- \delta_n$:
$$
\sup_{{K} \subset \Pi}\sup_{x \in \mathbb{R}} \left| {\Pr}^*\left(\sup_{\pi \in {K}} \widehat {Z}^*_{\pi} \leq x \right) - \Pr\left(\sup_{\pi \in {K}} {Z}_{\pi} \leq x \right) \right| \leq r_n.
$$
These relations imply that w.p. $\geq 1-\delta_n$:
$$
(1 -  \alpha -2 r_n)- \text{Quantile} \left( \sup_{\pi \in {K}} \widehat {Z}_{\pi} \right) \leq (1 - \alpha-r_n)-\text{Quantile} \left(\sup_{\pi \in {K}} {Z}_{\pi}\right)
$$
$$
\leq
(1 - \alpha) \text{-Quantile}^* \left(\sup_{\pi \in {K}} \widehat Z^*_{\pi}\right)
\leq
(1 - \alpha + r_n) \text{-Quantile} \left(\sup_{\pi \in {K}} {Z}_{\pi}\right).
$$
Thus, the first claim follows.

The second claim follows by the Gaussian concentration inequality of Borell-Sudakov-Tsirelson:
\[
q_{1-\alpha, K } \leq \mathrm{E} \sup_{\pi \in {K}} Z_{\pi} + \sqrt{2 \log {1/\alpha}},
\]
and, for $|K|\geq2$ and any $a>0$, Jensen's inequality: $\exp(a\,\mathrm{E}\max_{K}Z_\pi)\leq\sum_{\pi\in K}\mathrm{E}e^{aZ_\pi}=|K|e^{a^2/2}$, so that $a=\sqrt{2\log|K|}$ gives $\mathrm{E}\max_{K}Z_\pi\leq\sqrt{2\log|K|}$, whatever the correlations. \qed

\subsection*{Proof of Proposition \ref{thm:RWregret}}
Let
$$
\Delta^+_\Pi :=
\max_{\pi \in \Pi} ( \widehat Z_{\pi} ); \quad
\Delta^-_{\Pi_0} :=
\min_{\pi \in \Pi_0} ( \widehat Z_{\pi} );
\quad \widehat \pi:= \widehat\pi_{\text{RW}}(\widehat k);
\quad
\widehat \pi_0 \in \arg \min_{\pi \in \Pi_0} \widehat s( \pi);
$$
$$
L(\widehat k) := \max_{ \pi \in \Pi} \{ \widehat V(\pi) - \widehat k \widehat s( \pi) \} = \widehat V(\widehat \pi) - \widehat k \widehat s( \widehat \pi).
$$

We then have
\begin{eqnarray*}
L(\widehat k) -V_{\max} & = & \widehat V(\widehat \pi) - \widehat k \widehat s( \widehat \pi) - V_{\max}\\
& \geq &   \widehat V(\widehat \pi_0) - \widehat k \widehat s( \widehat \pi_0) - V_{\max}   \\
& = &   V(\widehat \pi_0) +(\widehat Z_{\widehat \pi_0} - \widehat k) \widehat s(\widehat \pi_0)  - V_{\max}  \\
& = &   (\widehat Z_{\widehat \pi_0} - \widehat k) \widehat s(\widehat \pi_0)  \\
& = &   (\widehat Z_{\widehat \pi_0} - \widehat k) \underline \sigma_{\Pi_0}/\sqrt{n}   \\
& \geq &     (\Delta^-_{\Pi_0} - \widehat k) \underline \sigma_{\Pi_0}/\sqrt{n} .
    \end{eqnarray*}

Also, we have
\begin{eqnarray*}
 L(\widehat k) - V(\widehat \pi)
 & = &  \widehat V(\widehat \pi) - \widehat k \widehat s (\widehat \pi) - V(\widehat \pi)  \\
 & = & \widehat Z_{\widehat \pi} \widehat s (\widehat \pi) - \widehat k \widehat s(\widehat \pi)  \\
 & \leq &  (\widehat Z_{\widehat \pi}- \widehat k)_+ \widehat s(\widehat \pi) \\
  & \leq &  (\widehat Z_{\widehat \pi}- \widehat k)_+ \bar \sigma_\Pi/\sqrt{n} \\
& \leq & (\Delta^+_\Pi - \widehat k)_+ \bar \sigma_\Pi/\sqrt{n}.
    \end{eqnarray*}
Combining the two bounds, we have
\begin{eqnarray*}
V_{\max}- V(\widehat \pi)
& = & V_{\max} - L(\widehat k) + L(\widehat k) - V(\widehat \pi) \\
&\leq & (-\Delta^-_{\Pi_0} +\widehat k) \underline\sigma_{\Pi_0}/\sqrt{n} +(\Delta^+_\Pi - \widehat k)_+ \bar \sigma_\Pi/\sqrt{n}.
\end{eqnarray*}
 By the proof of Lemma
\ref{lemma:quantile comparison} and the symmetry of the Gaussian vector $(Z_\pi)$,
we have that
$$
\Pr (\Delta^{+}_\Pi \leq q_{1- \beta, \Pi}) \geq 1- \beta - r_n, \quad\Pr (-\Delta^{-}_{\Pi_0} \leq q_{1- \beta, \Pi_0}) \geq 1- \beta - r_n.
$$
It then follows that with probability at least $1-2\beta -2 r_n$, we have
$$
V_{\max}- V(\widehat \pi)  \leq
(q_{1-\beta,\Pi_0} +\widehat k) \underline\sigma_{\Pi_0}/\sqrt{n} +(q_{1-\beta, \Pi} - \widehat k)_+ \bar \sigma_\Pi/\sqrt{n}.$$
The two events used do not involve $\widehat k$, so the bound holds for any data-dependent $\widehat k\geq0$. If $\widehat k\leq Cq$ with $q=q_{1-\beta,\Pi}$, then $q_{1-\beta,\Pi_0}\leq q$ and $\underline\sigma_{\Pi_0}\leq\bar\sigma_\Pi$ bound the right side by $(C+2)\bar\sigma_\Pi q/\sqrt n$; if also $\widehat k\geq q$, by $(C+1)\underline\sigma_{\Pi_0}q/\sqrt n$.
\qed

\subsection*{Proof of Proposition \ref{prop:report}}
Consider the event
\begin{eqnarray*}
\mathcal{E}_\Pi &:=& \{ V(\pi) \geq \widehat V(\pi) - \widehat k \widehat s(\pi), \ \forall \pi \in \Pi\} \\
& = & \{ - \widehat Z_{\pi} \geq   - \widehat k, \ \forall \pi \in \Pi\} \\
& = & \{ \max_{\pi \in \Pi }\widehat Z_{\pi} \leq  \widehat k\}.
\end{eqnarray*}
By Lemma \ref{lemma:quantile comparison}, with our bootstrap choice $\widehat k= \widehat q_{1- \alpha, \Pi}$ used for $\piPoLeCe$, the event
$$\mathcal{K}_\Pi := \bigl\{\kappa_{1- \alpha -2r_n, \Pi} \leq\widehat k
\leq q_{1- \alpha+r_n, \Pi} \bigr\}$$ holds with probability at least $1- \delta_n$. On $\mathcal K_\Pi$, the event $\{\max_{\pi\in\Pi}\widehat Z_\pi\leq\kappa_{1-\alpha-2r_n,\Pi}\}$ implies $\mathcal E_\Pi$. The former event has probability at least $1-\alpha-2r_n$ by the definition of $\kappa$. Hence $\mathcal{E}_\Pi$ and $\mathcal K_\Pi$ hold jointly
with probability at least $1-\alpha- r_n'$ for
$$
r_n' = 2  r_n + \delta_n,
$$ so the claim follows. \qed

\subsection*{Proof of Proposition \ref{thm:Polece}}
The first claim was shown in the proof of the previous proposition.

We next show the second claim about the regret bound.  The events $\mathcal{E}_\Pi$ and $\mathcal{K}_\Pi$ from the proof of Proposition~\ref{prop:report} jointly hold with probability at least $1-\alpha-r_n^{\prime}$. Note that $\mathcal{K}_\Pi$ implies that $\widehat q_{1- \alpha, \Pi} \leq q_{1- \alpha+r_n, \Pi}$.

  Since $V_{\max} = V (\widehat \pi_0)$ for $\widehat\pi_0\in\arg\min_{\pi\in\Pi_0}\widehat s(\pi)$, as in the proof of Proposition~\ref{thm:RWregret}, under $\mathcal{E}_\Pi$ and $\mathcal{K}_\Pi$ we have
\begin{eqnarray*} V_{\max}-V(\piPoLeCe) & \leq &  V_{\max} -
\max_{\pi \in \Pi} \{\widehat V(\pi) - \widehat q_{1- \alpha, \Pi} \widehat s(\pi)\} \\
& \leq &  V(\widehat \pi_0)
- (\widehat V(\widehat \pi_0) - \widehat q_{1- \alpha, \Pi} \widehat s(\widehat \pi_0)) \\
& = & -\widehat  Z_{\widehat \pi_0} \widehat s(\widehat \pi_0) + \widehat q_{1- \alpha, \Pi} \widehat s(\widehat \pi_0)  \\
& = &
(\underline{\sigma}_{\Pi_0}/\sqrt{n}) (\widehat q_{1- \alpha, \Pi} - \widehat  Z_{\widehat \pi_0} ) \\
& \leq &
(\underline{\sigma}_{\Pi_0}/\sqrt{n}) (q_{1- \alpha+r_n, \Pi} - \Delta^-_{\Pi_0}).
\end{eqnarray*}
By the proof of Lemma~\ref{lemma:quantile comparison} we have that
$ \Pr (-\Delta^{-}_{{\Pi_0}} \leq q_{1- \beta, \Pi_0}) \geq 1- \beta - r_n$. Thus
\begin{equation*}
	V_{\max}-  V(\piPoLeCe) \leq V_{\max}- LV_{1-\alpha,\Pi} \leq (\underline{\sigma}_{\Pi_0}/\sqrt{n}) ( q_{1- \alpha+r_n, \Pi} + q_{1- \beta, \Pi_0})
\end{equation*}
with probability at least $1-\alpha- \beta - r_n - r_n'$, since we have also used that $\mathcal{E}_\Pi$ and $\mathcal{K}_\Pi$ both hold, and since
$$
LV_{1-\alpha,\Pi} = \max_{\pi \in \Pi} \{\widehat V(\pi) - \widehat q_{1- \alpha, \Pi} \widehat s(\pi)\}.
$$
\qed

\clearpage
\putbib

\clearpage
\begin{figure}[p]
    \centering
\begin{subfigure}{0.48\linewidth}
        \centering
         \caption{Programs targeting age $>$ 25}
        \includegraphics[width=\linewidth]{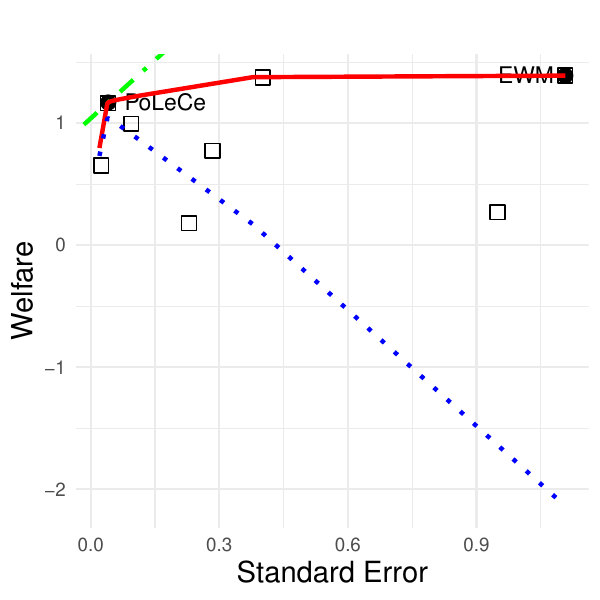}
    \end{subfigure}
    \hfill
    \begin{subfigure}{0.48\linewidth}
        \centering
        \caption{Programs targeting age $\leq$ 25}
        \includegraphics[width=\linewidth]{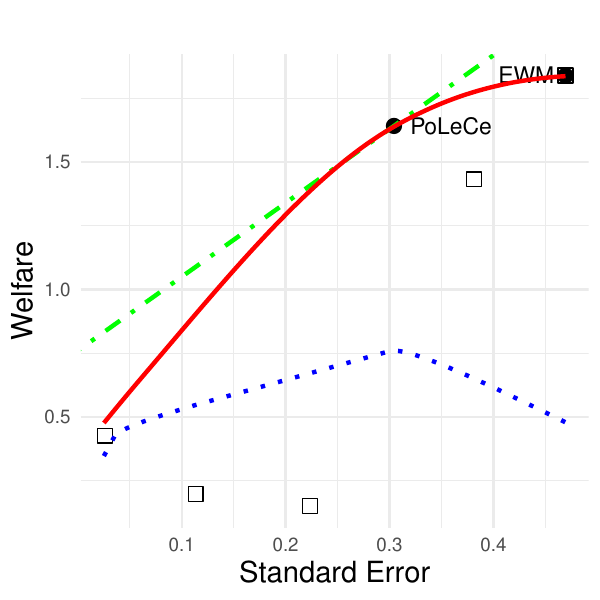}
    \end{subfigure}
\caption{\label{fig:frontier_MVPF} Efficient Decision Frontier.}
\begin{minipage}{1\textwidth}
{\small \emph{Notes}: Squares represent MVPF estimates and their standard errors for 14 programs collected by the Policy Impacts Library. The red solid curve shows the decision frontier for budget allocation across these programs, and the blue dotted curve gives the corresponding lower confidence bound (LCB). The solid square and circle mark allocations selected by EWM and PoLeCe, respectively. EWM selects the allocation with the highest welfare estimate. PoLeCe selects the allocation on the frontier with the highest LCB. Equivalently, PoLeCe is the point at which the green dashed line, whose slope equals the critical value used in constructing the LCB, is tangent to the decision frontier. The intercept of this line at zero standard error is the reported lower bound $LV_{0.95,\Pi}$. For programs targeting age $>$ 25 in Panel A, the programs included in the PoLeCe allocation with shares $\geq 1\%$, listed in order of their shares, are:  OHIE, Single Adults: 1.16 (0.04), Job Training Partnership Act, Adults: 1.38 (0.40), where each value represents the point estimate (standard error). The program selected by EWM is Holistic Wrap-around Services Can Improve Employment Rates: 1.39 (1.11).  For programs targeting age $\leq$ 25 in Panel B, the programs included in the PoLeCe allocation are:  Head Start: 1.84 (0.47) and Wisconsin Scholars Grant: 1.43 (0.38). The program selected by EWM is Head Start: 1.84 (0.47). Job Training Partnership Act, Youth: -0.23 (0.81) is omitted to improve the scale of the plot for the more precisely estimated programs.
\par}
\end{minipage}
\end{figure}

\clearpage
\begin{table}[p]{\small\caption{Results for Investment Allocations to Adult Programs (Age $>$ 25)}
\label{tab:example mvpf adult}
\begin{center}
\begin{tabularx}{\textwidth}{>{\raggedright\arraybackslash}Xrr}
\hline\hline
 & EWM & PoLeCe \\
 \hline
 \multicolumn{3}{l}{\textbf{Panel A. Outcome: MVPF}} \\
Estimated Value: $\widehat{V} (\widehat{\pi})$ & 1.39  &  1.17   \\
LCB: $ \widehat{V}(\widehat{\pi})-\widehat{q}_{0.95, \Pi}\,\widehat{s}(\widehat{\pi})$ & -2.14 & 1.04 \\
\hline
 \multicolumn{3}{l}{\textbf{Panel B. Selected Policy}} \\
Holistic Wrap-around Services Can Improve Employment Rates: 1.39 (1.11) & 1 & 0.004 \\
Padua Pilot: 0.18 (0.23) & 0 & 0  \\
Housing Vouchers in Chicago: 0.65 (0.02) & 0 & 0  \\
Oregon Health Insurance Experiment (Provided to Single Adults): 1.16 (0.04) & 0 & 0.969 \\
Income Maintenance Experiment in Seattle and Denver: 0.27 (0.95) & 0 & 0  \\
Paycheck Plus: EITC to Adults without Dependents: 1.00 (0.09) & 0 & 0  \\
Work Advance: 0.78 (0.28) & 0 & 0 \\
Job Training Partnership Act, Adults: 1.38 (0.40) & 0 & 0.027 \\
\hline\hline
\end{tabularx}
\end{center}
\begin{minipage}{1\textwidth}
\small
\textit{Note:} EWM denotes empirical welfare maximization, and PoLeCe denotes policy learning with confidence. In Panel A, we report the EWM and the PoLeCe lower confidence bound with $\alpha = 0.05$. In Panel B, we report the selected allocation for each program.  MVPF estimates and associated standard errors (in parentheses) are included  alongside each program name.
\end{minipage}
}
\end{table}

\clearpage
\begin{table}[p]
{\small
\caption{Results for Investment Allocations to Youth Programs (Age $\leq 25$)}
\label{tab:example mvpf youth}
\begin{center}
\begin{tabularx}{\textwidth}{>{\raggedright\arraybackslash}Xrr}
\hline\hline
 & EWM & PoLeCe \\
 \hline
 \multicolumn{3}{l}{\textbf{Panel A. Outcome: MVPF}} \\
Estimated Value: $\widehat{V} (\widehat{\pi})$ & 1.84  &  1.64   \\
LCB: $\widehat{V}(\widehat{\pi})-\widehat{q}_{0.95, \Pi} \widehat{s}(\widehat{\pi})$ & 0.48 & 0.76 \\
\hline
 \multicolumn{3}{l}{\textbf{Panel B. Selected Policy}} \\
Head Start Impact Study: 1.84 (0.47)  & 1 & 0.514 \\
Wisconsin Scholars Grant to Low-Income College Students: 1.43 (0.38)  & 0 & 0.486  \\
JobStart: 0.20 (0.11)   & 0 & 0 \\
Year Up: 0.43 (0.03)  & 0 & 0 \\
Job Corps: 0.15 (0.22)   & 0 & 0 \\
Job Training Partnership Act, Youth: -0.23 (0.81)  & 0 & 0\\
\hline\hline
\end{tabularx}
\end{center}
\begin{minipage}{1\textwidth}
\small
\textit{Note:} EWM denotes empirical welfare maximization, and PoLeCe denotes policy learning with confidence. In Panel A, we report the EWM and the PoLeCe lower confidence bound with $\alpha = 0.05$. In Panel B, we report the selected allocation for each program. MVPF estimates and associated standard errors (in parentheses) are included  alongside each program name.
\end{minipage}
}
\end{table}
\end{bibunit}

\clearpage
\setcounter{page}{1}
\markboth{ONLINE APPENDIX}{ONLINE APPENDIX}
\begin{center}
{\large\scshape Online Appendix to ``Policy Learning with Confidence''}
\end{center}
\medskip
% The Online Appendix has its own lettered sections, equation numbers, tables, and figures.
\setcounter{section}{0}\setcounter{equation}{0}\setcounter{table}{0}\setcounter{figure}{0}\setcounter{footnote}{0}
\renewcommand{\thesection}{\Alph{section}}
\numberwithin{equation}{section}
\renewcommand{\theequation}{\thesection.\arabic{equation}}
\renewcommand{\thefigure}{OA\arabic{figure}}
\renewcommand{\thetable}{OA\arabic{table}}
% Distinct hyperlink anchors, since the main-text appendix is also lettered.
\renewcommand{\theHsection}{OA.\Alph{section}}
\renewcommand{\theHequation}{\theHsection.\arabic{equation}}
\renewcommand{\theHfigure}{OA\arabic{figure}}
\renewcommand{\theHtable}{OA\arabic{table}}
\renewcommand{\theHaprop}{\theHsection.\arabic{aprop}}
% Citation and bibliography anchors of the Online Appendix differ from those of the main text,
% so that each citation links to the reference list of its own part.
\makeatletter
\def\hyper@natlinkstart#1{\Hy@backout{#1}\hyper@linkstart{cite}{cite.OA.#1}\def\hyper@nat@current{#1}}
\def\hyper@natlinkbreak#1#2{\hyper@linkend#1\hyper@linkstart{cite}{cite.OA.#2}}
\def\hyper@natanchorstart#1{\Hy@raisedlink{\hyper@anchorstart{cite.OA.#1}}}
\makeatother

\begin{bibunit}
\noindent The Online Appendix is organized as follows. Online Appendix~\ref{sec:literature} reviews the related literature in more detail. Online Appendix~\ref{sec:bootstrap conic root} explains how the Gaussian bootstrap critical value is computed for the allocation problem and characterizes it analytically when the menu is the full simplex and $\widehat\Omega$ is diagonal. Online Appendix~\ref{sec: RW Decision} shows that risk-aware rules arise from two alternative decision problems, with proofs in Online Appendix~\ref{sec: Supplement Proofs}. Online Appendix~\ref{sec:frontier algorithm} gives an algorithm for computing the efficient decision frontier. Online Appendix~\ref{sec:example treatment policies} describes the treatment-choice problem and applies PoLeCe to data from three randomized experiments, and Online Appendix~\ref{sec: computational experiments} compares PoLeCe with EWM in simulations calibrated to these applications. Online Appendix~\ref{sec:continuum} extends the guarantees to continuum policy classes, and Online Appendix~\ref{sec:screening} develops a two-step rule that screens out clearly inferior policies.

\noindent References to sections, equations, and propositions without a letter prefix are to the main paper, as are conditions (G) and (B).

\section{Related Literature}\label{sec:literature}

This appendix expands on the short review in the introduction.

\paragraph{Statistical decision theory and policy learning.} Prescriptions for decision making with sample data include conditional Bayes rules, maximin rules, and minimax regret rules. These are summarized in \cite{Manski2021}, which advocates the use of statistical decision theory from the frequentist perspective of \cite{Wald:50}. Minimax regret rules in particular have recently received renewed attention, starting with the pioneering work of \cite{Manski2004}.  Follow-up work includes a decision theoretic framework introduced by \cite{dehejia_program_2005}, finite-sample bounds considered by \cite{stoye_minimax_2009}, and an asymptotic framework introduced by \citet{hirano_asymptotics_2009,Hirano/Porter:Handbook}. While the use of statistical decision theory for selecting policy rules summarized in \citet{Manski2021,Manski:2023} is conceptually appealing, computational challenges remain, particularly in settings involving samples from multiple data sources as in our investment-allocation application.  Recent papers on  policy learning provide approaches for treatment rule estimates with favorable asymptotic regret properties. \citet{KT2018} prove the optimality of EWM rules in the sense that as the sample size increases expected regret converges to zero at the minimax rate. \cite{mbakop_model_2021} propose a penalized welfare maximization rule that penalizes the complexity of policy classes, and establish an oracle inequality for regret. \citet{AW2021} and \citet{DSLC2019} propose doubly-robust estimation of average welfare, which leads to an optimal rule even with quasi-experimental data. The rule proposed here provides a practical approach to providing a decision rule with favorable regret properties while also yielding a lower confidence bound for welfare. The nonprofit GiveWell has also highlighted the importance of communicating uncertainty in cost-effectiveness estimates for the sake of ``transparency'' \citep{salisbury_how_2024}.

\paragraph{Confidence sets and certified decisions.} PoLeCe maximizes $LV_{1-\alpha}(\pi)$, a $1-\alpha$ lower confidence band for $V(\cdot)$ constructed as in \cite{CLR2013}. The coverage guarantee is exact in the idealized Gaussian setting of Section~\ref{sec:motivation and risk-aware policy choice} and holds up to the approximation error in Proposition~\ref{thm:Polece} in the general setting. Because PoLeCe maximizes worst-case welfare over the confidence set defined by this band, it falls within the class of rules \cite{Manski2021} calls ``as-if optimization with set estimates''. Section 3.2 of \cite{Manski2021} suggested that one might consider decision making using confidence sets, but left this to future research.

The contemporaneously developed working paper \cite{andrews2025certifieddecisions} introduces the term \textit{P-certified decisions} for the class of rules that ensure an upper bound on loss -- equivalently a lower bound on welfare -- with probability at least $1 - \alpha$. \cite{andrews2025certifieddecisions} argue that such decisions are useful for ambiguity-averse decision makers and show that ``as-if'' decisions using confidence sets form an essentially complete class of P-certified decisions. A similar finding is made in \cite{Kiyani/Pappas/Roth/Hassani:25}, which studies rules that minimize worst-case loss over conformal prediction sets. Within this class, an optimal prediction set can be constructed from outcomes whose utility under a chosen action is at least a threshold determined by conditional quantiles. Use of these sets for as-if optimization with set estimates differs from the use of lower confidence bands for welfare that corresponds to the PoLeCe rule. \cite{Ben-Michael/Greiner/Imai/Jiang:25} also propose a P-certified decision rule for criminal release recommendations. Their rule again differs from ours, as it uses two-sided confidence bands for conditional outcome means to seek safe improvement over an existing policy under partial identification.\footnote{The setting studied by \cite{Ben-Michael/Greiner/Imai/Jiang:25} has other distinguishing features, for example using data based on deterministic functions of pre-trial risk assessment scores, and the need to take on the challenge of partially identified policy welfare.} While the PoLeCe rule falls within the class of P-certified decisions studied in \cite{andrews2025certifieddecisions}, its P-certificate is based on confidence bands of the type studied in \cite{CLR2013} and it thus differs from the other policy rules studied within this class. A convenient feature of the PoLeCe rule is that it is delivered directly by maximizing $LV_{1-\alpha}(\pi)$ in a single stage without requiring explicitly computing a maximin rule over a set estimate. Furthermore, it provides a rule that lies on the efficient decision frontier with the balance between estimated performance and sample variation pinned down by the critical value corresponding to the planner's desired confidence level.

P-certified decisions are also connected to the recent literature on inference on winners, e.g. \cite{BenjaminiSelected}, \cite{andrews2024inference}, \citet{Zrnic/Fithian:24,Zrnic/Fithian:25}. These approaches provide confidence sets for the welfare of the selected choice that has the highest empirical welfare in-sample. Inference approaches for optimal assignment and/or welfare in the population include those in the online supplement of \cite{KT2018}, \cite{Rai2019}, \cite{Armstrong/Shen:2023}, and \cite{ponomarev2024lowerconfidencebandoptimal}. Proposition~\ref{thm:Polece} of the paper shows that the PoLeCe rule provides a coverage guarantee for both the welfare of the selected policy and the optimal welfare $V_{\max}$.

Recent alternatives that account for sampling uncertainty in policy selection include \cite{Sun2021} and \cite{Moon:25}. \cite{Sun2021} considers policy selection when the decision maker faces a budget constraint and there is uncertainty regarding the cost of the policy. \cite{Moon:25} proposes an empirical Bayes approach to deal with statistical uncertainty in the welfare of different policies. The approach in this paper differs from each of these.

\section{Gaussian Bootstrap for Investment Allocation Problem}\label{sec:bootstrap conic root}

The Gaussian bootstrap is a resampling technique widely used in statistical inference to approximate the distribution of a statistic. A critical step in this procedure involves computing the statistic
\[
\bar{Z}^* = \max_{\pi \in \Pi} Z^*_{\pi},
\]
where
\[
Z^*_{\pi} = \frac{\pi^\top Z^*}{\sqrt{\pi^\top \widehat{\Omega}\,\pi}},
\]
and
\[
Z^* = \widehat{\Omega}^{1/2}\,N^*(0,I).
\]
Here, $N^*(0,I)$ denotes a standard Gaussian vector in $\mathbb{R}^{J}$, $\widehat{\Omega}$ is a positive-definite covariance estimator, and $\Pi \subset \mathbb{R}^{J}$ is a nonempty polytope contained in the probability simplex.

At first glance, the computation of $\bar{Z}^*$ appears to require solving a nonconcave fractional optimization problem:
\[
\max_{\pi \in \Pi} \frac{\pi^\top Z^*}{\sqrt{\pi^\top \widehat{\Omega}\,\pi}},
\]
which can be computationally prohibitive. However, by reformulating the problem as a root-finding task involving second-order conic programs, we can improve computational efficiency.

\subsection*{Reformulation as a Root-Finding Problem}

Instead of directly tackling the fractional optimization, we consider the following equivalence when $\bar Z^*\geq0$:
\[
\bar{Z}^* = \mathrm{root}_{t \ge 0} \left\{ \max_{\pi \in \Pi} \left( \pi^\top Z^* - t\,\sqrt{\pi^\top \widehat{\Omega}\,\pi} \right) = 0 \right\}.
\]
This reformulation allows us to solve for $\bar{Z}^*$ by identifying the unique value of $t$ that zeroes the function
\[
f(t) = \max_{\pi \in \Pi} \left( \pi^\top z - t\,\sqrt{\pi^\top \widehat{\Omega}\,\pi} \right)
\]
when $Z^* = z$.
If $f(0)<0$, replace the negative value of $\bar Z^*$ by zero.
For $0<\alpha<1/2$, this preserves its conditional bootstrap $(1-\alpha)$-quantile:
the quantile is positive because $\bar Z^*\geq Z_\pi^*\sim N(0,1)$
for any fixed $\pi\in\Pi$.
Lemma~\ref{sec:lem conic} below formalizes the properties of $f(t)$ necessary for this approach. The validity of the resulting critical value for the whole continuum of allocations is established in Online Appendix~\ref{sec:continuum-allocation}.

\subsection{Details for Gaussian Bootstrap Implementation}\label{sec: Gaussian Bootstrap Details}

The following Lemma provides formalities for implementing the Gaussian bootstrap procedure described above.
\begin{alemma}[Existence, Uniqueness, and Differentiability of the Root]\label{sec:lem conic}
    Let $\Pi \subset \mathbb{R}^J$ be a nonempty polytope contained in the probability simplex, and let $\widehat{\Omega}$ be a positive-definite matrix. Define
    \[
    f(t) = \max_{\pi \in \Pi} \left( \pi^\top z - t\,\sqrt{\pi^\top \widehat{\Omega}\,\pi} \right) \quad \text{for } t \ge 0.
    \]
    Then:
    \begin{enumerate}
        \item For each fixed $t \ge 0$, the optimization problem
        \[
        \max_{\pi \in \Pi} \left( \pi^\top z - t\,\sqrt{\pi^\top \widehat{\Omega}\,\pi} \right)
        \]
        can be reformulated as a concave second-order conic optimization problem.

        \item The function $f$ is continuous and strictly decreasing, with
        $f(t)\to-\infty$ as $t\to\infty$.
        There exists a unique $t^*\geq0$ such that $f(t^*)=0$
        if and only if $f(0)\geq0$.

        \item For $t>0$, the optimizer is unique and the derivative of $f$ at $t$ is given by
        \[
        f'(t) = -\,\sqrt{\,\pi^*(t)^\top \widehat{\Omega}\,\pi^*(t)\,},
        \]
        where $\pi^*(t) \in \arg\max_{\pi \in \Pi} \left( \pi^\top z - t\,\sqrt{\pi^\top \widehat{\Omega}\,\pi} \right)$.
    \end{enumerate}
\end{alemma}

\subsection*{Implications for Gaussian Bootstrap}
The reformulation of computing $\bar{Z}^*$ as a root-finding problem over a family of second-order cone (SOC) programs significantly enhances computational efficiency for draws with $f(0)\geq0$. Each evaluation of $f(t)$ involves solving a concave SOC optimization problem, which is well-supported by modern optimization solvers. The continuity and strict decrease of $f$ ensure that bisection on an interval bracketing $t^*$ converges to the unique root. Moreover, the explicit expression for the derivative $f'(t)$ facilitates the use of derivative-based optimization methods, further streamlining the computational process.

\begin{proof}[Proof of Lemma~\ref{sec:lem conic}]
    \textbf{1. Second-Order Conic Reformulation.}
    Fix any $t\geq0$. Introduce an auxiliary variable $s$ subject to
    the cone constraint below. The optimization problem becomes
    \[
    \begin{aligned}
    \max_{\pi, s} \quad & \pi^\top z + t\,s \\
    \text{subject to} \quad & \|\widehat{\Omega}^{1/2}\,\pi\|_2 \le -s, \\
    & s \le 0, \\
    & \pi \in \Pi.
    \end{aligned}
    \]
    For each $\pi\in\Pi$ and every $t\geq0$, choosing
    $s=-\sqrt{\pi^\top\widehat\Omega\pi}$ maximizes the objective over feasible
    values of $s$. Thus the optimal value is $f(t)$.
    The norm constraint is a standard second-order cone constraint.
    Since $\Pi$ is a polytope, its constraints are linear.
    Therefore, the problem is a concave second-order conic program.

    \textbf{2. Existence and Uniqueness of the Root.}
    \textbf{Existence.}
    Observe that
    \[
    f(0) = \max_{\pi \in \Pi} \pi^\top z,
    \]
    which is finite since $\Pi$ is bounded and closed. As $t\to\infty$, the term
    $-t\,\sqrt{\pi^\top\widehat\Omega\pi}$ dominates uniformly over
    $\pi\in\Pi$, since $\min_{\pi\in\Pi}\sqrt{\pi^\top\widehat\Omega\pi}>0$,
    leading to $f(t)\to-\infty$. Since $f(t)$ is continuous (as the pointwise
    maximum of continuous functions over a compact set), the Intermediate Value
    Theorem guarantees the existence of some $t^*\geq0$ such that $f(t^*)=0$
    provided $f(0)\geq0$.

    \textbf{Uniqueness.}
    From part 3, for $t>0$ we have
    \[
    f'(t) = -\,\sqrt{\pi^*(t)^\top\widehat\Omega\pi^*(t)} < 0,
    \]
    which, together with continuity at zero, implies that $f(t)$ is
    strictly decreasing. A strictly decreasing continuous function can cross
    zero at most once. Therefore, the root $t^*$ satisfying $f(t^*)=0$ is unique.
    If $f(0)<0$, strict decrease excludes a nonnegative root.

    \textbf{3. Derivative of $f(t)$.}
    Define
    \[
    F(\pi,t) = \pi^\top z - t\,\sqrt{\pi^\top\widehat\Omega\pi}.
    \]
    Then
    \[
    f(t) = \max_{\pi\in\Pi} F(\pi,t).
    \]
    For $t>0$, $F(\cdot,t)$ is strictly concave on $\Pi$: positive
    definiteness and the simplex restriction make
    $\|\widehat\Omega^{1/2}\pi\|_2$ strictly convex. Hence the maximizer is unique.
    By the envelope theorem, the derivative of $f(t)$ with respect to $t$
    for $t>0$ is
    \[
    f'(t) = \frac{\partial}{\partial t}F(\pi^*(t),t)
    = -\,\sqrt{\pi^*(t)^\top\widehat\Omega\pi^*(t)}.
    \]
    Since $\widehat\Omega$ is positive definite and $\pi^*(t)\in\Pi$
    is nonzero, it follows that
    $\pi^*(t)^\top\widehat\Omega\pi^*(t)>0$. Therefore, $f'(t)<0$,
    establishing that $f(t)$ is strictly decreasing.
\end{proof}

\subsection{The case of the full simplex}\label{sec:full simplex}
When $\Pi$ is the simplex and $\widehat{\Omega}$ is diagonal the critical value can be computed directly without optimization or the bootstrap. Since $Z^*_\pi$ does not change when $\pi$ is multiplied by a positive constant, the maximum over the simplex equals the maximum over all nonzero $\pi$ with nonnegative entries. That is
\[
\bar{Z}^* = \max_{\pi \in \Pi}  \frac{\pi^\top Z^*}{\sqrt{\pi^\top \widehat{\Omega}\,\pi}} = \max_{\pi \in \mathbb R^J_+, \pi \neq 0}  \frac{\pi^\top Z^*}{\lVert\pi^\top \widehat{\Omega}^{1/2}\ \rVert},
\]
where $\mathbb{R}^J_+$ denotes the nonnegative orthant in $\mathbb{R}^J$. Using $Z^{\ast} = \widehat{\Omega}^{1/2}\,N^*(0,I)$ and defining the closed convex cone $\mathcal{K}=\widehat\Omega^{1/2}\mathbb R^J_+$ we then have
\begin{equation*}
    \bar{Z}^* = \max_{w \in \mathcal{K}, w\neq 0}  \frac{w^{\top}N^*(0,I)}{\lVert w \rVert}\text{.}
\end{equation*}
We now focus on characterizing the $1-\alpha$ quantiles for $\max(0,{\bar Z^*})$, which, as noted previously, are equivalent to those of ${\bar Z^*}$ for $\alpha < 1/2$. We abbreviate $N^*(0,I)$ with $N^*$.

We first show that
\begin{equation}\label{eq:positive part}
 \max(\bar Z^*,0)=\|P_{\mathcal K}N^*\|,
\end{equation}
where $P_{\mathcal K}$ is the Euclidean projection onto $\mathcal K$. By Moreau's decomposition, $N^*=P_{\mathcal K}N^*+(N^*-P_{\mathcal K}N^*)$, where the two terms are orthogonal and the second lies in the polar cone $\mathcal K^\circ=\{y:w^\top y\leq0\text{ for all }w\in\mathcal K\}$. Hence, for every $w\in\mathcal K$,
\[
 w^\top N^*\leq w^\top P_{\mathcal K}N^*\leq\|w\|\,\|P_{\mathcal K}N^*\|,
\]
so that $\bar Z^*\leq\|P_{\mathcal K}N^*\|$. This further implies that for every $w\in\mathcal K$, $w \neq 0$,
\[
 \frac{w^\top N^*}{\Vert w \rVert }\leq \|P_{\mathcal K}N^*\|.
\]
If $P_{\mathcal K}N^*\neq0$, the choice $w=P_{\mathcal K}N^*$ gives $w^\top N^*=\|P_{\mathcal K}N^*\|^2$, so the bound is attained and $\bar Z^*=\|P_{\mathcal K}N^*\|>0$. If $P_{\mathcal K}N^*=0$, then $N^*\in\mathcal K^\circ$, so $w^\top N^*\leq0$ for every $w\in\mathcal K$, and $\bar Z^*\leq0$. In both cases~\eqref{eq:positive part} holds.

If $\widehat\Omega$ is diagonal, $\mathcal K$ is the nonnegative orthant and $P_{\mathcal K}N^*$ has coordinates $\max(N^*_j,0)$ so that
\[
 \max(\bar Z^*,0)=\Bigl\{\sum_{j=1}^J\max(N^*_j,0)^2\Bigr\}^{1/2}=\|N^*_+\|,
\]
and its square has the chi-bar-squared distribution $\sum_{k=0}^J\binom Jk2^{-J}\chi^2_k$, with $\chi^2_0$ the point mass at zero. For $1-\alpha>1/2$, the critical value $\widehat q_{1-\alpha,\Pi}$ is therefore the solution $t>0$ of
\[
 \sum_{k=0}^J\binom Jk2^{-J}\Pr\bigl(\chi^2_k\leq t^2\bigr)=1-\alpha\text{,}
\]
which depends only on the number of programs $J$ and increases with $J$. It does not depend on the standard errors. The same argument gives the two-sided maximum used in Online Appendix~\ref{sec:screening}. For each $w$, either $w^\top N^*\geq0$ or $-w^\top N^*\geq0$, so the larger of the suprema for $N^*$ and $-N^*$ is nonnegative. Hence $\max_\pi|Z^*_\pi|=\max(\|N^*_+\|,\|N^*_-\|)$, where $N^*_-$ has coordinates $\max(-N^*_j,0)$.

\section{Alternative Motivations for Risk-Aware Decisions}\label{sec: RW Decision}

We consider a decision maker (DM) who holds posterior beliefs \(B\) about policy welfare \(\{V(\pi)\}\) given by
\begin{equation}\label{eq:belief}
V(\pi) \;\sim\; \widehat{V}(\pi) \;+\; \widehat{s}(\pi)\,Z_\pi,
\end{equation}
where \(\{Z_\pi\}\) is a centered Gaussian vector with unit variance \(\bigl(\mathrm{E}_B[Z_\pi^2] = 1\bigr)\).\footnote{This is motivated by the approximate Bayesian framework of \cite{doksumlo1990}, in which the DM treats empirical estimates as the input data and uses approximate likelihoods of these estimators to form beliefs about parameters of interest.}
Note that \(\widehat{V}(\pi)\) and \(\widehat{s}(\pi)\) are fixed parameters under the DM's posterior \(B\). They may arise from empirical estimates, incorporate prior information, or reflect other model uncertainties.

The following result shows that under the conditions stated below, this framework encompasses all risk-aware, rational policy choices of a DM whose risk preferences follow the von Neumann-Morgenstern expected utility model. Specifically, suppose the DM has a strictly increasing and concave utility function \(U\) and aims to maximize
\[
\max_{\pi \in \Pi}\,\mathrm{E}_B\bigl[ U\bigl(V(\pi)\bigr)\bigr].
\]

\begin{aprop}[Posterior Expected Utility Maximization]\label{prop:DM1}
Consider the DM described above. Assume \(\Pi\) is compact and convex, \(\widehat{V}\) is continuous and concave on \(\Pi\), and \(\widehat{s}\) is continuous, convex, and positive on \(\Pi\). Assume also that \(U\) is continuously differentiable and that \(|U(x)|+|U'(x)|\leq Ae^{{b}|x|}\) for some constants $A>0$ and \({b>0}\).
\begin{enumerate}
\item If \(U\) is strictly concave, then every maximizer of posterior expected utility solves
\begin{equation}\label{eq:utilityDM}
\max_{\pi \in \Pi}\,\Bigl\{\widehat{V}(\pi)\;-\;k_{UB}\,\widehat{s}(\pi)\Bigr\}
\end{equation}
for some constant \(k_{UB}>0\) that depends on \(U\) and \(B\).
\item If \(U\) is linear (the DM is risk-neutral), then the DM solves the above program with \(k_{UB} = 0\). Furthermore, any limit point of the solutions to \eqref{eq:utilityDM} as \(k_{UB} \searrow 0\) lies within the solution set for \(k_{UB} = 0\).
\end{enumerate}
\end{aprop}

The geometric conditions hold for a compact convex menu of allocations, where $\widehat V(\pi)=\pi'\widehat R$ is linear and $\widehat s(\pi)=\|\widehat\Omega^{1/2}\pi\|/\sqrt n$ is convex; with $\widehat\Omega$ positive definite, $\widehat s$ is strictly convex on the simplex, so \eqref{eq:utilityDM} then has a unique solution for $k_{UB}>0$. Convexity cannot be dropped. On a finite menu, the expected-utility choice is efficient, but it need not maximize any criterion of the form \eqref{eq:utilityDM}. For example, take three policies with $(\widehat V,\widehat s)$ equal to $(0,1)$, $(1,2)$, and $(5/2,3)$, and $U(x)=1-e^{-5x/8}$. The middle policy uniquely maximizes expected utility, yet it maximizes $\widehat V-k\widehat s$ only if $k\leq1$ and $k\geq3/2$.

Thus, under the conditions of Proposition~\ref{prop:DM1}, a risk-averse DM must solve a variant of \eqref{eq:utilityDM}. As an example, for exponential utility \(U(v) = 1 - \exp(-\lambda v)\), one obtains
\[
\mathrm{E}[U(V(\pi))]
\;=\;
1
\;-\;
\exp\Bigl(
    -\lambda\,\widehat{V}(\pi)
    \;+\;
    \tfrac12\,\lambda^2\,\widehat{s}^2(\pi)
\Bigr).
\]
Maximizing this over \(\pi \in \Pi\) is equivalent to maximizing the mean--variance criterion $\widehat V(\pi)-\lambda\widehat s^2(\pi)/2$, and under the conditions of Proposition~\ref{prop:DM1} its solution solves \eqref{eq:utilityDM} with \(k_{UB}=\lambda\widehat s(\widehat\pi)>0\), where $\widehat\pi$ is the optimal choice. In contrast, setting \(k_{UB} = 0\) recovers the risk-neutral (EWM) decision. Notably, as \(k_{UB}\) decreases to zero, every limit point of solutions to the risk-weighted program \eqref{eq:utilityDM} is a lowest-risk EWM policy.

We can gain another perspective by considering a DM with Gaussian beliefs as described above in \eqref{eq:belief}, but whose risk preferences are more directly shaped by a combination of regret and estimation risk, as suggested by \cite{liese_statistical_2008}.\footnote{The suggestion was made to unify the problems of selection and estimation.}

Consider the loss function
$$
L(V(\pi)) = V_{\max} - V(\pi) + k_R \bigl|\widehat V(\pi) - V(\pi)\bigr| \sqrt{\bm{\pi}/2},
$$
where $k_R\geq0$ and boldface $\bm{\pi}$ denotes the area of the unit circle. This loss function has regret component
$$
V_{\max} - V(\pi) = \max_{\bar \pi \in \Pi} V(\bar \pi) - V(\pi)
$$
stemming from using the policy rule \(\pi\) rather than the unknown welfare-maximizing choice, and the estimation risk component
\(\bigl|\widehat V(\pi) - V(\pi)\bigr|\).
Consider a DM who minimizes the expected regret risk-aware loss function:
$$
\min_{\pi \in \Pi}\E_B \bigl[L(V(\pi))\bigr].
$$

\begin{aprop}[Best Posterior Regret-Risk Aware Decision]\label{prop:DM2}
Assume \(\widehat{V}(\pi)\) and \(\widehat{s}(\pi)\) are continuous in \(\pi\), that \(\Pi\) is compact, and that $\E_B|V_{\max}|<\infty$. The DM described above finds the optimal policy by solving:
\[
\max_{\pi \in \Pi}\,\Bigl\{\,\widehat{V}(\pi)\;-\;k_R\,\widehat{s}(\pi)\Bigr\}.
\]
\end{aprop}

The integrability condition holds for a finite menu and in the allocation problem with Gaussian beliefs about $R$, since then $|V_{\max}|\leq\|R\|$.

Propositions \ref{prop:DM1} and \ref{prop:DM2} explain why DMs may prefer a strictly positive \(k\), but the exact value of \(k\) is application-specific. One way to set it is by conducting experiments where DMs respond to hypothetical or incentivized lotteries/choices, then calibrating $k$ based on their observed choices. Another approach is to present a range of allocations generated by risk-aware rules, have DMs select their preferred allocations, and infer values of $k$ consistent with those selections. Generally, there is no single universal way to determine the ``right'' $k$; its choice depends on context and stakeholder preferences. Using a critical value for $k$ as in PoLeCe provides both interpretable regret and reporting guarantees.

\section{Proofs of Propositions \ref{prop:DM1} and \ref{prop:DM2}}\label{sec: Supplement Proofs}

\subsection*{Proof of Proposition \ref{prop:DM1}}

Claim (1). Let \(v = \widehat{V}(\pi)\) and \(s = \widehat{s}(\pi)\). Then
\[
\mathrm{E}_B\bigl[U\bigl(V(\pi)\bigr)\bigr]
\;=\;
\mathrm{E}\bigl[U(v + sZ)\bigr],
\quad
Z \sim N(0,1).
\]
Define
$
f(v,s) :=\mathrm{E}\bigl[U(v + sZ)\bigr].
$
The growth condition on $U$ and $U'$ justifies differentiation under the expectation, so $f$ is continuously differentiable. The function \((v,s)\mapsto f(v,s)\) is \emph{increasing} in \(v\) and \emph{decreasing} in \(s\). Indeed,
\[
\frac{\partial}{\partial v} f(v,s) \;=\; \mathrm{E}[U'(v+sZ)] \;>\; 0,
\quad
\frac{\partial}{\partial s} f(v,s)
\;=\;
\mathrm{E}\bigl[U'(v + sZ)\,Z\bigr]
\;<\;0 \quad\text{for } s>0,
\]
because \(U'(\cdot)>0\) and \(-U'(v + sZ)\) is comonotonic with \(Z\) (a concave, strictly increasing function has a positive derivative, and strict concavity makes \(U'\) strictly decreasing). The second inequality follows from Chebyshev's association inequality, and the first holds because the expectation of a strictly positive random variable is positive.

The map $\pi\mapsto f(\widehat V(\pi),\widehat s(\pi))$ is continuous on the compact set $\Pi$, so a maximizer $\widehat\pi$ exists. Put $(v^*,s^*)=(\widehat V(\widehat\pi),\widehat s(\widehat\pi))$, and fix any $\pi\in\Pi$. For $t\in(0,1]$, convexity of $\Pi$ gives $\pi_t=(1-t)\widehat\pi+t\pi\in\Pi$, and concavity of $\widehat V$ and convexity of $\widehat s$ give
\[
 \widehat V(\pi_t)\geq v_t:=(1-t)v^*+t\widehat V(\pi),\qquad
 \widehat s(\pi_t)\leq s_t:=(1-t)s^*+t\widehat s(\pi).
\]
Since $f$ is increasing in $v$ and decreasing in $s$, and $\widehat\pi$ is optimal,
\[
 f(v_t,s_t)\leq f(\widehat V(\pi_t),\widehat s(\pi_t))\leq f(v^*,s^*).
\]
The right derivative of $t\mapsto f(v_t,s_t)$ at $t=0$ is therefore nonpositive:
\[
 \frac{\partial f}{\partial v}(v^*,s^*)\{\widehat V(\pi)-v^*\}
 +\frac{\partial f}{\partial s}(v^*,s^*)\{\widehat s(\pi)-s^*\}\leq0.
\]
Divide by $\partial f/\partial v>0$ and set $k_{UB}=-(\partial f/\partial s)/(\partial f/\partial v)$, evaluated at $(v^*,s^*)$; then $k_{UB}>0$ and
\[
 \widehat V(\pi)-k_{UB}\widehat s(\pi)\leq v^*-k_{UB}s^*\qquad\text{for every }\pi\in\Pi.
\]
Thus $\widehat\pi$ solves \eqref{eq:utilityDM}. The constant $k_{UB}$ depends on $U$ and, through $(v^*,s^*)$, on the beliefs $B$. Without convexity, the comparison along the segment from $\widehat\pi$ to $\pi$ is unavailable, and the conclusion can fail, as the three-policy example after Proposition~\ref{prop:DM1} shows.

Claim (2). The first assertion of the claim for $k_{UB}=0$ follows from $\mathrm{E}_B [V(\pi)] = \widehat V(\pi)$.  The second assertion of the claim follows from Berge's maximum theorem.

Choose $\piEWM$ to have the smallest standard error among EWM policies. For any $k>0$, optimality of $\widehat\pi_{\mathrm{RW}}(k)$ gives
\[
\begin{aligned}
0&\leq\widehat V(\piEWM)-\widehat V(\widehat\pi_{\mathrm{RW}}(k))\\
 &\leq k\bigl\{\widehat s(\piEWM)-\widehat s(\widehat\pi_{\mathrm{RW}}(k))\bigr\}.
\end{aligned}
\]
Thus $\widehat s(\widehat\pi_{\mathrm{RW}}(k))\leq\widehat s(\piEWM)$. Every limit point as $k\searrow0$ is an EWM policy by the preceding argument and, by continuity of $\widehat s$, has the smallest standard error among EWM policies.
\qed

\subsection*{Proof of Proposition \ref{prop:DM2}.} Observing that
\[
\E_B \bigl[V_{\max} - V(\pi)\bigr] = \E_B \bigl[V_{\max}\bigr] - \widehat{V}(\pi),
\]
and
\[
\E_B \bigl|\widehat{V}(\pi) - V(\pi)\bigr| = \widehat{s}(\pi)\sqrt{\frac{2}{\boldsymbol{\pi}}},
\]
the program becomes
\[
\min_{\pi \in \Pi} \Bigl\{ \E_B \bigl[V_{\max}\bigr] - \widehat{V}(\pi) + k_R \,\widehat{s}(\pi)\Bigr\}.
\]
We can now omit the term \(\E_B \bigl[V_{\max}\bigr]\), which is finite by assumption, resulting in the expression stated in the claim, without affecting the solution of the program. \(\qed\)

\section{Efficient Decision Frontier Computation}\label{sec:frontier algorithm}

In this appendix, we provide pseudocode of the algorithm to compute the efficient decision frontier. 
For a finite menu, the algorithm returns points maximizing~\eqref{eq:polece} for some $k\geq0$.
The code below builds a monotonically nondecreasing, concave frontier from the input of points. Exact duplicate points should be merged before the algorithm is run, keeping a record of the policies they represent.

\begin{algorithm}
\caption{Build a Monotonically Nondecreasing, Concave Frontier}
\label{alg:build_frontier}
\begin{algorithmic}[1]

\Require $points = (Risk, Value)$ such that $Risk$ and $Value$ have the same length
\Ensure The set of points that form the efficient frontier, including EWM ties

\Procedure{FrontierEstimation}{$points$}

    \State Create data structure from input
    \State Sort by ascending $Risk$, tie-break by descending $Value$
    \State $frontier \gets$ empty list

    \For{\(i \gets 1\) to the number of rows of $points$}
        \State $p \gets \text{row $i$ of $points$}$ 
        \Comment{extract row \(i\) from $points$}
         
        \If{$frontier$ is not empty \textbf{and} $p.Value < frontier[\text{last}].Value$}
            \State \textbf{continue} 
            \Comment{skip to next iteration; ensure monotonicity}
        \EndIf
        
        \While{there are at least two elements in $frontier$}
            \State $p2 \gets \text{last frontier point}$
            \State $p1 \gets \text{second to last frontier point}$
            
            \If{$\displaystyle \frac{p.Value - p1.Value}{p.Risk - p1.Risk} 
                  > 
                  \frac{p2.Value - p1.Value}{p2.Risk - p1.Risk}$}
                \State remove $p2$ from $frontier$ 
                \Comment{ensure concavity}
            \Else
                \State \textbf{break}
            \EndIf
        \EndWhile
        
        \State append $p$ to $frontier$
    \EndFor
    
    \State \textbf{return} $frontier$ 
    \Comment{the set of frontier points}

\EndProcedure

\end{algorithmic}
\end{algorithm}

Note that the algorithm loops over sorted points, skipping any new point with a lower Value than the frontier's last point, thereby enforcing a monotonically nondecreasing frontier. It also ensures the concavity of the frontier by removing a ``bulge,'' thus guaranteeing that the slopes do not increase as we move right. Here, the bulge is checked by whether the following holds:
\[
\frac{p.Value - p1.Value}{p.Risk - p1.Risk} 
\;>\;
\frac{p2.Value - p1.Value}{p2.Risk - p1.Risk}.
\]
That is, each new point to the frontier must produce a slope that is less than or equal to the previous segment's slope.

For the investment allocation problem, it may be infeasible to discretize $\Pi$ and apply the above algorithm to plot the efficiency frontier. Instead, we start with a fine grid for $k{\geq0}$ centered at $\hat{q}_{0.95,\Pi},$ and sample points on the efficient decision frontier by solving the corresponding  optimization problem as in ~\eqref{eq:polece}. Appendix~\ref{sec:bootstrap conic root} provides a computationally efficient way to solve such optimization and therefore to plot the efficiency frontier.

\section{Treatment Choice}\label{sec:example treatment policies}
This appendix considers the problem of selecting treatment choice policies. Section \ref{sec: problem of choosing better treatment policies} describes the treatment choice problem, and Section \ref{sec: RCT applications} then demonstrates application of PoLeCe in three empirical settings using data from RCTs.

\subsection{Problem of Choosing Better Treatment Policies.}\label{sec: problem of choosing better treatment policies} 
Given a finite treatment set \(\mathcal{T} \) %
and potential outcomes \(Y(t)\) for each treatment \(t\), a policy \(\pi \in \Pi \) specifies probabilities \(\pi(t \mid X)\) of assigning treatment \(t\) given covariates \(X\), where \(\Pi \) denotes the set of allowable policies. For example, if any fractional allocation is allowed the vector of treatment probabilities can be any point in the unit simplex at each covariate value, while if treatment assignment is required to be a deterministic function of covariates then it must be a vertex at each covariate value. Other constraints can also be accommodated by appropriate specification of \(\Pi \). 

The welfare function of a policy is then defined as
\[
  V(\pi) 
  \;=\; 
  \mathrm{E}\!\Bigl[\sum_{t \in \mathcal{T}} Y(t)\,\pi\bigl(t\mid X\bigr)\Bigr],
\]
which captures the expected outcome under policy \(\pi\).

We assume we have data \((W_i)_{i=1}^n\) in the form of i.i.d.\ copies of \(W=(Y,T,X)\), where \(T\) is the treatment, \(X\) are covariates, and \(Y\) is the outcome. We also assume the data were collected under the unconfoundedness assumption, namely that the assigned treatment \(T\) is independent of the potential outcome \(Y(t)\) given \(X\). This implies that 
\(\mathrm{E}[Y(t)\mid X]\) is identified from the regression 
\[ 
g(t,X) = \mathrm{E}[Y \mid T=t,X], 
\] 
provided that \(p(t \mid X) = \Pr[T=t \mid X] > 0\).

Let \(p(t \mid X) = \Pr[T=t \mid X]\) be the propensity score, and define the Riesz representer
\[
  H(\pi) 
  \;=\; 
  \sum_{t \in \mathcal{T}}\mathbf{1}\{T = t\}\,\frac{\pi(t \mid X)}{p(t \mid X)}.
\]
Under standard conditions, including \(\mathrm{E}[H^2(\pi)]<\infty\) -- a relevant overlap condition\footnote{This requires that policies do not place weight on treatments for which there is insufficient randomization at particular values of \(X\). Note that the standard overlap condition can fail, but the welfare of a particular policy is identified, for example, if \(H(\pi)\) and \(Y\) have a finite second moment.} -- we obtain the following regression and representer formulas for identifying the policy's welfare:
\begin{equation}\label{eq:dual}
  V(\pi) 
  \;=\; 
  \mathrm{E}\biggl[\sum_{t \in \mathcal{T}}g(t,X) \,\pi(t\mid X)\biggr]
   \;=\;
  \mathrm{E}\!\bigl[Y\,H(\pi)\bigr],
\end{equation}
which can be viewed either as aggregating up regressions or as reweighting observed outcomes by how the policy \(\pi\) differs from the propensity score \(p\). Combining these two approaches leads to efficient representation and estimation of the welfare function via the doubly robust score
\citep[e.g.,][]{Newey94}:
\begin{equation}\label{eq:dr-level}
V(\pi) \;=\; \mathrm{E}[\psi_\pi(W)],
\quad \text{where} \quad
\psi_\pi(W) 
  \;:=\; 
  \sum_{t \in \mathcal{T}}g(t,X)\,\pi(t \mid X) 
  \;+\; 
  H(\pi)\,\bigl[Y - g(T,X)\bigr],
\end{equation}
whose centered version \(\psi_\pi(W)-V(\pi)\) is the influence function.

\noindent\textbf{Added Value of Policies vs. Control Policy.} It is often the case that we want to make decisions depending upon whether policies improve upon a control treatment. Let us denote by $t=0$ the control state. In this case, we can re-define:
\begin{equation}\label{eq:welfate-additive}
  V(\pi) 
  \;=\; 
  \mathrm{E}\!\Bigl[\sum_{t \in \mathcal{T}} \{Y(t) -Y(0)\} \,\pi\bigl(t\mid X\bigr)\Bigr].
\end{equation}
The efficient score for estimation and inference for $V(\pi)$ is now
\begin{equation}
    \label{eq:newey}
    \psi_\pi(W) := \sum_{t \in \mathcal{T}}( g(t,X)-g(0,X))\pi(t \mid X) +  H(\pi) \,\bigl[Y - g(T,X)\bigr] 
\end{equation}
where we re-define  
$H(\pi) = \sum_{t \in \mathcal{T}} \ \left [\frac{\mathbf{1}\{T = t\}}{p(t \mid X)}-\frac{\mathbf{1}\{T = 0\}}{p(0 \mid X)}\right] \pi(t \mid X)$.

We follow the convention of encapsulating in $X$ both covariates that are confounders as well as those for which the planner may assign different treatments.

The following inference statements require suitable moment and nuisance-estimation rate conditions and, for uniform Gaussian approximation, positive score variances and appropriate policy-class complexity, linearization, and studentization conditions; see Section~\ref{Sec: Reg Guarantees of RW policies}.
In randomized experiments, the representers are known, and we only need consistent cross-fitted learners for the regression functions to obtain consistent, asymptotically normal, efficient estimators \(\widehat{V}(\pi)\).\footnote{With randomized experimental data, the pseudo-consistency of a regression estimator for some function (though not necessarily the true function) is sufficient to satisfy all desirable properties except for asymptotic efficiency.} %
In observational studies, one can employ cross-fitted, modern statistical or machine learning methods to estimate the Riesz representers and the regressions. The resulting estimators satisfy 
\[
 \widehat{V}(\pi) = V(\pi)  + \mathbb{E}_n \bigl(\psi_\pi(W_i)- V(\pi)\bigr) + o_p\bigl(1/\sqrt{n}\bigr),
\]
uniformly over \(\pi \in \Pi\), where \(\Pi\) is either a finite set or a set whose complexity does not grow too fast. By further applying high-dimensional central limit theorems, we obtain
\[
\{\widehat{V}(\pi)\} \;\overset{a}{\sim}\; \{V(\pi) + s(\pi)\, Z_\pi\}, 
\quad Z_\pi \sim N(0,1), 
\quad s^2(\pi) = \mathrm{Var}(\psi_\pi)/n,
\]
where \(\{Z_\pi\}_{\pi \in \Pi}\) is a centered Gaussian process, with each component having standard normal distribution;  and \(\overset{a}{\sim}\) means ``approximately distributed as'' in the sense of Condition (G) in Section \ref{Sec: Reg Guarantees of RW policies}.

In the applications below, we estimate both the outcome regressions and the propensity scores using finite-dimensional parametric models. Under correct specification of both models and standard regularity conditions, full-sample fitting does not require cross-fitting, since nuisance estimation contributes no additional first-order term to the welfare estimators.

\subsection{Empirical Applications to Treatment Choice}\label{sec: RCT applications}
In this section we demonstrate application of PoLeCe to the problem of choosing better treatment policies. We use RCT data from \cite{banerjee2021selecting}, \cite{DupasRobinson2013}, and \cite{Schilbach2019} to illustrate the usefulness of our method and to show how the results can be incorporated into research papers that analyze RCTs. These papers focus on estimating the added value relative to control (as defined in \eqref{eq:welfate-additive}). In contrast, we focus on each policy's welfare level rather than its gain relative to control. While these papers originally used regression adjustment, we rely on doubly robust scores
 via \eqref{eq:dr-level} to obtain welfare estimates for potential precision gains.

In each application, we use the outcome and propensity models
\[
g(t,X)=a_t+X'b, \qquad
\log\!\left\{\frac{p(t\mid X)}{1-p(t\mid X)}\right\}=c_t+X'd_t.
\]
We fit $g$ by least squares, with treatment-specific intercepts and common covariate slopes, and $p$ by a separate binary logistic regression for each treatment, including control. Both models are fitted on the full analysis sample.
 
\subsubsection{Immunization Nudges} The first example illustrates the trade-off between welfare and precision in selecting the best arm. \cite{banerjee2021selecting} conducted   a large-scale RCT of nudges to encourage immunization in the state of
Haryana, Northern India. The study employed a cross-randomized design involving three main types of nudges: providing incentives, sending short messaging service (SMS) reminders, and seeding community ambassadors. Each type included multiple variations, resulting in  74 distinct treatments and a control. While \cite{banerjee2021selecting} proposed a method to group similar treatments for analysis, we retain the original set of treatments for illustrative purposes.

The outcomes include   the number of measles shots administered, and number of shots per dollar spent. We observe 7,370 village-month observations and use doubly robust scores to obtain the welfare estimates $\widehat{V}(\pi)$ and their standard errors. Here, $X$ consists of indicators for district--calendar year--calendar month combinations. Following  \cite{banerjee2021selecting},  we   weight both model fits and the averaging of doubly robust scores by village population, and standard errors are clustered at the sub-center (SC) level. We calculate $\widehat{k} = \widehat{q}_{1-\alpha,\Pi}$ by the bootstrap procedure in Condition~\eqref{eq:bootstrap approx}, using the SC-clustered joint covariance to calibrate the Gaussian maximum over all 75 treatment policies.

\begin{figure}[htbp]
    \centering
    \begin{subfigure}{0.5\linewidth}
        \centering
        \includegraphics[width=\linewidth]{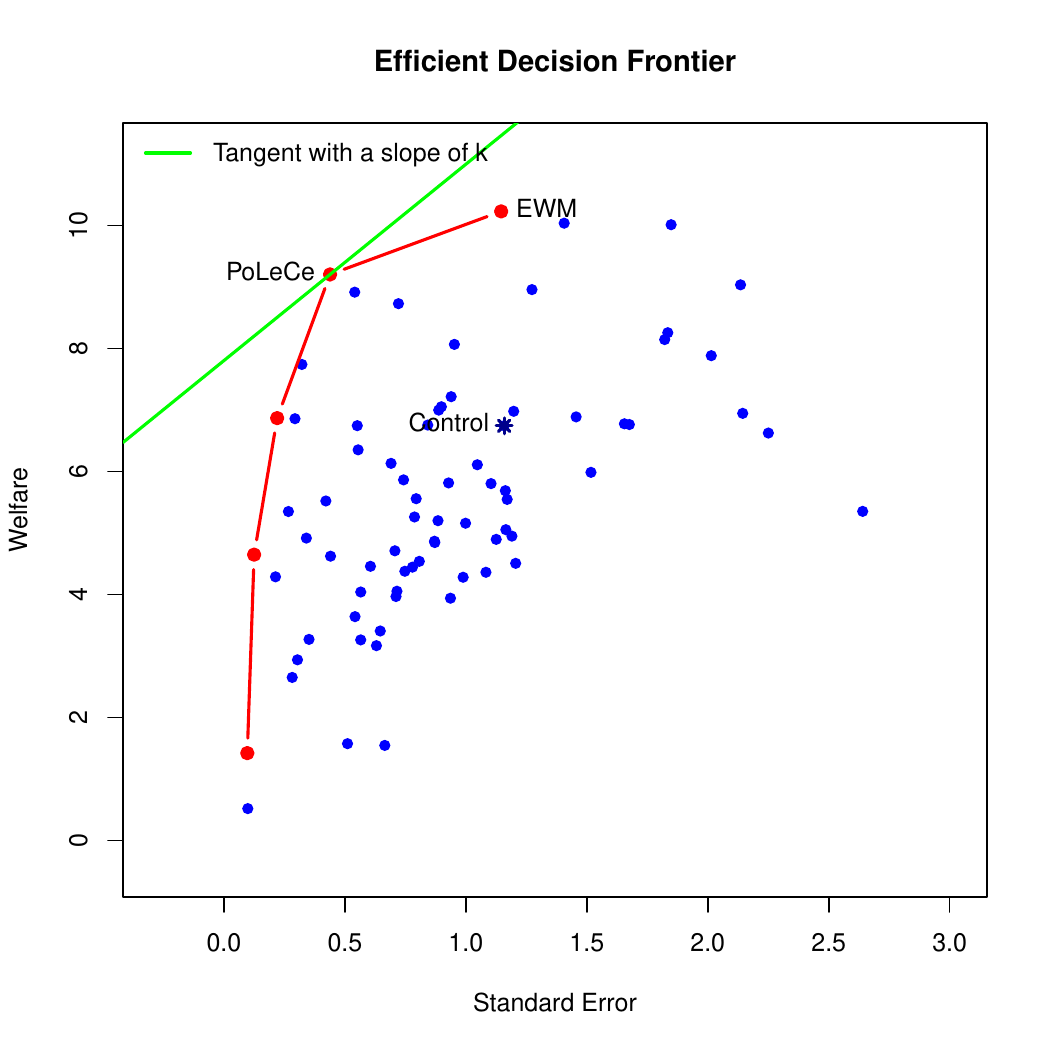}
        \caption{Measles shots outcome}
        \label{fig:shot_Measles1}
    \end{subfigure}%
    \hfill
    \begin{subfigure}{0.5\linewidth}
        \centering
        \includegraphics[width=\linewidth]{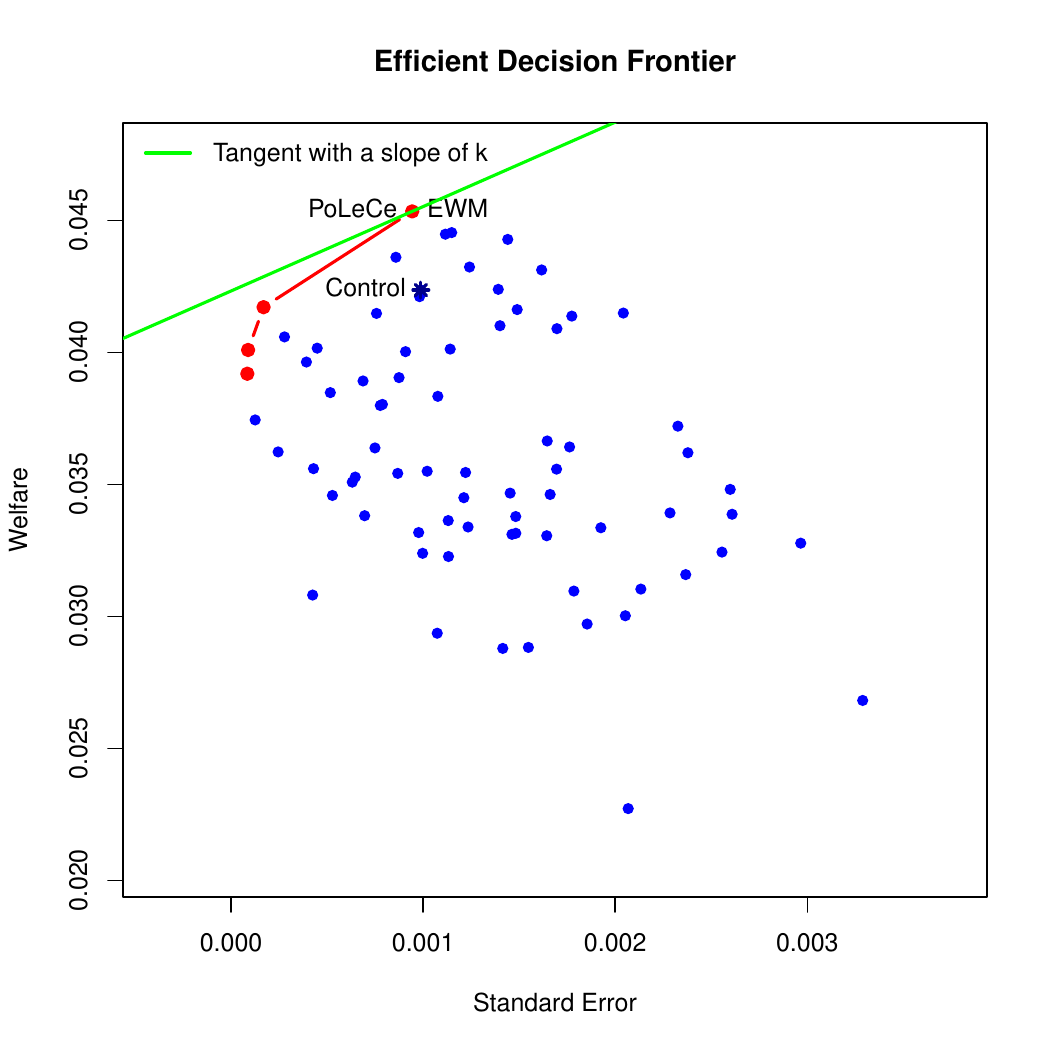}
        \caption{Shots per dollar outcome}
        \label{fig:shots_per_dollar}
    \end{subfigure}
    \caption{\label{fig:frontier} Welfare estimates against precision}
\begin{minipage}{1\textwidth} %

{\small \emph{Notes}: Each blue circle represents a pair of welfare estimate and its standard error $(\widehat{V}(\pi),\widehat{s}(\pi))$ for a given treatment $\pi$. The connected red circles show the decision frontier.
\par}
\end{minipage}
\end{figure}

As shown in Figure~\ref{fig:frontier}, there is substantial variation in the precision of welfare estimates across treatments. While EWM selects the treatment with the highest welfare estimate $\widehat{V}(\pi)$, PoLeCe adjusts for estimation uncertainty and may choose differently. Here we set $\alpha=0.05$ which places large emphasis on precision.  For example, for shots per dollar, EWM and PoLeCe both select \texttt{trusted info hub, no incentive, low SMS}, as there is minimal trade-off between welfare and precision. However, for number of shots, EWM selects \texttt{trusted info hub, low flat, no reminder}, while PoLeCe selects \texttt{info hub, low slope, high SMS} because the latter has comparable welfare but much higher precision.

\begin{table}[htbp]
{\small
\begin{center}
\caption{Results for Immunization Nudges}
\label{tab:example nudge}
\begin{tabular}{cccc}
\hline\hline
 & EWM & PoLeCe & Control\\
 \hline
 \multicolumn{4}{l}{Panel A. Outcome: Measles shots} \\ 
Estimated Value: $\widehat{V} (\widehat{\pi})$ & 10.23 & 9.21  & 6.74  \\  
LCB: $\widehat{V}(\widehat{\pi})-\widehat q_{0.95,\Pi}\widehat{s}(\widehat{\pi})$ &  6.57 & 7.80 & 3.05\\
 \hline
 \multicolumn{4}{l}{Panel B. Outcome: Shots per dollar} \\
Estimated Value: $\widehat{V} (\widehat{\pi})$ & 0.045  &  0.045  &  0.042 \\  
LCB: $\widehat{V}(\widehat{\pi})-\widehat q_{0.95,\Pi}\widehat{s}(\widehat{\pi})$ &  0.042 & 0.042 & 0.039\\
 
\hline
 \multicolumn{4}{l}{Panel C. Selected Policy: } \\ 
 \texttt{Measles shots}  & \emph{trusted info hub,} & \emph{info hub,}&  \emph{no info,}\\
   & \emph{low flat} &  \emph{low slope} &\emph{no incentive}  \\
      & \emph{no reminder} &  \emph{high SMS} &\emph{no reminder}  \\
 \texttt{Shots per dollar}  & \emph{trusted info hub,} & \emph{trusted info hub,}  &  \emph{no info,}\\
   & \emph{no incentive} & \emph{no incentive}  & \emph{no incentive} \\ 
      & \emph{low SMS} & \emph{low SMS}  & \emph{no reminder} \\ 
\hline
 
\end{tabular}
\end{center}
 \begin{minipage}{1\textwidth} %
{Note. 
EWM refers to empirical welfare maximization, 
and 
PoLeCe corresponds to policy learning with confidence.
In Panels A and B, the EWM  and PoLeCe lower confidence bounds with $\alpha = 0.05$
are provided. 
In Panel C, selected treatment combinations are given. 
\par}
\end{minipage}
}
\end{table}

\subsubsection{Informal Savings Technologies}

 \cite{DupasRobinson2013} conducted a field experiment in Kenya to study how  simple informal savings technologies can  increase investment in preventative health and reduce vulnerability to health shocks.   The experiment was run through Rotating Savings and Credit Associations (ROSCAs),
and all participants were required to be enrolled in a ROSCA at the start.\footnote{ROSCAs are informal savings groups.  Members come together on a regular basis and contribute to a common pot of money which is
taken home by one member on a rotating basis.}

Specifically, 771 ROSCA  participants belonged to ROSCAs randomized to five treatment groups:
(i) \emph{Control};
(ii) \emph{Safe Box}: participants were given a simple locked box made out of metal, while they were asked to record what health product they were saving for, and its cost, on a passbook;
(iii) \emph{Lockbox}: participants  were given a passbook and a locked box identical to those in the \emph{Safe Box} treatment, except that the program officer kept the key;
(iv) \emph{Health Pot}: participants were given a side pot that the members could contribute to in addition to the regular ROSCA pot;\footnote{Unlike the regular pot, this pot would be earmarked for a specific health product.}
(v) \emph{Health Savings Account, or HSA}: participants were encouraged to make regular deposits into an individual HSA managed by the ROSCA treasurer.

The outcomes include investments in health and measures of whether people have trouble affording medical treatments.  As an illustration, we focus on health investments in terms of amount (in Ksh) spent on preventative health products. We choose two baseline covariates: \texttt{Female} and \texttt{Married} purely for simplicity of illustration, without taking a stand on whether they \emph{should} be used in this setting.\footnote{In practice, equity, fairness or other concerns could make the use of these covariates for heterogeneous policy assignment unwarranted, and the use of other covariates, either in addition or separately, may be preferred. One could then compute the PoLeCe rule that incorporates the desired covariates.} 
We obtain 691 individuals after removing  80 participants who received multiple treatments.  The randomization was done after stratifying on some ROSCA characteristics.  
We incorporate strata indicators (omitting one reference stratum) as well as \texttt{Female} and \texttt{Married} in $X$ for calculating the doubly robust scores. We use equal observation weights in fitting both models and averaging the scores. For inference, we treat individuals as independent observations.%

\begin{table}[htbp]
{\small
\caption{Results for Informal Savings Technologies}
\begin{center}
\begin{tabular}{cccc}
\hline\hline
  & EWM & PoLeCe & Control \\
 \hline
 \multicolumn{4}{l}{Panel A. Outcome: Health Investment (in Ksh)} \\
Estimated Value: $\widehat{V} (\widehat{\pi})$ & 890  & 587 & 299  \\  
LCB: $\widehat{V}(\widehat{\pi})-\widehat q_{0.95,\Pi}\widehat{s}(\widehat{\pi})$ & 100 & 366 & 137 \\
\hline\multicolumn{3}{l}{Panel B. Selected Policy: } \\ 
(\texttt{Female}, \texttt{Married}) $=(0,0)$ & \emph{Health Pot} & \emph{Health Pot} & \emph{Control} \\
(\texttt{Female}, \texttt{Married}) $=(0,1)$ & \emph{Health Pot} & \emph{Health Pot} & \emph{Control} \\
(\texttt{Female}, \texttt{Married}) $=(1,0)$ & \emph{Health Pot} & \emph{Safe Box} & \emph{Control} \\
(\texttt{Female}, \texttt{Married}) $=(1,1)$ & \emph{Safe Box}     &  \emph{Health Pot} & \emph{Control} \\
\hline\hline
\\
\end{tabular}
\end{center}
\label{tab:example poorsavemore}
\begin{minipage}{1\textwidth} %
{Note. 
EWM refers to empirical welfare maximization, 
and 
PoLeCe corresponds to policy learning with confidence.
In Panel A, EWM  and PoLeCe lower confidence bounds with $\alpha = 0.05$
are provided along with that for the control group policy. 
In Panel B, selected treatment for each vector of (\texttt{Female}, \texttt{Married})
is given. 
\par}
\end{minipage}
}
\end{table}

There are $|\Pi|=625 = 5^4$ possible treatment policies altogether because there are five treatment values
and four possible values in the support of the two covariates. 
Figure~\ref{fig:DP2013} plots all welfare estimates relative to their precision. The highest welfare is observed below 900 Ksh, but it is accompanied by the largest standard error. Therefore, it is advisable to account for the precision of these estimates when selecting an optimal policy. To implement PoLeCe at the level $\alpha=0.05$ and calculate $\widehat{q}_{0.95,\Pi}$, we apply the bootstrap procedure from Condition~\eqref{eq:bootstrap approx}.
Table~\ref{tab:example poorsavemore} summarizes the empirical results.
In Panel A, the maximized welfare by EWM is 890 Ksh, which corresponds to the highest welfare in Figure~\ref{fig:DP2013}. 
The estimated welfare of the PoLeCe policy $\widehat{V}(\piPoLeCe)$ at $\alpha = 0.05$ is 587, which is much less than no-precision-corrected $\widehat{V} (\widehat{\pi}_{\text{EWM}})$ but substantially greater than $\widehat{V}(\widehat{\pi}_{\text{EWM}})-\widehat q_{0.95,\Pi}\widehat{s}(\widehat{\pi}_{\text{EWM}})$.  Furthermore, the chosen optimal policies differ for females. EWM selects \emph{Health Pot} for unmarried women and \emph{Safe Box} for married women. In contrast, after accounting for standard errors, PoLeCe recommends \emph{Safe Box} for unmarried women and \emph{Health Pot} for married women.

\begin{figure}[htbp]
    \centering
    \begin{subfigure}{0.5\linewidth}
        \centering
        \includegraphics[width=\linewidth]{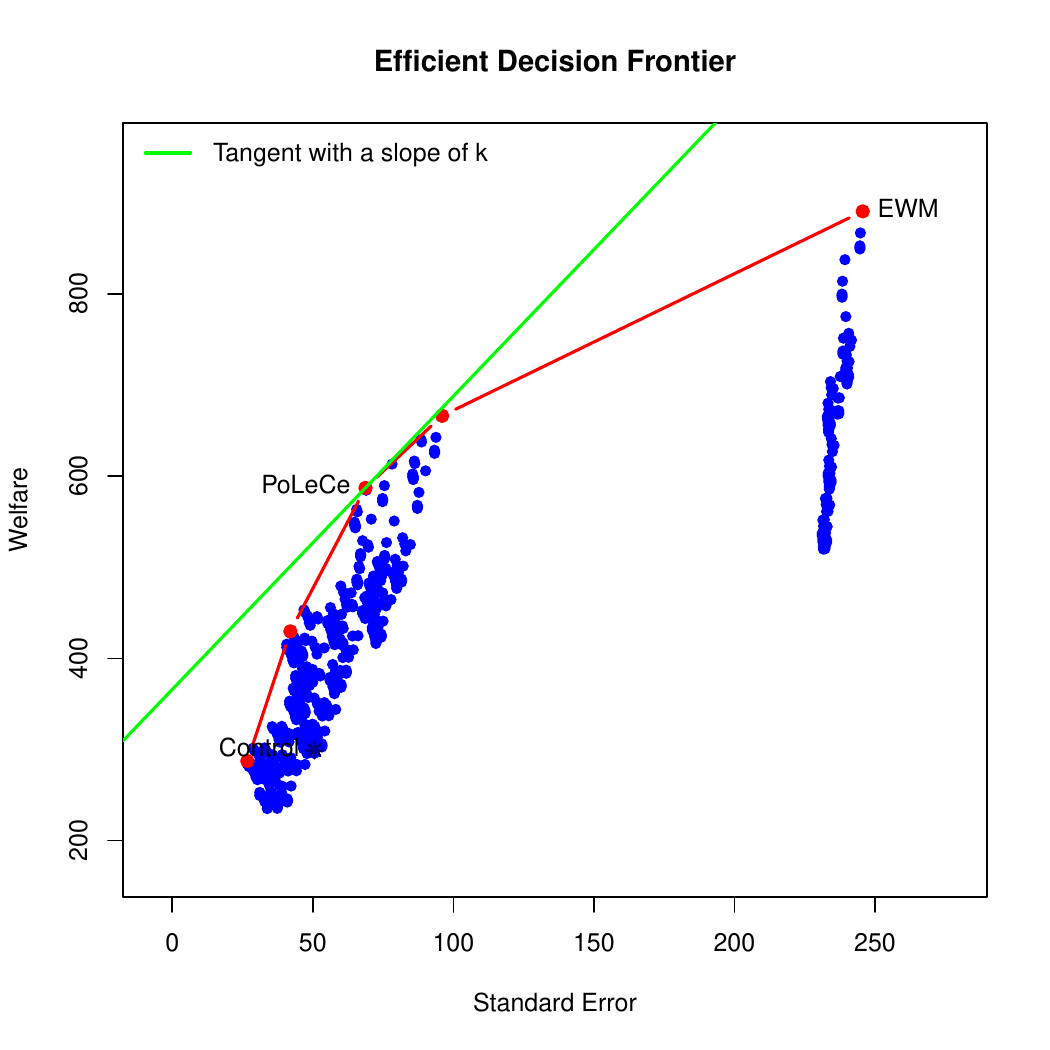}
        \caption{Health investment (in Ksh)}
        \label{fig:DP2013}
    \end{subfigure}%
    \hfill
    \begin{subfigure}{0.5\linewidth}
        \centering
        \includegraphics[width=\linewidth]{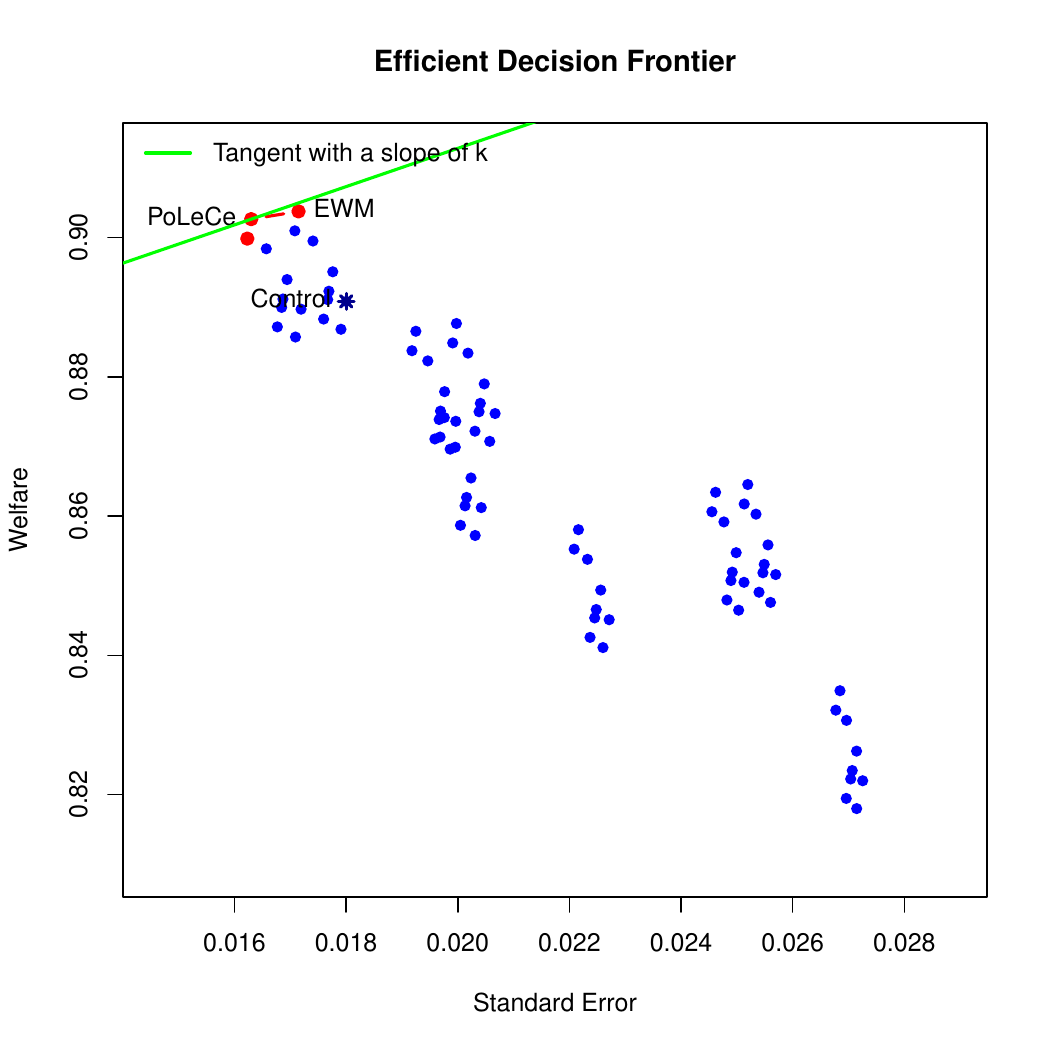}
        \caption{Labor supply}
        \label{fig:Schilbach:2019}
    \end{subfigure}
    \caption{\label{fig:frontier:additional} Welfare estimates against precision}
\begin{minipage}{.8\textwidth} %
{\tiny\emph{Notes}: Each blue circle represents a pair of welfare estimate and its standard error $(\widehat{V}(\pi),\widehat{s}(\pi))$ for a given policy $\pi$. The connected red circles show the decision frontier. %
\par}
\end{minipage}    
\end{figure}

\subsubsection{Alcohol and Self-Control}

\cite{Schilbach2019} studied alcohol consumption among low-income workers in India. Specifically, 
229 cycle-rickshaw drivers in India were randomized to three treatment groups:
(i) \emph{control group}: participants were paid Rs 90 (\$1.50) each day for visiting the study office;
(ii) \emph{incentive group}: drivers were  given incentives to remain sober---the payment was Rs 60 (\$1.00) if they arrived at the office with a positive blood alcohol content (BAC) and Rs 120 (\$2.00) if they arrived sober;
(iii) \emph{choice group}: individuals were given the same incentives earlier in the study and then offered to choose between incentives and unconditional payments in the later phase of the study. 
See \cite{Schilbach2019} for details on the experiment and background. 

As an illustration, we focus on labor supply measured by the fraction of days individuals worked 
from day 5 (the first day of sobriety incentives) through day 19 (the last day of sobriety incentives).
For simplicity, we choose two baseline covariates:
\texttt{Baseline sober} $= 1$ if the baseline fraction sober is 1 and \texttt{Baseline sober} $=0$ otherwise;
\texttt{Owns rickshaw} $ = 1$ if a driver owns a rickshaw and \texttt{Owns rickshaw} $ = 0$ otherwise. 
After removing observations with missing values, there are 222 individuals.
Given complete randomization, we calculate doubly robust scores based on the two baseline covariates, which constitute $X$ in both models. We use equal observation weights in fitting both models and averaging the scores.

\begin{table}[htbp]
\caption{Results for Self-control}
{\small
\begin{center}
\begin{tabular}{cccc}
\hline\hline
 & EWM & PoLeCe & Control \\
 \hline
 \multicolumn{4}{l}{Panel A. Outcome: Labor Supply} \\
Estimated Value: $\widehat{V} (\widehat{\pi})$ &  $0.904$  & $0.903$  & $0.891$ \\  
LCB: $\widehat{V}(\widehat{\pi})-\widehat q_{0.95,\Pi}\widehat{s}(\widehat{\pi})$ & $0.857$ & $0.858$ & $0.841$ \\
\hline
 \multicolumn{3}{l}{Panel B. Selected Policy: } \\ 
(\texttt{Baseline sober}, \texttt{Owns rickshaw}) $=(0,0)$ & Control    & Control  & Control \\
(\texttt{Baseline sober}, \texttt{Owns rickshaw}) $=(0,1)$ & Control & Incentive & Control \\
(\texttt{Baseline sober}, \texttt{Owns rickshaw}) $=(1,0)$ & Choice    & Choice & Control \\
(\texttt{Baseline sober}, \texttt{Owns rickshaw}) $=(1,1)$ & Choice     & Choice & Control \\
\hline
\\
\end{tabular}
\end{center}
\label{tab:example}
\begin{minipage}{1\textwidth} %
{Note. 
EWM refers to empirical welfare maximization, 
and 
PoLeCe corresponds to policy learning with confidence.
Panel A reports welfare estimates and lower confidence bounds with $\alpha = 0.05$ for EWM, PoLeCe, and the control policy.
In Panel B, selected treatment for each vector of (\texttt{Baseline sober}, \texttt{Owns rickshaw})
is given. 
\par}
\end{minipage}
}
\end{table}

There are \(|\Pi|=81 =  3^4 \) possible treatment policies in total since there are three treatment levels and four possible values of the two covariates.
Figure~\ref{fig:Schilbach:2019} plots the welfare estimates against their precision. Although the highest welfare appears with a relatively small standard error, this does not necessarily imply statistical significance.  To implement PoLeCe at the level $\alpha=0.05$ and calculate $\widehat{q}_{0.95,\Pi}$, we apply the bootstrap procedure from Condition~\eqref{eq:bootstrap approx}.
Table~\ref{tab:example} summarizes the empirical results. The EWM and PoLeCe estimates, along with their lower confidence bounds, are very similar and only slightly exceed those of the control group.

\section{Computational Experiments}\label{sec: computational experiments}
Here we provide results of computational experiments calibrated to the three empirical applications in Section~\ref{sec:example treatment policies}.  %
Specifically, we take the welfare estimates and their standard errors $(\widehat V(\pi),\widehat s(\pi))$ as shown in Figure~\ref{fig:frontier} (DGP 1 and 2), Figure~\ref{fig:DP2013} (DGP 3) and Figure~\ref{fig:Schilbach:2019} (DGP 4) as true parameter values $(V(\pi),s(\pi))$. For each DGP, we take independent draws from $N(V(\pi),s^2(\pi))$ for $\pi\in \Pi$ and then assess the performance of EWM and PoLeCe based on $(\widehat{V}^{\ast}(\pi),s(\pi))$ where $\widehat{V}^{\ast}(\pi)$ are the simulation draws.  We focus on $\alpha=0.05$ and since we take $\widehat{V}^{\ast}(\pi)$ to be independent across  $\pi\in \Pi$, we obtain the critical value $\widehat{q}_{1-\alpha,\Pi}$ via the bootstrap as discussed in Condition \eqref{eq:bootstrap approx} by setting the correlation matrix to be the identity matrix.  As a benchmark, we also report the performance of choosing the control policy, whose identity is known in each experiment.

Table~\ref{tab:mc-haryana} summarizes the simulation results.   In line with our theoretical results, when there is heterogeneity in the estimation risk for policies and  $\underline{\sigma}_{\Pi_0}$ (the lower bound on the risk of the best policies $\Pi_0$) is small, as in DGP 1, 2 and 4, PoLeCe outperforms EWM in terms of tail regret. In DGP 4, PoLeCe has lower mean and median regret than control, but a higher 95th percentile of regret. Finally, by construction, the lower confidence bound for PoLeCe is the highest, and is higher than that for EWM in all four DGPs considered in this simulation.

The gains are not uniform, however. In DGP 1, PoLeCe has a much smaller 95th percentile of regret than EWM but larger mean and median regret. In DGP 3, it does worse than EWM on all three regret measures. This is consistent with Proposition~\ref{thm:Polece}. In DGP 3 the best policy is also the one whose welfare is estimated least precisely (Figure~\ref{fig:DP2013}), so $\underline\sigma_{\Pi_0}$ is large and the regret bound offers no improvement over EWM. PoLeCe is designed to maximize the reported lower bound on welfare, and it does so in every design; it is not designed to minimize mean regret.

 \begin{table}[!htb]
\caption{Calibrated Simulation Results}\label{tab:mc-haryana}
\begin{center}
\begin{tabular}{llrrrr}
  \hline\hline
     & & Average  &Median  & 95\%-percentile & Average\\ 
   & & regret & regret & regret & Welfare LCB\\   
 \hline
 Measles shots & EWM & \textbf{7.98\%} & \textbf{2.12\%} & 32.10\% & -36.37\% \\
(DGP 1) & PoLeCe & 8.89\% & 10.02\% & \textbf{12.84\%} & \textbf{-21.58\%} \\
 & Control & 34.08\% & 34.08\% & 34.08\% & -70.39\% \\
\hline
Shots per dollar & EWM & 1.74\% & 1.77\% & 4.89\% & -6.68\% \\
(DGP 2) & PoLeCe & \textbf{1.02\%} & \textbf{0.00\%} & \textbf{3.83\%} & \textbf{-6.01\%} \\
 & Control & 6.57\% & 6.57\% & 6.57\% & -13.56\% \\
\hline
Informal Saving Technology & EWM & \textbf{15.90\%} & \textbf{17.21\%} & \textbf{33.85\%} & -54.14\% \\
(DGP 3) & PoLeCe & 28.13\% & 30.77\% & 41.45\% & \textbf{-45.01\%} \\
 & Control & 66.39\% & 66.39\% & 66.39\% & -87.81\% \\
\hline
Alcohol and Self-Control & EWM & 1.38\% & 1.08\% & 4.46\% & -3.66\% \\
(DGP 4) & PoLeCe & \textbf{0.84\%} & \textbf{0.47\%} & 2.21\% & \textbf{-3.46\%} \\
 & Control & 1.43\% & 1.43\% & \textbf{1.43\%} & -7.86\% \\
   \hline   
   \hline   
\end{tabular}
\end{center}
\begin{minipage}{1\textwidth} %
{Note. Based on 10,000 simulation draws with nominal level $\alpha = 0.05$.   Regret is measured in percentage $(V_{\max}-V(\widehat\pi))/V_{\max}$. The Welfare LCB is $\widehat{V}(\widehat\pi)-\widehat q_{1-\alpha, \Pi} \widehat s(\widehat\pi)$, reported as  percentage below $V_{\max}$, that is, $(\widehat{V}(\widehat\pi)-\widehat q_{1-\alpha, \Pi} \widehat s(\widehat\pi)-V_{\max})/V_{\max}$.

\par}
\end{minipage}
\end{table}

\section{Continuum Policy Classes}\label{sec:continuum}

The reporting guarantees and regret bounds of Section~\ref{Sec: Reg Guarantees of RW policies} extend to a continuum of policies when conditions (G) and (B) hold for suprema over the class. Indeed, the guarantees of the main text depend on the policy class only through the distributions of a few suprema of the standardized estimation error, and a supremum over a continuum can be well approximated, with a controlled error, by a maximum over a sufficiently fine finite set of representative policies. Moving between the continuum and a finite menu in this way has been studied in depth by \citet{CCK2014Suprema,CCK2016SPA}, and we apply their argument here. 

\paragraph{Roadmap.}
The appendix has three parts: Section~\ref{sec:continuum-regret} states the reporting and regret guarantees for a possibly uncountable policy class. They rest on the Gaussian and bootstrap approximations holding for suprema over the class. Section~\ref{sec:continuum-allocation} shows that these approximations hold for the allocation problem of our application, where a finite vector of returns drives the whole estimation-error process and no discretization is needed. Section~\ref{sec:continuum-general} treats general classes, such as treatment rules that depend on continuous covariates, by the finite-net argument, and shows how to approximate a continuum with a sufficiently fine discrete approximation.

\subsection{Reporting and regret over a continuum}\label{sec:continuum-regret}

Let $\Pi$ be a possibly uncountable policy class, which may depend on $n$. Write
\[
 \widehat Z_\pi=\frac{\widehat V(\pi)-V(\pi)}{\widehat s(\pi)},
 \qquad Z_\pi\sim N(0,1),\qquad \pi\in\Pi,
\]
where $Z$ is a centered Gaussian process and $\widehat s(\pi)>0$. Let $\widehat Z^*$ be a conditionally centered Gaussian bootstrap process with unit marginal variances and correlation function induced by the estimated covariance function for $\widehat Z_\pi$, and let ${\Pr}^*$ condition on the data. All processes and subclasses below are assumed separable, with measurable, finite suprema. Assume that $V_{\max}=\sup_{\Pi}V$ is finite, the set $\Pi_0=\argmax_{\Pi}V$ is nonempty, and the policy selections and extrema used below exist and are measurable. Compactness and suitable continuity are familiar sufficient conditions for the attainment.

We define the following key random variables,  for $K\in\{\Pi,\Pi_0\}$ and $h\in\{-1,1\}$:
\[
 T_K^h=\sup_{\pi\in K}h\widehat Z_\pi,
 \qquad M_K^h=\sup_{\pi\in K}hZ_\pi,
 \qquad T_K^{*,h}=\sup_{\pi\in K}h\widehat Z^*_\pi.
\]
The reporting guarantee uses $T_\Pi^1$; the regret bound also uses $T_{\Pi_0}^{-1}$. Since $Z$ is centered Gaussian, $M_K^{-1}$ and $M_K^1$ have the same distribution. Thus it suffices to establish, for a deterministic $r_n\to0$,
\begin{equation}\label{eq:continuumG}
 \sup_{t\in\mathbb R}\left|\Pr(T_K^h\leq t)-\Pr(M_K^h\leq t)\right|\leq r_n,
 \qquad K\in\{\Pi,\Pi_0\},\quad h\in\{-1,1\},
\end{equation}
and, on a data event of probability at least $1-\delta_n$, where $\delta_n\to0$,
\begin{equation}\label{eq:continuumB}
 \sup_{t\in\mathbb R}\left|{\Pr}^*(T_\Pi^{*,1}\leq t)-\Pr(M_\Pi^1\leq t)\right|\leq r_n.
\end{equation}
These conditions extend (G) and (B) to the suprema needed for policy choice.

The critical value now accounts for the whole policy class. Let
\[
 q_{u,K}=u\text{-Quantile}(M_K^1),\qquad
 \widehat q_{u,\Pi}=u\text{-Quantile}^*(T_\Pi^{*,1}),
\]
and define $\piPoLeCe$ to maximize $\widehat V(\pi)-\widehat q_{1-\alpha,\Pi}\widehat s(\pi)$ over $\Pi$. Denote this maximum by $L_\Pi$ and put
\[
 s_{0,n}=\min_{\pi\in\Pi_0}\widehat s(\pi),
 \qquad \bar s_n=\sup_{\pi\in\Pi}\widehat s(\pi).
\]
These standard errors retain the heterogeneity that drives the improvement in the regret bound.

\begin{aprop}[PoLeCe on a continuum]\label{prop:continuum}
Suppose \eqref{eq:continuumG}--\eqref{eq:continuumB} hold, the Gaussian suprema have continuous distribution functions, and the extrema specified above are attained. Let $0<\alpha,\beta<1/2$ and $r_n<\min(\alpha,\beta)$. Then
\[
 \Pr\!\left\{V(\pi)\geq\widehat V(\pi)
 -\widehat q_{1-\alpha,\Pi}\widehat s(\pi)\text{ for every }\pi\in\Pi\right\}
 \geq1-\alpha-2r_n-\delta_n.
\]
In particular, with this probability, $V_{\max}\geq V(\piPoLeCe)\geq L_\Pi$. Moreover, with probability at least $1-\alpha-\beta-3r_n-\delta_n$,
\begin{equation}\label{eq:continuum-regret}
 V_{\max}-V(\piPoLeCe)
 \leq V_{\max}-L_\Pi
 \leq s_{0,n}\{q_{1-\beta,\Pi_0}+q_{1-\alpha+r_n,\Pi}\}.
\end{equation}
For any measurable $\widehat k\geq0$, a maximizer $\widehat\pi_{\rm RW}$ of $\widehat V-\widehat k\widehat s$ satisfies, with probability at least $1-2\beta-2r_n$,
\[
 V_{\max}-V(\widehat\pi_{\rm RW})
 \leq s_{0,n}(q_{1-\beta,\Pi_0}+\widehat k)
       +\bar s_n(q_{1-\beta,\Pi}-\widehat k)_+.
\]
\end{aprop}

We note that a bounded, separable Gaussian process with unit variances satisfies a Gaussian anti-concentration bound, so the continuity requirement is natural here. 

The key point perhaps is that for a continuum, Gaussian complexity  replaces cardinality: if $\mathrm E\sup_{\pi\in K}Z_\pi<\infty$, Gaussian concentration gives
\[
 q_{1-\beta,K}\leq\mathrm E\sup_{\pi\in K}Z_\pi+\sqrt{2\log(1/\beta)}.
\]
The  previous cardinality bound $\sqrt{2\log|K|}$ in the finite-menu analysis is not used for an infinite $K$.  Informally, we can refer to the quantity 
$$
\mathrm E\sup_{\pi\in K}Z_\pi
$$
as the Gaussian width. This width can be bounded using Dudley's entropy integral (and sharp upper and lower bounds follow from Talagrand's generic chaining bounds).

\begin{proof}
On the bootstrap approximation event, quantile comparison gives
\[
 q_{1-\alpha-r_n,\Pi}\leq\widehat q_{1-\alpha,\Pi}
 \leq q_{1-\alpha+r_n,\Pi}.
\]
By \eqref{eq:continuumG}, the event $T_\Pi^1\leq q_{1-\alpha-r_n,\Pi}$ has probability at least $1-\alpha-2r_n$. Its intersection with the bootstrap event has probability at least $1-\alpha-2r_n-\delta_n$ and yields both the uniform lower bound and the upper bound on the bootstrap critical value.

Let $\widetilde\pi_0\in\Pi_0$ be a measurable minimizer of $\widehat s$ over $\Pi_0$. On that intersection,
\[
 V_{\max}-L_\Pi
 \leq\widehat s(\widetilde\pi_0)
       (\widehat q_{1-\alpha,\Pi}-\widehat Z_{\widetilde\pi_0})
 \leq s_{0,n}(q_{1-\alpha+r_n,\Pi}+T_{\Pi_0}^{-1}).
\]
The additional event $T_{\Pi_0}^{-1}\leq q_{1-\beta,\Pi_0}$ has probability at least $1-\beta-r_n$. A union bound now proves \eqref{eq:continuum-regret}. 

For the final claim, optimality gives the deterministic inequality
\[
 V_{\max}-V(\widehat\pi_{\rm RW})
 \leq s_{0,n}(T_{\Pi_0}^{-1}+\widehat k)
       +\bar s_n(T_\Pi^1-\widehat k)_+.
\]
Then apply \eqref{eq:continuumG} separately to its two suprema and use another union bound.
\end{proof}

\subsection{Fractional program allocations}\label{sec:continuum-allocation}

We now verify the conditions of Proposition~\ref{prop:continuum} for the allocation problem of Section~\ref{sec:motivation and risk-aware policy choice}, following the approach of \cite{CLR2013}. Let $\Pi^{\rm c}$ be a nonempty compact convex subset of the probability simplex in $\mathbb R^J$, for example the set of all allocations satisfying the constraints in~\eqref{Pi Fraction Allocation Problem}. Assume $\Omega$ and $\widehat\Omega$ are positive definite, and define
\[
 U_n=\sqrt n\,\Omega^{-1/2}(\widehat R-R),\qquad
 g(\pi)=\frac{\Omega^{1/2}\pi}{\|\Omega^{1/2}\pi\|},\qquad
 A_n=\Omega^{-1/2}\widehat\Omega\,\Omega^{-1/2},\qquad
 \lambda_n=\|A_n-I_J\|_{\rm op}.
\]
Since $V(\pi)=\pi'R$, $s(\pi)=\|\Omega^{1/2}\pi\|/\sqrt n$, and $\widehat s(\pi)=\|\widehat\Omega^{1/2}\pi\|/\sqrt n$, the studentized estimation error is
\begin{equation}\label{eq:alloc-process}
 \widehat Z_\pi=\frac{s(\pi)}{\widehat s(\pi)}\,g(\pi)'U_n,
 \qquad \pi\in\Pi^{\rm c},
\end{equation}
and the limiting process is $Z_\pi=g(\pi)'N$ with $N\sim N(0,I_J)$. The whole process is thus a smooth transformation of the $J$-vector $U_n$, and each $g(\pi)$ has unit length.

\cite{CLR2013} exploit the same structure for bounding functions estimated by series methods, $\theta(v)=p(v)'\beta$, whose studentized estimation error is a normalized linear combination of the estimated coefficients. Following them, we couple $U_n$ with a standard Gaussian vector and control the effect of the coupling uniformly in $\pi$ by the Cauchy--Schwarz inequality. The supremum of the limiting process is a convex function of a $J$-dimensional standard normal vector, so the bound of \cite{Ball:1993} on the Gaussian surface area of convex sets gives an anti-concentration inequality whose constant grows only like $J^{1/4}$. In our application, where the allocations range over the whole simplex and the estimates are independent, the supremum is a chi-type variable, and the constant does not grow with $J$ at all. These inequalities then turn this uniform bound into the distributional conditions~\eqref{eq:continuumG}--\eqref{eq:continuumB}. We impose
the following condition:

\begin{itemize}
\item[(S)] For each $n$ there are, on a suitable enlargement of the probability space, a vector $N_n\sim N(0,I_J)$ and constants $\eta_n\geq0$, $0\leq\zeta_n\leq1/2$, and $\gamma_n\geq0$ such that
\[
 \Pr(\|U_n-N_n\|>\eta_n)\leq\gamma_n,\qquad
 \Pr(\lambda_n>\zeta_n)\leq\gamma_n.
\]
\end{itemize}

\begin{alemma}[Allocation problem]\label{lem:allocation}
Suppose (S) holds, and let $t\geq0$. Put
\[
 \varepsilon_n=2\eta_n+4\zeta_n(\sqrt J+t),\qquad
 r_n=2\gamma_n+e^{-t^2/2}+{c_J}\,\varepsilon_n{,\qquad c_J=8J^{1/4}} .
\]
Then, for every nonempty $K\subseteq\Pi^{\rm c}$ and $h\in\{-1,1\}$,
\[
 \sup_{x\in\mathbb R}\Bigl|\Pr\Bigl(\sup_{\pi\in K}h\widehat Z_\pi\leq x\Bigr)
 -\Pr\Bigl(\sup_{\pi\in K}hZ_\pi\leq x\Bigr)\Bigr|\leq r_n .
\]
Moreover, on the event $\{\lambda_n\leq\zeta_n\}$, whose probability is at least $1-\gamma_n$, the bootstrap process $\widehat Z^*_\pi=\pi'\widehat\Omega^{1/2}N^*/\|\widehat\Omega^{1/2}\pi\|$ satisfies the same inequality with ${\Pr}^*$ and $\widehat Z^*$ in place of $\Pr$ and $\widehat Z$. If $\Pi^{\rm c}$ is the whole simplex and $\Omega$ is diagonal, both claims hold with $c_J=2$ for $K=\Pi^{\rm c}$ and for every face $K$ of $\Pi^{\rm c}$, whatever the value of $J$.
\end{alemma}

\begin{proof} 
\emph{Step 1: standard errors.} Since $\widehat s(\pi)^2/s(\pi)^2=g(\pi)'A_ng(\pi)$ and $\|g(\pi)\|=1$,
\[
 \Bigl|\frac{\widehat s(\pi)}{s(\pi)}-1\Bigr|\leq
 \Bigl|\frac{\widehat s(\pi)^2}{s(\pi)^2}-1\Bigr|\leq\lambda_n
 \qquad\text{for every }\pi.
\]
If $\lambda_n\leq1/2$, it follows that $s(\pi)/\widehat s(\pi)\leq2$ and $|s(\pi)/\widehat s(\pi)-1|\leq2\lambda_n$.

\emph{Step 2: coupling.} Let $Z_\pi=g(\pi)'N_n$; this process has the required Gaussian law. By~\eqref{eq:alloc-process},
\[
 \widehat Z_\pi-Z_\pi=\frac{s(\pi)}{\widehat s(\pi)}\,g(\pi)'(U_n-N_n)
 +\Bigl(\frac{s(\pi)}{\widehat s(\pi)}-1\Bigr)g(\pi)'N_n .
\]
The Cauchy--Schwarz inequality gives $|g(\pi)'(U_n-N_n)|\leq\|U_n-N_n\|$ and $|g(\pi)'N_n|\leq\|N_n\|$. Hence, on the event
\[
 \mathcal E_n=\bigl\{\|U_n-N_n\|\leq\eta_n,\ \lambda_n\leq\zeta_n,\ \|N_n\|\leq\sqrt J+t\bigr\},
\]
we have $\sup_{\pi\in\Pi^{\rm c}}|\widehat Z_\pi-Z_\pi|\leq2\eta_n+2\zeta_n(\sqrt J+t)\leq\varepsilon_n$. The norm is 1-Lipschitz and $\E\|N_n\|\leq\sqrt J$, so Gaussian concentration gives $\Pr(\|N_n\|>\sqrt J+t)\leq e^{-t^2/2}$, and $\Pr(\mathcal E_n^c)\leq2\gamma_n+e^{-t^2/2}$.

\emph{Step 3: anti-concentration.} Let $S=\sup_{K}h\widehat Z_\pi$ and $M=\sup_{K}hZ_\pi$. As a function of $N_n$, $M$ is the supremum of the linear functions $\xi\mapsto hg(\pi)'\xi$, $\pi\in K$, whose coefficient vectors have unit length. Since $g$ is continuous on $\Pi^{\rm c}$, the supremum over $K$ equals the supremum over its closure, which is compact. Hence $M$ is convex and 1-Lipschitz, with a gradient of unit length almost everywhere, and Lemma~\ref{lem:anticonc} in Online Appendix~\ref{sec:screening}, with $\ell=J$ and $b_j=0$, gives $\Pr(|M-y|\leq\varepsilon)\leq8J^{1/4}\varepsilon$ for all $y\in\mathbb R$ and $\varepsilon>0$. Since $|S-M|\leq\varepsilon_n$ on $\mathcal E_n$,
\[
 \Pr(S\leq x)\leq\Pr(M\leq x+\varepsilon_n)+\Pr(\mathcal E_n^c)
 \leq\Pr(M\leq x)+{c_J}\varepsilon_n+\Pr(\mathcal E_n^c),
\]
and symmetrically $\Pr(S\leq x)\geq\Pr(M\leq x)-{c_J}\varepsilon_n-\Pr(\mathcal E_n^c)$. This proves the first claim.

\emph{Step 3$'$: the whole simplex with diagonal $\Omega$.} In this case the supremum is a chi-type variable, and its anti-concentration constant does not depend on $J$. Let $K$ be $\Pi^{\rm c}$ or one of its faces, and let $A$ be the set of programs that the allocations in $K$ may use, with $|A|=m\geq1$. Since $\Omega^{1/2}$ is diagonal with positive entries, $\{g(\pi):\pi\in K\}$ is the set of unit vectors with nonnegative entries supported on $A$. Write $\xi=hN_n$. If some $\xi_j$, $j\in A$, is positive, then $M=\|\xi_{A,+}\|$, the norm of the positive part of $(\xi_j)_{j\in A}$, which is the variable of Online Appendix~\ref{sec:full simplex}; otherwise $M=\max_{j\in A}\xi_j<0$. Conditionally on the set of positive coordinates, of size $k\geq1$, these coordinates are independent half-normal variables, so $M$ has the $\chi_k$ distribution. The $\chi_k$ density is largest at $\sqrt{k-1}$; by Stirling's lower bound $\Gamma(z)\geq\sqrt{2\pi}\,z^{z-1/2}e^{-z}$, its largest value is at most $(1-1/k)^{(k-1)/2}\sqrt{e/\pi}<1$ for $k\geq2$, and it equals $\sqrt{2/\pi}<1$ for $k=1$. On the negative half-line, $M$ has density $m\phi(x)\Phi(x)^{m-1}\leq m2^{1-m}\phi(0)<1$. Hence the density of $M$ is less than one, and $\Pr(|M-y|\leq\varepsilon)\leq2\varepsilon$ for all $y$ and $\varepsilon$. Step~3 then goes through with $c_J=2$.

\emph{Step 4: bootstrap.} Conditional on the data, $\widehat Z^*$ is a centered Gaussian process with covariance $\pi'\widehat\Omega\widetilde\pi/(\|\widehat\Omega^{1/2}\pi\|\,\|\widehat\Omega^{1/2}\widetilde\pi\|)$. Because $\widehat\Omega=\Omega^{1/2}A_n\Omega^{1/2}$, the process $\{s(\pi)/\widehat s(\pi)\}\,g(\pi)'A_n^{1/2}N^*$ with $N^*\sim N(0,I_J)$ has the same covariance, hence the same conditional law. Compare it with $Z^*_\pi=g(\pi)'N^*$, which has the law of $Z$. Using $\|A_n^{1/2}-I_J\|_{\rm op}\leq\|A_n-I_J\|_{\rm op}$, which holds eigenvalue by eigenvalue, and Step~1,
\[
 \Bigl|\frac{s(\pi)}{\widehat s(\pi)}\,g(\pi)'A_n^{1/2}N^*-g(\pi)'N^*\Bigr|
 \leq2\lambda_n\|N^*\|+2\lambda_n\|N^*\|.
\]
When $\lambda_n\leq\zeta_n$ and $\|N^*\|\leq\sqrt J+t$, the right side is at most $4\zeta_n(\sqrt J+t)\leq\varepsilon_n$. The argument of Step~3, or of Step~3$^\prime$, applied conditionally on the data, completes the proof.
\end{proof}

In the MVPF illustration, $\Omega$ and $\widehat\Omega$ are diagonal, and $\widehat\Omega/n$ has the squared standard errors on the diagonal, so $\lambda_n$ is the largest absolute relative error of the estimated variances.

\begin{acorollary}[Fixed number of programs]\label{cor:allocation}
Let $J$ and $\Pi^{\rm c}$ be fixed, let $\Omega$ be positive definite, and suppose $\sqrt n(\widehat R-R)$ converges in distribution to $N(0,\Omega)$ and $\widehat\Omega\to_P\Omega$. Then the conclusions of Proposition~\ref{prop:continuum} hold for $\Pi=\Pi^{\rm c}$ with $r_n\to0$ and $\delta_n\to0$. In particular, PoLeCe computed over all allocations in $\Pi^{\rm c}$ satisfies
\[
 \liminf_{n\to\infty}\Pr\bigl\{V(\pi)\geq\widehat V(\pi)-\widehat q_{1-\alpha,\Pi^{\rm c}}\,\widehat s(\pi)\ \text{for every }\pi\in\Pi^{\rm c}\bigr\}\geq1-\alpha,
\]
and the regret bound~\eqref{eq:continuum-regret} holds with probability at least $1-\alpha-\beta-o(1)$.
\end{acorollary}

\begin{proof}
$U_n$ converges in distribution to $N(0,I_J)$. By Strassen's theorem \citep[Section~11.6]{Dudley:2002}, there are, on an enlarged probability space, vectors $N_n\sim N(0,I_J)$ with $\Pr(\|U_n-N_n\|>\eta_n)\leq\eta_n$, where $\eta_n\to0$ is the Prokhorov distance between the two laws. Since $\lambda_n\to_P0$, there are constants $\zeta_n\to0$ with $\Pr(\lambda_n>\zeta_n)\to0$. Choose $t=t_n\to\infty$ slowly enough that $\zeta_nt_n\to0$. Lemma~\ref{lem:allocation} then gives \eqref{eq:continuumG} and \eqref{eq:continuumB} with $r_n\to0$ and $\delta_n\to0$. The remaining conditions of Proposition~\ref{prop:continuum} hold as well. The criteria $\widehat V-k\widehat s$ are continuous on the compact set $\Pi^{\rm c}$, so their maxima are attained; the set $\Pi_0=\argmax_{\pi\in\Pi^{\rm c}}\pi'R$ is a nonempty face of $\Pi^{\rm c}$; and the Gaussian suprema have continuous distribution functions by the anti-concentration inequality used in Step~3.
\end{proof}

Lemma~\ref{lem:allocation} also covers a growing number of programs. Sufficient conditions for $r_n\to0$ are $\{\eta_n+\zeta_n(\sqrt J+t_n)\}c_J\to0$, $\gamma_n\to0$, and $t_n\to\infty$. In general, with $c_J=8J^{1/4}$, the first condition means $\eta_nJ^{1/4}\to0$ and $\zeta_nJ^{3/4}\to0$ when $t_n$ grows slowly: the coupling error $\eta_n$ and the studentization error $\zeta_n\sqrt J$ are then small relative to $J^{-1/4}$. For the whole simplex with diagonal $\Omega$, as in our application, $c_J=2$, and the first condition reduces to $\eta_n\to0$ and $\zeta_n\sqrt J\to0$. The remaining factor $\sqrt J$ comes from studentization: the relative error of the standard errors multiplies $\sup_\pi|g(\pi)'N|$, which is of order $\sqrt J$. For comparison, the anti-concentration inequality of \cite{CCK2014}, whose constant grows with $\E\sup_\pi Z_\pi$, the Gaussian ``width'' of $Z$, for which $\sqrt J$ is an upper bound, would require $\{\eta_n+\zeta_n(\sqrt J+t_n)\}\sqrt J\to0$. A dimension-free constant is not available for every subset $K$ of allocations: if $\Omega=I_J$ and $K$ is the set of vertices, the supremum is the maximum of $J$ independent standard normal variables, whose largest density value is of order $\sqrt{2\log J}$. This is in the spirit of the strong approximation conditions that \cite{CLR2013} impose for series estimators with an increasing number of terms. No grid over allocations enters the argument. For polytope allocation sets, the conic procedure of Online Appendix~\ref{sec:bootstrap conic root} computes the bootstrap critical value without discretizing $\Pi^{\rm c}$, and for the full simplex with diagonal $\widehat\Omega$, Section~\ref{sec:full simplex} gives an analytical characterization of this critical value.

\subsection{General policy classes: finite approximation}\label{sec:continuum-general}

{The allocation problem is special because a $J$-dimensional vector generates the entire estimation-error process. For general classes, for example treatment rules that depend on continuous covariates, no such representation is available. The suprema in \eqref{eq:continuumG}--\eqref{eq:continuumB} can then be handled by approximating the class with a finite net, as described in the rest of this appendix.}

\subsubsection{Understanding the result through discretization}\label{sec:continuum-discretization}

A continuum can often be represented, for inference, by a sufficiently large finite menu. The idea is to choose $p=p_n$ representative policies so that maximizing over them gives nearly the same answer as taking the supremum over the full class. The required accuracy is determined by the distribution of the Gaussian maximum: as $p$ grows, that distribution can become more concentrated.

This is the finite-net argument in \citet{CCK2016SPA}, especially Theorems~2.1--2.2 and Step~1 of the proof of Theorem~2.1; see also \citet{CCK2014Suprema}. The empirical, Gaussian, and bootstrap suprema are reduced to finite maxima, and local oscillation bounds control the differences. The argument concerns these suprema and does not require a uniform coupling of the entire processes.

To see the relevant scale, let $K_n\subset K$ contain at most $p_n$ representatives. A sufficient requirement is
\begin{equation}\label{eq:negligible-discretization}
 \sup_{\pi\in K}hU_\pi-\max_{\pi\in K_n}hU_\pi
 =o_P\!\left(\frac{1}{\sqrt{\log(ep_n)}}\right),
 \qquad h\in\{-1,1\},
\end{equation}
for the empirical and Gaussian processes $U=\widehat Z,Z$, with the corresponding conditional statement for the bootstrap. Gaussian anti-concentration then makes the resulting distributional error vanish. Thus, a suitable finite $p_n$ recovers the continuum result, provided the finite-dimensional Gaussian and bootstrap approximations remain valid at that dimension. The following lemma states this argument quantitatively.

A useful semimetric $d$ measures how much the normalized scores vary across policies. For example, if $f_\pi$ is the centered, variance-one influence function for policy $\pi$, the canonical choice is
\[
 d(\pi,\widetilde\pi)=\|f_\pi-f_{\widetilde\pi}\|_{L^2(P)}
 =\{\mathrm E(Z_\pi-Z_{\widetilde\pi})^2\}^{1/2}.
\]
Assume each relevant $K$ is totally bounded under $d$. Choose a deterministic finite $\varepsilon_n$-net $K_n\subset K$ and a map $a_{n,K}:K\to K_n$ with $d(\pi,a_{n,K}(\pi))\leq\varepsilon_n$; take $a_{n,K}(\pi)=\pi$ on $K_n$. The net for $\Pi_0$ must lie in $\Pi_0$. It is a proof device, not an input the researcher must know. In particular, intersecting an arbitrary grid for $\Pi$ with $\Pi_0$ need not give such a net.

For a process $U$ indexed by $K$, consider the modulus of continuity:
\[
 \omega_U(K,\varepsilon)=
 \sup_{\substack{\pi,\widetilde\pi\in K\\d(\pi,\widetilde\pi)\leq\varepsilon}}
 |U_\pi-U_{\widetilde\pi}|.
\]
The elementary inequality underlying discretization is
\begin{equation}\label{eq:net-sandwich}
 0\leq\sup_{\pi\in K}hU_\pi-\max_{\pi\in K_n}hU_\pi
 \leq\omega_U(K,\varepsilon_n),\qquad h\in\{-1,1\}.
\end{equation}
Every policy has a representative on the net, and the right side bounds the loss from replacing it by that representative.

The next lemma separates finite-dimensional approximation from this local error. Its quantitative condition matters when the class grows: a Gaussian maximum can become concentrated, so an error that is merely $o_P(1)$ need not preserve coverage.

\begin{alemma}[Transfer from finite nets]\label{lem:continuum-transfer}
Let $p_n=\max\{|\Pi_n|,|(\Pi_0)_n|\}$. Suppose the finite vectors on each $K_n$ satisfy the rectangular Gaussian approximation condition with error $\rho_n$. Suppose also that, for some deterministic $a_n,\tau_n\geq0$,
\[
 \Pr\{\omega_{\widehat Z}(K,\varepsilon_n)>a_n\}\leq\tau_n,
 \qquad
 \Pr\{\omega_Z(K,\varepsilon_n)>a_n\}\leq\tau_n,
 \qquad K\in\{\Pi,\Pi_0\}.
\]
On a common data event of probability at least $1-\delta_n$, suppose the finite bootstrap vector on $\Pi_n$ satisfies the rectangular bootstrap approximation condition with error $\rho_n$ and
\[
 {\Pr}^*\{\omega_{\widehat Z^*}(\Pi,\varepsilon_n)>a_n\}\leq\tau_n.
\]
Then \eqref{eq:continuumG} and \eqref{eq:continuumB} hold with
\begin{equation}\label{eq:continuum-rate}
 r_n=\rho_n+2\tau_n+C a_n\sqrt{\log(e p_n)},
\end{equation}
where $C$ is a universal constant. Hence the continuum approximation is valid if $\delta_n\to0$, $\rho_n+\tau_n\to0$, and $a_n\sqrt{\log(e p_n)}\to0$. To retain a given positive finite-dimensional error rate $\rho_n$, it suffices to make $\tau_n+a_n\sqrt{\log(e p_n)}=o(\rho_n)$.
\end{alemma}

\begin{proof}
Fix $K$ and $h$, and write $S=T_K^h$, $S_n=\max_{K_n}h\widehat Z_\pi$, $M=M_K^h$, and $M_n=\max_{K_n}hZ_\pi$. By \eqref{eq:net-sandwich}, outside events of probability at most $\tau_n$ each, $0\leq S-S_n\leq a_n$ and $0\leq M-M_n\leq a_n$. Rectangular approximation applies to either sign. Consequently,
\[
 \Pr(M_n\leq t-a_n)-\rho_n-\tau_n
 \leq\Pr(S\leq t)\leq\Pr(M_n\leq t)+\rho_n,
\]
and
\[
 \Pr(M_n\leq t-a_n)-\tau_n
 \leq\Pr(M\leq t)\leq\Pr(M_n\leq t).
\]
A finite centered Gaussian maximum with unit marginal variances satisfies
\[
 \sup_t\Pr(t-a<M_n\leq t+a)
 \leq C a\sqrt{\log(e|K_n|)}.
\]
This anti-concentration bound allows arbitrary correlations, including a singular covariance matrix; see \citet{CCK2015Comparison}. Applying it to the preceding inequalities proves the first assertion, with the stated (conservative) error bound. Conditional on the common data event, the identical argument with $\widehat Z^*$ proves the bootstrap assertion.
\end{proof}

The lemma concerns deterministic classes and nets. If a policy class or computational grid is learned from the same observations, its randomness requires a separate conditional argument or a suitable uniform approximation. One convenient construction is to build the class on an independent training sample and apply the result conditionally on that sample.

\subsubsection{Verifying the local approximation}\label{sec:continuum-verify}

To connect the lemma to welfare estimation, suppose $W_1,\ldots,W_n$ are independent and identically distributed, and let $\varphi_\pi$ be a centered influence function with variance $\sigma^2(\pi)>0$. Define $f_\pi=\varphi_\pi/\sigma(\pi)$ and $\mathbb G_n f=n^{-1/2}\sum_i\{f(W_i)-Pf\}$. A useful starting point is the studentized expansion
\begin{equation}\label{eq:uniform-linear-continuum}
 \widehat Z_\pi=\mathbb G_n f_\pi+R_n(\pi),
 \qquad \sup_{\pi\in\Pi}|R_n(\pi)|\leq u_n
\end{equation}
with high probability.\footnote{Uniform linear expansions, with explicit bounds on $u_n$, hold for many estimators used in policy learning, including debiased machine learning estimators of many parameters of interest \citep{belloni2018uniformly,quintas2022finite}. \cite{BCFH:2017} establish such an expansion for a continuum of parameters of interest, such as distribution and quantile treatment effects, and show that multiplier bootstrap inference is uniformly valid over the continuum. On a finite net of $p_n$ policies, the high-dimensional Gaussian and bootstrap approximations of \cite{CCKK-AOS} and \cite{CCK-AOP}, reviewed in \cite{CCKK-Review}, apply to the maximum of the leading term $\mathbb G_nf_\pi$; \cite{CCK2014Suprema,CCK2016SPA} give such approximations for the supremum over a continuum.} The expansion gives
\[
 \omega_{\widehat Z}(K,\varepsilon_n)
 \leq\sup_{\substack{\pi,\widetilde\pi\in K\\d(\pi,\widetilde\pi)\leq\varepsilon_n}}
 |\mathbb G_n(f_\pi-f_{\widetilde\pi})|+2u_n.
\]
Thus local empirical-process bounds and a sufficiently small uniform remainder supply the sampling modulus. Gaussian local bounds supply the corresponding modulus for $Z$; conditional multiplier bounds and control of estimated scores supply the bootstrap modulus.

One familiar setting is a pointwise measurable class $\mathcal F_n=\{f_\pi:\pi\in\Pi\}$ with envelope $F_n$, suitable envelope moments, and a VC-type uniform entropy bound
\[
 \sup_Q N(\mathcal F_n,L^2(Q),\varepsilon\|F_n\|_{L^2(Q)})
 \leq(A_n/\varepsilon)^{v_n},\qquad 0<\varepsilon\leq1,
\]
where $Q$ ranges over finitely supported probability measures. 
The results of \citet{CCK2014Suprema,CCK2016SPA} provide tools for such classes, including classes whose complexity increases with $n$. Applying them requires checking their moment and rate conditions for the particular scores. For PoLeCe, the remaining task is to choose $\varepsilon_n$ so that the resulting finite approximation error and local moduli satisfy \eqref{eq:continuum-rate}. Making the mesh finer reduces oscillations but increases the covering number; both terms must be controlled.

Uniform linearization and studentization must be accurate on the scale of a maximum. For example, on a finite net, suppose the expansion with true standard deviations has remainder at most $b_n$, the relative standard-error error is at most $c_n<1/2$, and $\max_{\Pi_n}|\mathbb G_n f_\pi|=O_P(\sqrt{\log(e p_n)})$. Then the studentized remainder is
\[
 O_P\!\left(b_n+c_n\sqrt{\log(e p_n)}\right).
\]
Sufficient negligibility conditions are $b_n\sqrt{\log(e p_n)}\to0$ and $c_n\log(e p_n)\to0$. For the continuum one also needs the uniform and local controls above.

\section{Screening Out Clearly Inferior Policies}\label{sec:screening}

The PoLeCe penalty $q_{1-\alpha,\Pi}$ is a quantile of the largest standardized estimation error over the whole menu. Because the errors are standardized, every policy contributes to this maximum in the same way, including policies whose estimated welfare is so low that they could never be selected. Such policies raise the critical value, lower the reported bound, and enter the regret bound of Proposition~\ref{thm:Polece} through $q_{1-\alpha,\Pi}$. The same issue arises in testing moment inequalities, where inequalities that are far from binding inflate the critical value. The literature offers two remedies. One drops such inequalities after a preliminary step, as in the adaptive inequality selection of \cite{CLR2013}. The other recenters each inequality by a first-step confidence bound and adjusts the level, as in \cite{romano_practical_2014}; \cite{Bai/Santos/Shaikh:22} show that this approach remains valid when the number of inequalities is large relative to the sample size. We adapt the second approach. It requires no threshold for dropping a policy, and its validity argument carries over to our problem with little change.

We work in the setting of Section~\ref{Sec: Reg Guarantees of RW policies}: a finite menu of $p\geq2$ policies, estimated standard errors $\widehat s(\pi)$, the Gaussian approximation~\eqref{eq:gaussian approx}, and the Gaussian bootstrap of Section~\ref{Sec: PoLeCe and properties} together with condition~\eqref{eq:bootstrap approx}. The idealized setting~\eqref{eq:ideal} of Section~\ref{sec:motivation and risk-aware policy choice} is the special case in which all approximation errors vanish. Section~\ref{sec:screen-procedure} describes the procedure, Section~\ref{sec:screen-results} states its guarantees, including their extension to the continuum of allocations, and Section~\ref{sec:screen-proofs} proves them. Sections~\ref{sec:screen-examples} and~\ref{sec:screen-application} illustrate the procedure. Throughout, let
\[
 L(k)=\max_{\pi\in\Pi}\{\widehat V(\pi)-k\,\widehat s(\pi)\},
\]
so that PoLeCe reports $L(\widehat q_{1-\alpha,\Pi})$.

\subsection{The procedure}\label{sec:screen-procedure}
Fix $\gamma\in(0,\alpha)$. Let $\widehat k_1$ be the $(1-\gamma)$-quantile of $\max_{\pi\in\Pi}|\widehat Z^*_\pi|$ under ${\Pr}^*$, and let $\widehat{\bar q}=\widehat q_{1-\alpha+\gamma,\Pi}$ be the $(1-\alpha+\gamma)$-quantile of $\max_{\pi\in\Pi}\widehat Z^*_\pi$.
\begin{enumerate}
\item \emph{Shifts.} For each policy, compute the shift
\[
 \widehat G(\pi)=\bigl\{L(2\widehat k_1+\widehat{\bar q})-\widehat V(\pi)-\widehat k_1\widehat s(\pi)\bigr\}_+ .
\]
\item \emph{Critical value.} Let $\widehat k$ be the $(1-\alpha+\gamma)$-quantile, under ${\Pr}^*$, of $\max_{\pi\in\Pi}\{\widehat Z^*_\pi-\widehat G(\pi)/\widehat s(\pi)\}$. Conditional on the data, the shift $\widehat G(\pi)$ and the standard error $\widehat s(\pi)$ are fixed numbers, as the standard errors are in the Gaussian bootstrap of Section~\ref{Sec: PoLeCe and properties}, so this step uses the same Gaussian draws $\widehat Z^*$.
\item \emph{Selection.} Choose $\widehat\pi_{\rm S}\in\argmax_{\pi\in\Pi}\{\widehat V(\pi)-\widehat k\,\widehat s(\pi)\}$ and report $L(\widehat k)$.
\end{enumerate}
In the idealized setting, $\widehat s=s$ and the bootstrap vector has the law of $Z$, so $\widehat k_1$, $\widehat{\bar q}$, and $\widehat k$ are exact Gaussian quantiles.

PoLeCe and the screened rule have the same nominal level, $1-\alpha$. PoLeCe spends all of $\alpha$ on its critical value. The screened rule spends $\gamma$ on the first step and $\alpha-\gamma$ on the second, which is why $\widehat{\bar q}$ and $\widehat k$ are quantiles at level $1-\alpha+\gamma$ rather than $1-\alpha$. All comparisons of the two rules below, in the examples and in the application, are at this common nominal level.

Here is the idea behind the shifts. The quantity $\widehat V(\pi)+\widehat k_1\widehat s(\pi)$ is an upper confidence bound for $V(\pi)$, valid simultaneously for all policies with probability close to $1-\gamma$, and $L(2\widehat k_1+\widehat{\bar q})$ is a conservative benchmark for the welfare that the selected policy is guaranteed to reach. A positive shift $\widehat G(\pi)$ means that policy $\pi$ falls short of this benchmark even at its upper confidence bound, and the shift is the size of that shortfall in welfare units. We subtract that shortfall, measured in standard errors, from the policy's estimation error in the maximum that determines the critical value, so a policy with a large shortfall effectively drops out. A policy that is close to the benchmark, or whose welfare is imprecisely estimated, has shift $\widehat G(\pi)=0$ and enters as before. Following \cite{romano_practical_2014}, we take $\gamma=\alpha/10$.

Two features of the procedure deserve comment. First, the first step has a cost. The screened rule spends $\gamma$ of the error budget on the first step, so its second step is run at level $1-\alpha+\gamma$ rather than $1-\alpha$. The shifts are nonnegative, $\widehat G\geq0$, so the screened critical value never exceeds $\widehat{\bar q}$, that is, $\widehat k\leq\widehat{\bar q}$; here $\widehat{\bar q}$ is the critical value that the second step would use if no policy were shifted. Compared with the PoLeCe critical value $\widehat q_{1-\alpha,\Pi}$, there are two cases. If enough policies are shifted, then $\widehat k<\widehat q_{1-\alpha,\Pi}$, so that screening lowers the penalty and raises the reported bound $L(\widehat k)$. If no policy is shifted, then $\widehat k=\widehat{\bar q}$, which exceeds $\widehat q_{1-\alpha,\Pi}$ because it is a quantile of higher level. The screened penalty is then slightly larger than the PoLeCe penalty, and the reported bound is slightly lower, at the same nominal level. A screen that removes nothing is therefore not the same as no screen: it lowers the lower confidence bound by the cost of the first step. The decision to screen, and the value of $\gamma$, must be made before the data are seen, because the guarantee below is for a rule that is fixed in advance. Choosing between the screened and the unscreened bound after looking at both does not in general have nominal level $1-\alpha$.

Second, $\widehat\pi_{\rm S}$ is a risk-aware rule with a data-dependent penalty, so it lies on the efficient frontier and Proposition~\ref{thm:RWregret} applies to it. Lowering the penalty moves the selection along the frontier toward higher estimated welfare and larger standard errors. To see this, let $k'<k$, and let $\pi'$ and $\pi$ maximize $\widehat V-k'\widehat s$ and $\widehat V-k\widehat s$. Optimality of each gives
\[
 \widehat V(\pi')-k'\widehat s(\pi')\geq\widehat V(\pi)-k'\widehat s(\pi),\qquad
 \widehat V(\pi)-k\widehat s(\pi)\geq\widehat V(\pi')-k\widehat s(\pi').
\]
Adding the two inequalities gives $(k-k')\{\widehat s(\pi')-\widehat s(\pi)\}\geq0$, so $\widehat s(\pi')\geq\widehat s(\pi)$. The first inequality then gives $\widehat V(\pi')\geq\widehat V(\pi)+k'\{\widehat s(\pi')-\widehat s(\pi)\}\geq\widehat V(\pi)$.

\subsection{Main results}\label{sec:screen-results}
This section states the guarantees for the screened rule: a reporting guarantee (Proposition~\ref{prop:screen-valid}), a regret bound (Proposition~\ref{prop:screen-regret}), and their extension to the continuum of allocations (Corollary~\ref{cor:screen-allocation}). The proofs are in Section~\ref{sec:screen-proofs}.

The results need one condition beyond~\eqref{eq:gaussian approx} and~\eqref{eq:bootstrap approx}: the standard errors must be accurate in relative terms, uniformly over the menu.
\begin{itemize}
\item[(SE)] There are constants $\nu_n\in[0,1/2]$ and $\theta_n\geq0$ such that $\Pr\bigl\{\max_{\pi\in\Pi}|\widehat s(\pi)/s(\pi)-1|>\nu_n\bigr\}\leq\theta_n$.
\end{itemize}
Note that consistency of the standard errors is one ingredient in verifying~\eqref{eq:gaussian approx}; condition (SE) asks for it explicitly, at a rate. We need it because the shifts are measured in units of welfare, and $\widehat s$ converts them into standard-error units. In the allocation problem, (SE) holds with $\nu_n=\zeta_n$ and $\theta_n=\gamma_n$ under condition (S) of Online Appendix~\ref{sec:continuum-allocation}, by Step~1 of the proof of Lemma~\ref{lem:allocation}.

\paragraph{The comparison term.}
The reporting guarantee involves a comparison term $\Delta_n$, which measures how much approximation error can move the shifts. To define it, suppose $r_n<\min(\gamma,\alpha-\gamma)$. Let $k_1^\pm$ be the $(1-\gamma\pm r_n)$-quantiles of $\max_\pi|Z_\pi|$, let $\bar q^\pm$ be the $(1-\alpha+\gamma\pm r_n)$-quantiles of $\max_\pi Z_\pi$, where $Z\sim N(0,C)$ is the Gaussian vector in~\eqref{eq:gaussian approx}, and let $m(c)=\max_\pi\{V(\pi)-c\,s(\pi)\}$. Here $k_1^\pm\geq0$ because $\max_\pi|Z_\pi|\geq0$, and $\bar q^\pm>0$ because the levels exceed $1/2$. With $K^\pm=(1\pm\nu_n)(k_1^\pm+\bar q^\pm)$, put
\[
 h^-(\pi)=\frac{\{m(K^+)-V(\pi)\}_+}{(1+\nu_n)s(\pi)},\qquad h^+(\pi)=\frac{\{m(K^-)-V(\pi)\}_+}{(1-\nu_n)s(\pi)},
\]
so that $h^-\leq h^+$.

The gap between the two shifts reflects the approximation errors $r_n$ and $\nu_n$, and it disappears in the idealized setting. The comparison term measures what the gap costs in probability:
\begin{equation}\label{eq:screen-Delta}
 \Delta_n=\sup_{x\geq0}\Bigl[\Pr\Bigl\{\max_\pi\bigl(Z_\pi-h^-(\pi)\bigr)\geq x\Bigr\}-\Pr\Bigl\{\max_\pi\bigl(Z_\pi-h^+(\pi)\bigr)\geq x\Bigr\}\Bigr].
\end{equation}
In the idealized setting, $r_n=\nu_n=0$, so $h^-=h^+$ and $\Delta_n=0$.

Finally, let $\kappa_n=\max_\pi s(\pi)/\min_\pi s(\pi)$ be the ratio of the largest to the smallest standard error.

\begin{aprop}[Reporting guarantee with screening]\label{prop:screen-valid}
Suppose that conditions~\eqref{eq:gaussian approx}, \eqref{eq:bootstrap approx}, and (SE) hold, that $0<\gamma<\alpha<1/2$, and that $r_n<\min(\gamma,\alpha-\gamma)$. Then
\[
 \Pr\bigl\{V(\widehat\pi_{\rm S})\geq L(\widehat k)\bigr\}\geq1-\alpha-4r_n-\delta_n-\theta_n-\Delta_n .
\]
Since $V_{\max}\geq V(\widehat\pi_{\rm S})$, the reported value is also a lower confidence bound for $V_{\max}$ with the same probability. In the idealized setting the bound is $1-\alpha$. Moreover, with $\alpha$ and $\gamma$ fixed, $\Delta_n\to0$ whenever $\kappa_n\sqrt{\log p}\,\bigl(r_n+\nu_n\sqrt{\log p}\bigr)\to0$.
\end{aprop}

In plain terms, the comparison term is negligible when the menu is not too large relative to the accuracy of the Gaussian and standard-error approximations. Lemma~\ref{lem:screen-Delta} in Section~\ref{sec:screen-proofs} gives an explicit bound on $\Delta_n$.

\paragraph{Regret.}
The second result bounds the regret of the screened rule. Let $\underline\sigma_{\Pi_0}=\sqrt n\min_{\pi\in\Pi_0}\widehat s(\pi)$ as in \eqref{def:sigma_bars}. With $K^\flat=(1+\nu_n)(3k_1^++\bar q^+)$, define the deterministic shifts
\[
 h^\flat(\pi)=\frac{\bigl\{m(K^\flat)-V(\pi)-2(1+\nu_n)k_1^+s(\pi)\bigr\}_+}{(1+\nu_n)s(\pi)},
\]
and let $c^+_n$ be the $(1-\alpha+\gamma+r_n)$-quantile of $\max_\pi\{Z_\pi-h^\flat(\pi)\}$. The shifts $h^\flat$ are a lower bound for the data-dependent shifts, obtained from the left inequality in Lemma~\ref{lem:shifts}(i). A lower bound on the shifts gives an upper bound on the critical value, and so on the regret.

\begin{aprop}[Regret with screening]\label{prop:screen-regret}
Under the conditions of Proposition~\ref{prop:screen-valid}, with probability at least $1-\alpha-\beta-5r_n-\delta_n-\theta_n-\Delta_n$,
\[
 V_{\max}-V(\widehat\pi_{\rm S})\leq V_{\max}-L(\widehat k)\leq\frac{\underline\sigma_{\Pi_0}}{\sqrt n}\bigl(q_{1-\beta,\Pi_0}+c^+_n\bigr).
\]
\end{aprop}

The bound of Proposition~\ref{thm:Polece} for PoLeCe has $q_{1-\alpha+r_n,\Pi}$ in place of $c^+_n$. The constant $c^+_n$ is driven by the policies that cannot be ruled out as inferior, since $h^\flat(\pi)>0$ only when $V(\pi)+2(1+\nu_n)k_1^+s(\pi)<m(K^\flat)$. In large samples it involves only the best policies. To see this, suppose that $\Pi$ and $V$ are fixed, $s(\pi)=\sigma(\pi)/\sqrt n$, and $r_n,\nu_n\to0$. Then $K^\flat$ converges, $m(K^\flat)\to V_{\max}$, $h^\flat(\pi)=0$ for $\pi\in\Pi_0$, and $h^\flat(\pi)\to\infty$ for $\pi\notin\Pi_0$. Hence $c^+_n\to q_{1-\alpha+\gamma,\Pi_0}$, and the bound becomes
\[
 \frac{\underline\sigma_{\Pi_0}}{\sqrt n}\bigl\{q_{1-\beta,\Pi_0}+q_{1-\alpha+\gamma,\Pi_0}+o(1)\bigr\},
\]
as if the set of best policies were known. A similar conclusion holds for a growing menu when the smallest gap between optimal and suboptimal policies, divided by $k_1^+\max_\pi s(\pi)$, tends to infinity. In that case the dependence on the full menu is replaced by dependence on the optimal set; when the number of optimal policies remains bounded, the factor $\sqrt{\log p}$ disappears from the regret bound.

\begin{aremark}
If only a lower confidence bound for $V_{\max}$ is wanted, the benchmark $L(2\widehat k_1+\widehat{\bar q})$ can be replaced by $L(\widehat k_1)$, as in the inequality selection of \cite{CLR2013}, and screening is then much more aggressive. The reason is that on $\mathcal E$ we have $L(\widehat k_1)\leq V_{\max}$, by Step~1 of the proof of Lemma~\ref{lem:shifts}, so that the shifts are at most $V_{\max}-V(\pi)$ in welfare units. A failure of $V_{\max}\geq L(\widehat k)$ then forces $\max_\pi\{\widehat Z_\pi-(V_{\max}-V(\pi))/\widehat s(\pi)\}>\widehat k$, and the argument of Proposition~\ref{prop:screen-valid} applies with these simpler shifts; we omit the details. The wider benchmark $L(2\widehat k_1+\widehat{\bar q})$ is the price of a guarantee for the selected policy.
\end{aremark}

\paragraph{Fractional program allocations.}
The results extend to the continuum of allocations $\Pi^{\rm c}$ of Online Appendix~\ref{sec:continuum-allocation}, with maxima over $\Pi$ replaced by suprema over $\Pi^{\rm c}$. Recall that condition (S) there involves constants $\eta_n$, $\zeta_n$, and $\gamma_n$, and that $\varepsilon_n=2\eta_n+4\zeta_n(\sqrt J+t)$ in Lemma~\ref{lem:allocation}; the constant $\gamma_n$ is unrelated to the first-step level $\gamma$.

\begin{acorollary}[Screening over all allocations]\label{cor:screen-allocation}
Suppose condition (S) of Online Appendix~\ref{sec:continuum-allocation} holds, and let $t\geq0$. Then, provided $0<\gamma<\alpha<1/2$ and $\widetilde r_n<\min(\gamma,\alpha-\gamma)$, Propositions~\ref{prop:screen-valid} and~\ref{prop:screen-regret} hold for the screened rule over $\Pi^{\rm c}$ with the following changes: $r_n$ is replaced by $\widetilde r_n=2\gamma_n+e^{-t^2/2}+8J^{1/4}\varepsilon_n$, the constant $r_n$ of Lemma~\ref{lem:allocation} with $c_J=8J^{1/4}$, $\delta_n$ by $\gamma_n$, and $(\nu_n,\theta_n)$ by $(\zeta_n,\gamma_n)$. In Lemma~\ref{lem:screen-Delta}, the constant $\varpi$ is replaced by $4J^{1/4}$, $\sqrt{2\log p}$ and $\sqrt{2\log(2p)}$ are replaced by $\sqrt J$, and $\kappa_n=\sup_{\Pi^{\rm c}}s/\inf_{\Pi^{\rm c}}s$. In particular, if $J$ is fixed and the conditions of Corollary~\ref{cor:allocation} hold, then
\[
 \liminf_{n\to\infty}\Pr\bigl\{V(\widehat\pi_{\rm S})\geq L(\widehat k)\bigr\}\geq1-\alpha .
\]
\end{acorollary}

\subsection{Proofs}\label{sec:screen-proofs}

\paragraph{Roadmap.}
The guarantee for the selected policy can fail in only two ways. Either the first step misleads us, so that some policy's welfare lies outside its band, which has probability about $\gamma$. Or the first step is sound, but the largest standardized estimation error, net of the shifts, still reaches the critical value $\widehat k$, which has probability about $\alpha-\gamma$ by the choice of $\widehat k$. The two probabilities add up to $\alpha$.

The difficulty is that the shifts are computed from the data, so the second event is not a fixed function of a Gaussian vector. Lemma~\ref{lem:shifts} handles this on the event that the first step is sound, by trapping the data-dependent shifts between shifts that use the true welfare $V$ in place of its estimate. In the idealized setting these shifts are constants, and the proof is short. In general, we replace them by deterministic upper and lower versions, at the price of a comparison term $\Delta_n$, which Lemma~\ref{lem:screen-Delta} shows to be small with the help of two facts about Gaussian maxima. A reader who wants only the idea can read the procedure, Lemma~\ref{lem:shifts} and the argument that follows it, and then turn to Section~\ref{sec:screen-examples}.

\paragraph{Two facts about Gaussian maxima.}
The approximation arguments below rest on two facts. Both concern the maximum of unit-variance Gaussian variables, and both hold whatever the number of variables and however the variables are correlated. To state them, we write the Gaussian vector in terms of independent components. If $(Z_\pi)$ has correlation matrix $C$ of rank $\ell$, then $C=AA^\top$ for a matrix $A$ with $\ell$ columns, whose rows $a_j$ have unit length because the diagonal of $C$ is one. The vector $(Z_\pi)$ then has the law of $(a_j^\top\xi)$ with $\xi\sim N(0,I_\ell)$. In the allocation problem, the vectors $a_j$ are the vectors $g(\pi)$ of Online Appendix~\ref{sec:continuum-allocation}, and $\xi=N$. A maximum $\sup_ja_j^\top\xi$ is a convex function of $\xi$, and it is 1-Lipschitz, because every $a_j$ has unit length.

The first fact is that a maximum of standardized variables is no more spread out than a single standard normal variable. A small change in the confidence level therefore moves the critical value by at most what it would move for one normal variable, however many variables there are. We use this to show that a perturbation of order $r_n$ in a level causes only a small perturbation in a critical value.

\begin{alemma}[Quantile spacing]\label{lem:spacing}
Let $M=\sup_{j\in\mathcal J}a_j^\top\xi$, where $\xi\sim N(0,I_\ell)$ and $\|a_j\|=1$ for every $j$, and let $Q_M$ be the quantile function of $M$. Then $Q_M(v)-Q_M(u)\leq\Phi^{-1}(v)-\Phi^{-1}(u)$ for $0<u<v<1$.
\end{alemma}

\begin{proof}
The function $\xi\mapsto M$ is 1-Lipschitz, because $\|a_j\|=1$ for every $j$. Fix $u\in(0,1)$ and $r\geq0$, and let $A=\{M\leq Q_M(u)\}$, so that $\Pr(\xi\in A)\geq u$. By the Gaussian isoperimetric inequality \citep{Sudakov/Tsirelson:1974,Borell:1975}, the $r$-neighborhood $A_r$ satisfies $\Pr(\xi\in A_r)\geq\Phi\{\Phi^{-1}(u)+r\}$. Since $M$ is 1-Lipschitz, $A_r\subseteq\{M\leq Q_M(u)+r\}$. Taking $r=\Phi^{-1}(v)-\Phi^{-1}(u)$ gives $\Pr\{M\leq Q_M(u)+r\}\geq v$, hence $Q_M(v)\leq Q_M(u)+r$.
\end{proof}

The second fact is that a maximum cannot place much probability on a short interval. We use this to show that a small perturbation of a shift causes only a small perturbation in a probability. For $p$ variables the constant is of order $\sqrt{\log p}$, and for a process driven by $\ell$ independent components it is of order $\ell^{1/4}$.

\begin{alemma}[Anti-concentration]\label{lem:anticonc}
Let $T=\sup_{j\in\mathcal J}(a_j^\top\xi-b_j)$, where $\xi\sim N(0,I_\ell)$, $\|a_j\|=1$, and either $\mathcal J$ is finite or $\mathcal J$ is compact and $j\mapsto(a_j,b_j)$ is continuous. Then, for all $x\in\mathbb R$ and $w>0$,
\[
 \Pr(x\leq T<x+w)\leq4\ell^{1/4}w .
\]
If $\mathcal J$ is finite with $p\geq2$ elements, the bound $(\sqrt{2\log p}+2)\,w$ also holds.
\end{alemma}

\begin{proof}
The second bound is Nazarov's inequality \citep{Nazarov:2003,CCK:2017AoP} applied to the rectangle $\{a_j^\top\xi<x+b_j+w\text{ for all }j\}$. The first bound has a geometric proof. Each $a_j^\top\xi-b_j$ is an affine function whose slope vector $a_j$ has unit length. Hence $T$, their maximum, is convex and attains its supremum, and wherever $T$ is differentiable its gradient is $a_j$ for a maximizing index $j$, so that $\|\nabla T\|=1$ almost everywhere. Since the gradient has unit length, the shell $\{x\leq T<x+w\}$ has thickness $w$, so its Gaussian probability is $w$ times the Gaussian surface area of a level set, averaged over levels. Formally, by the coarea formula, $\Pr(x\leq T<x+w)=\int_x^{x+w}\int_{\{T=t\}}\phi_\ell\,d\mathcal H^{\ell-1}\,dt$, where $\phi_\ell$ is the standard normal density on $\mathbb R^\ell$ and $\mathcal H^{\ell-1}$ is surface measure. The inner integral is the Gaussian surface measure of the convex set $\{T\leq t\}$. By the theorem of \cite{Ball:1993}, no convex set in $\mathbb R^\ell$ has Gaussian surface measure larger than $4\ell^{1/4}$.
\end{proof}

\paragraph{Control of the shifts.}
The next lemma holds the idea of the argument, and it is pure algebra, with no probability in it. It uses the analogue of $L(c)$ with the true welfare in place of its estimate, $\widehat m(c)=\max_{\pi\in\Pi}\{V(\pi)-c\,\widehat s(\pi)\}$. Part~(i) says that, whenever the first step is sound, the shifts computed from the data are squeezed between two quantities that do not involve $\widehat V$. Part~(ii) says that, if the guarantee fails, the largest shifted estimation error is at least $\widehat k$.

\begin{alemma}[Control of the shifts]\label{lem:shifts}
On the event $\mathcal E=\{\max_\pi|\widehat Z_\pi|\leq\widehat k_1\}$:
\begin{enumerate}
\item[(i)] for every $\pi$,
\[
 \bigl\{\widehat m(3\widehat k_1+\widehat{\bar q})-V(\pi)-2\widehat k_1\widehat s(\pi)\bigr\}_+\leq\widehat G(\pi)\leq\widetilde G(\pi):=\bigl\{\widehat m(\widehat k_1+\widehat{\bar q})-V(\pi)\bigr\}_+ ;
\]
\item[(ii)] if $V(\widehat\pi_{\rm S})<L(\widehat k)$, then $\max_\pi\{\widehat Z_\pi-\widetilde G(\pi)/\widehat s(\pi)\}\geq\widehat k$.
\end{enumerate}
\end{alemma}

\begin{proof}
\emph{Step 1: a sandwich for $L$.} On $\mathcal E$ we have $|\widehat Z_\pi|\leq\widehat k_1$ for every $\pi$, which means
\[
 V(\pi)-\widehat k_1\widehat s(\pi)\leq\widehat V(\pi)\leq V(\pi)+\widehat k_1\widehat s(\pi).
\]
Subtract $c\,\widehat s(\pi)$ throughout and maximize over $\pi$. This gives $\widehat m(c+\widehat k_1)\leq L(c)\leq\widehat m(c-\widehat k_1)$ for every $c$.

\emph{Step 2: part~(i).} Take $c=2\widehat k_1+\widehat{\bar q}$ in Step~1, so that $\widehat m(3\widehat k_1+\widehat{\bar q})\leq L(c)\leq\widehat m(\widehat k_1+\widehat{\bar q})$. The display in Step~1 also gives $V\leq\widehat V+\widehat k_1\widehat s\leq V+2\widehat k_1\widehat s$. Now the shift $\widehat G=\{L(c)-(\widehat V+\widehat k_1\widehat s)\}_+$ grows with the benchmark $L(c)$ and shrinks with the upper bound $\widehat V+\widehat k_1\widehat s$. Inserting the upper bound on $L(c)$ and the lower bound on $\widehat V+\widehat k_1\widehat s$ therefore gives $\widehat G\leq\widetilde G$, and inserting the lower bound on $L(c)$ and the upper bound on $\widehat V+\widehat k_1\widehat s$ gives the left inequality.

\emph{Step 3: part~(ii).} Write $\widehat\pi=\widehat\pi_{\rm S}$. Since $\widehat G\geq0$, we have $\widehat k\leq\widehat{\bar q}$. Step~1 and the fact that $\widehat m$ is nonincreasing give $L(\widehat k)\geq\widehat m(\widehat k_1+\widehat k)\geq\widehat m(\widehat k_1+\widehat{\bar q})$. By the definition of $\widehat Z$, we have $L(\widehat k)=V(\widehat\pi)+\widehat s(\widehat\pi)(\widehat Z_{\widehat\pi}-\widehat k)$. If $V(\widehat\pi)<L(\widehat k)$, then $\widehat s(\widehat\pi)(\widehat Z_{\widehat\pi}-\widehat k)=L(\widehat k)-V(\widehat\pi)\geq\widehat m(\widehat k_1+\widehat{\bar q})-V(\widehat\pi)$. The left side is positive, so it is at least $\widetilde G(\widehat\pi)$, and dividing by $\widehat s(\widehat\pi)$ gives~(ii).
\end{proof}

\paragraph{The guarantee when the quantiles are known.}
The argument is easiest to see in the idealized setting, in which $\widehat s=s$, the bootstrap vector has the law of $Z$, and $\widehat k_1$ and $\widehat{\bar q}$ are exact Gaussian quantiles, hence constants. Then the standardized shift $\widetilde h:=\widetilde G/s$ is a vector of constants. By Lemma~\ref{lem:shifts}(i), $\widehat G/s\leq\widetilde h$ on $\mathcal E$, so $\max_\pi\{\widehat Z^*_\pi-\widehat G(\pi)/s(\pi)\}\geq\max_\pi\{\widehat Z^*_\pi-\widetilde h(\pi)\}$. Hence $\widehat k$ is at least the $(1-\alpha+\gamma)$-quantile $\widetilde k$ of $\max_\pi\{Z_\pi-\widetilde h(\pi)\}$. Now suppose the guarantee fails. Then either $\mathcal E$ fails, which has probability $\gamma$, or, by Lemma~\ref{lem:shifts}(ii), $\max_\pi\{Z_\pi-\widetilde h(\pi)\}\geq\widehat k\geq\widetilde k$, which has probability $\alpha-\gamma$. The failure probability is therefore at most $\alpha$.

\paragraph{Deterministic brackets in general.}
In general, $\widehat k_1$, $\widehat{\bar q}$, and $\widehat s$ are random, so the shifts $\widetilde G/\widehat s$ in Lemma~\ref{lem:shifts}(ii) are random as well, and the argument above does not apply as it stands. The remedy is to trap the shifts between the two deterministic shifts $h^-$ and $h^+$ of Section~\ref{sec:screen-results}. We now check that they do. Let $\mathcal F$ be the event on which the bootstrap approximation~\eqref{eq:bootstrap approx} holds and $\max_\pi|\widehat s(\pi)/s(\pi)-1|\leq\nu_n$; thus $\Pr(\mathcal F^c)\leq\delta_n+\theta_n$. On $\mathcal F$, the events $\{\max_\pi|z_\pi|\leq x\}$ and $\{\max_\pi z_\pi\leq x\}$ are rectangles, so the argument of Lemma~\ref{lemma:quantile comparison} gives $k_1^-\leq\widehat k_1\leq k_1^+$ and $\bar q^-\leq\widehat{\bar q}\leq\bar q^+$. Together with $(1-\nu_n)s\leq\widehat s\leq(1+\nu_n)s$ and the nonnegativity of the four quantiles, this gives $K^-s\leq(\widehat k_1+\widehat{\bar q})\widehat s\leq K^+s$. Hence $m(K^+)\leq\widehat m(\widehat k_1+\widehat{\bar q})\leq m(K^-)$, and
\begin{equation}\label{eq:screen-bracket}
 h^-\leq\widetilde G/\widehat s\leq h^+\qquad\text{on }\mathcal F.
\end{equation}

\begin{proof}[Proof of Proposition~\ref{prop:screen-valid}]
A failure is the event $\{V(\widehat\pi_{\rm S})<L(\widehat k)\}$. To bound its probability, we split the sample space into three pieces: $\mathcal E^c\cap\mathcal F$, where the first step misleads us; $\mathcal F^c$, where an approximation fails; and $\mathcal E\cap\mathcal F$, where neither happens. The third piece is the delicate one.

\emph{Step 1: the first step misleads us.} On $\mathcal F$ we have $\widehat k_1\geq k_1^-$, so $\mathcal E^c\subseteq\{\max_\pi|\widehat Z_\pi|>k_1^-\}$. Condition~\eqref{eq:gaussian approx} gives $\Pr(\mathcal E^c\cap\mathcal F)\leq\Pr(\max_\pi|Z_\pi|>k_1^-)+r_n=\gamma+2r_n$.

\emph{Step 2: a failure forces a large shifted maximum.} Suppose a failure occurs on $\mathcal E\cap\mathcal F$. By Lemma~\ref{lem:shifts}(ii) and~\eqref{eq:screen-bracket}, $\max_\pi\{\widehat Z_\pi-h^-(\pi)\}\geq\widehat k$.

\emph{Step 3: the critical value is large.} By Lemma~\ref{lem:shifts}(i) and~\eqref{eq:screen-bracket}, $\widehat G/\widehat s\leq h^+$ on $\mathcal E\cap\mathcal F$. So $\widehat k$ is at least the $(1-\alpha+\gamma)$-quantile of $\max_\pi\{\widehat Z^*_\pi-h^+(\pi)\}$ under ${\Pr}^*$. Let $x^+$ be the $(1-\alpha+\gamma-r_n)$-quantile of $\max_\pi\{Z_\pi-h^+(\pi)\}$. Since $\{\max_\pi(z_\pi-h^+(\pi))\leq x\}$ is a rectangle, condition~\eqref{eq:bootstrap approx} makes the bootstrap quantile at least $x^+$. Hence $\widehat k\geq x^+$. Moreover, $x^+>0$: for $\pi_0\in\Pi_0$, the inequality $m(K^-)\leq V(\pi_0)$ gives $h^+(\pi_0)=0$, so $x^+\geq\Phi^{-1}(1-\alpha+\gamma-r_n)>0$.

\emph{Step 4: the probability of a failure on $\mathcal E\cap\mathcal F$.} Steps~2 and~3 give
\[
 \text{failure}\cap\mathcal E\cap\mathcal F\subseteq\Bigl\{\max_\pi\bigl(\widehat Z_\pi-h^-(\pi)\bigr)\geq x^+\Bigr\}.
\]
The shifts and the threshold on the right are deterministic, so~\eqref{eq:gaussian approx} and~\eqref{eq:screen-Delta} bound its probability by
\[
 \Pr\Bigl\{\max_\pi\bigl(Z_\pi-h^-(\pi)\bigr)\geq x^+\Bigr\}+r_n\leq\Pr\Bigl\{\max_\pi\bigl(Z_\pi-h^+(\pi)\bigr)\geq x^+\Bigr\}+\Delta_n+r_n=\alpha-\gamma+2r_n+\Delta_n .
\]
The last equality uses that the maximum has a continuous distribution, and the definition of $x^+$ as a quantile; the supremum in~\eqref{eq:screen-Delta} is over $x\geq0$, and $x^+>0$ by Step~3. Adding the three pieces, $\Pr(\mathcal E^c\cap\mathcal F)+\Pr(\mathcal F^c)+\Pr(\text{failure}\cap\mathcal E\cap\mathcal F)\leq(\gamma+2r_n)+(\delta_n+\theta_n)+(\alpha-\gamma+2r_n+\Delta_n)$ completes the proof. The statement about $\Delta_n$ is part of Lemma~\ref{lem:screen-Delta} below.
\end{proof}

The proof shows slightly more: with probability at least $1-\alpha-4r_n-\delta_n-\theta_n-\Delta_n$, the guarantee holds and the event $\mathcal E\cap\mathcal F$ occurs. The probability budget is easy to read off. The first step fails with probability at most $\gamma+2r_n$, the approximations (B) and (SE) fail with probability at most $\delta_n+\theta_n$, and on the remaining event the shifted maximum reaches the critical value with probability at most $\alpha-\gamma+2r_n+\Delta_n$. The level $\gamma$ spent on the first step is paid for by raising the level of the second step from $1-\alpha$ to $1-\alpha+\gamma$.

\paragraph{Bounding the comparison term.}
The comparison term $\Delta_n$ measures how much the screening shifts can move because of approximation error in the first step and in the standard errors. It is small when the two deterministic shifts $h^-$ and $h^+$ are close, because a maximum cannot put much probability on a short interval (Lemma~\ref{lem:anticonc}). With $\kappa_n$ as in Section~\ref{sec:screen-results}, let
\[
 \Delta^{(1)}_n=\Phi^{-1}(1-\gamma+r_n)-\Phi^{-1}(1-\gamma-r_n),\qquad
 \Delta^{(2)}_n=\Phi^{-1}(1-\alpha+\gamma+r_n)-\Phi^{-1}(1-\alpha+\gamma-r_n),
\]
and, for $\tau\in(0,1)$, let $H_\tau=\sqrt{2\log p}+\sqrt{2\log(1/\tau)}$.

\begin{alemma}[Comparison term]\label{lem:screen-Delta}
Let $\ell$ be the rank of $C$. Under the conditions of Proposition~\ref{prop:screen-valid}, for every $\tau\in(0,1)$,
\[
 \Delta_n\leq\tau+\varpi\,\bar w_n(\tau),\qquad \varpi=\min\bigl\{\sqrt{2\log p}+2,\;4\ell^{1/4}\bigr\},
\]
\[
 \bar w_n(\tau)=\frac{\kappa_n\bigl\{\Delta^{(1)}_n+\Delta^{(2)}_n+2\nu_n(k_1^++\bar q^+)\bigr\}+2\nu_nH_\tau}{1-\nu_n},
\]
where $k_1^+\leq\sqrt{2\log(2p)}+\sqrt{2\log\{1/(\gamma-r_n)\}}$ and $\bar q^+\leq\sqrt{2\log p}+\sqrt{2\log\{1/(\alpha-\gamma-r_n)\}}$. Consequently, with $\alpha$ and $\gamma$ fixed, $\Delta_n\to0$ whenever $\kappa_n\sqrt{\log p}\,\bigl(r_n+\nu_n\sqrt{\log p}\bigr)\to0$.
\end{alemma}

\begin{proof}
\emph{Step 1: width of the shifts.} Since $m(K^-)\geq m(K^+)$, we have $\{m(K^-)-V(\pi)\}_+\leq\{m(K^+)-V(\pi)\}_++m(K^-)-m(K^+)$. Dividing by $(1-\nu_n)s(\pi)$ and comparing with $h^-$ gives
\[
 h^+(\pi)-h^-(\pi)\leq\frac{m(K^-)-m(K^+)}{(1-\nu_n)s(\pi)}+\frac{2\nu_n}{1-\nu_n}\,h^-(\pi).
\]
The function $m$ is Lipschitz with constant $\max_\pi s(\pi)$. By Lemma~\ref{lem:spacing}, applied to $\max_\pi|Z_\pi|$ (the maximum of the $2p$ variables $\pm Z_\pi$) and to $\max_\pi Z_\pi$, the quantiles at levels $1-\gamma\pm r_n$ differ by at most $\Delta^{(1)}_n$, and those at levels $1-\alpha+\gamma\pm r_n$ by at most $\Delta^{(2)}_n$. Hence
\[
 K^+-K^-\leq\Delta^{(1)}_n+\Delta^{(2)}_n+2\nu_n(k_1^++\bar q^+).
\]
The first term in the bound on $h^+-h^-$ is therefore at most $\kappa_n(K^+-K^-)/(1-\nu_n)$. The second term is at most $2\nu_nH_\tau/(1-\nu_n)$ whenever $h^-(\pi)<H_\tau$. Thus $h^+(\pi)-h^-(\pi)\leq\bar w_n(\tau)$ whenever $h^-(\pi)<H_\tau$. The bounds on $k_1^+$ and $\bar q^+$ follow from Gaussian concentration, as in the proof of Lemma~\ref{lemma:quantile comparison}.

\emph{Step 2: truncation.} The width bound fails for policies with large shifts, because the second term grows with $h^-$. Such policies matter little, since a Gaussian maximum rarely reaches $H_\tau$. So we cap the shifts. Let $\bar h^\pm=\min(h^\pm,H_\tau)$, so that $\bar h^-\leq\bar h^+\leq\bar h^-+\bar w_n(\tau)$, and fix $x\geq0$. Lowering $h^-$ to $\bar h^-$ can only raise the first probability in~\eqref{eq:screen-Delta}. Lowering $h^+$ to $\bar h^+$ changes the event $\{\max_\pi(Z_\pi-h^+(\pi))\geq x\}$ only on $\{\max_\pi Z_\pi\geq H_\tau\}$, which has probability at most $\tau$ because $\E\max_\pi Z_\pi\leq\sqrt{2\log p}$. Hence the bracket in~\eqref{eq:screen-Delta} is at most
\[
 \tau+\Pr\Bigl\{\max_\pi\bigl(Z_\pi-\bar h^-(\pi)\bigr)\geq x\Bigr\}-\Pr\Bigl\{\max_\pi\bigl(Z_\pi-\bar h^+(\pi)\bigr)\geq x\Bigr\}.
\]

\emph{Step 3: anti-concentration.} Because $\bar h^+\leq\bar h^-+\bar w_n(\tau)$, the difference of the two probabilities is at most $\Pr\{x\leq\max_\pi(Z_\pi-\bar h^-(\pi))<x+\bar w_n(\tau)\}$. By Lemma~\ref{lem:anticonc} with $\xi$ of dimension $\ell$, this is at most $\varpi\,\bar w_n(\tau)$. The rate statement follows by choosing $\tau=\tau_n\to0$ slowly.
\end{proof}

\begin{proof}[Proof of Proposition~\ref{prop:screen-regret}]
\emph{Step 1: the critical value is not too large.} On $\mathcal E\cap\mathcal F$, the left inequality in Lemma~\ref{lem:shifts}(i), together with $\widehat m(3\widehat k_1+\widehat{\bar q})\geq m(K^\flat)$ and $\widehat s\leq(1+\nu_n)s$, gives $\widehat G/\widehat s\geq h^\flat$. Hence, by~\eqref{eq:bootstrap approx}, $\widehat k\leq c^+_n$.

\emph{Step 2: regret.} Let $\mathcal A=\{-\min_{\pi\in\Pi_0}\widehat Z_\pi\leq q_{1-\beta,\Pi_0}\}$. By~\eqref{eq:gaussian approx} and the symmetry of $Z$, we have $\Pr(\mathcal A^c)\leq\beta+r_n$. Let $\widehat\pi_0\in\argmin_{\pi\in\Pi_0}\widehat s(\pi)$. On the intersection of $\mathcal A$ with the event noted after the proof of Proposition~\ref{prop:screen-valid},
\[
 V_{\max}-V(\widehat\pi_{\rm S})\leq V_{\max}-L(\widehat k)\leq V_{\max}-\{\widehat V(\widehat\pi_0)-\widehat k\,\widehat s(\widehat\pi_0)\}
 =\widehat s(\widehat\pi_0)(\widehat k-\widehat Z_{\widehat\pi_0})\leq\frac{\underline\sigma_{\Pi_0}}{\sqrt n}\bigl(c^+_n+q_{1-\beta,\Pi_0}\bigr).
\]
The first inequality holds because the guarantee holds on this event, and the second because $L(\widehat k)$ is a maximum over all policies, including $\widehat\pi_0$.
\end{proof}

\begin{proof}[Proof of Corollary~\ref{cor:screen-allocation}]
The proof has three parts. First, Lemma~\ref{lem:shifts} involves no probability and holds verbatim with suprema over $\Pi^{\rm c}$; so does the bracket~\eqref{eq:screen-bracket}, once~\eqref{eq:gaussian approx} and~\eqref{eq:bootstrap approx} are available in the form described next. Second, the rest of the argument uses~\eqref{eq:gaussian approx} and~\eqref{eq:bootstrap approx} only through events of the form $\{\mathcal T(\widehat Z)<x\}$ and their bootstrap counterparts, where $\mathcal T(z)$ is $\sup_\pi\{z_\pi-h(\pi)\}$ for a deterministic continuous $h\geq0$, $\sup_\pi|z_\pi|$, or $\sup_{\pi\in\Pi_0}(-z_\pi)$. Each such $\mathcal T$ satisfies $|\mathcal T(z)-\mathcal T(z')|\leq\sup_\pi|z_\pi-z'_\pi|$. So it suffices to compare $\widehat Z$ with $Z$ in the supremum norm and to know that $\mathcal T(Z)$ does not concentrate. Third, we check these two facts:

\emph{Anti-concentration.} With $Z_\pi=g(\pi)^\top N$ as in Online Appendix~\ref{sec:continuum-allocation}, the variable $\mathcal T(Z)$ has the form of Lemma~\ref{lem:anticonc} with $\ell=J$, so $\Pr\{|\mathcal T(Z)-x|\leq\varepsilon\}\leq8J^{1/4}\varepsilon$. Lemmas~\ref{lem:spacing} and~\ref{lem:anticonc} apply with $\ell=J$, and $\E\sup_\pi Z_\pi\leq\E\|N\|\leq\sqrt J$.

\emph{Coupling.} On the coupling event in Step~2 of the proof of Lemma~\ref{lem:allocation}, which has probability at least $1-2\gamma_n-e^{-t^2/2}$, we have $\sup_\pi|\widehat Z_\pi-Z_\pi|\leq\varepsilon_n$. Hence $|\Pr\{\mathcal T(\widehat Z)<x\}-\Pr\{\mathcal T(Z)<x\}|\leq\widetilde r_n$. Step~4 of that proof gives the bootstrap counterpart on $\{\lambda_n\leq\zeta_n\}$, an event of probability at least $1-\gamma_n$, and Step~1 gives (SE).

Note that, for fixed $J$, all error terms vanish if $t=t_n\to\infty$ slowly, as in the proof of Corollary~\ref{cor:allocation}, because $\kappa_n$ does not depend on $n$.
\end{proof}

\subsection{Stylized examples}\label{sec:screen-examples}
Table~\ref{tab:screening} applies the procedure to three stylized menus with independent estimates, together with a menu of 100 policies for comparison. The estimates are exactly Gaussian with known standard errors, so all approximation errors in Proposition~\ref{prop:screen-valid} vanish. In all examples $\alpha=0.05$ and $\gamma=\alpha/10$. With three policies, $q_{0.95,\Pi}=2.121$, $k_1=3.143$, and $\bar q=2.164$. Since the leading policy is the same in all three examples, the benchmark is $L(2k_1+\bar q)=0.831$ in each. The last column reports the coverage of the guarantee in 10{,}000 simulated samples in which true welfare equals the displayed estimates. The replication package contains R and Python code that reproduces this table and the results in Section~\ref{sec:screen-application}.

\begin{table}[ht]
\caption{Screening in stylized menus}\label{tab:screening}
\footnotesize\setlength{\tabcolsep}{4pt}
\begin{tabular}{clcccccc}
\toprule
 Ex. & Estimates (s.e.) & Shifts $\widehat G/s$ & \multicolumn{2}{c}{Penalty} & \multicolumn{2}{c}{Reported bound} & Coverage\\
\cmidrule(lr){4-5}\cmidrule(lr){6-7}
 & & & $q_{0.95,\Pi}$ & $\widehat k$ & PoLeCe & Screened & PoLeCe/Scr.\\
\midrule
1 & 1.00 (.02), 0.60 (.02), 0.50 (.02) & 0, 8.4, 13.4 & 2.121 & 1.695 & 0.958 & 0.966 & .982 / .955\\
2 & 1.00 (.02), 0.90 (.30), 0.50 (.02) & 0, 0, 13.4 & 2.121 & 2.000 & 0.958 & 0.960 & .974 / .965\\
3 & 1.00 (.02), 0.50 (.40), 0.30 (.40) & 0, 0, 0 & 2.121 & 2.164 & 0.958 & 0.957 & .984 / .982\\
4 & 1.00 (.02); 99 at 0.50 (.02) & 0, 9.5, \ldots & 3.283 & 1.695 & 0.934 & 0.966 & .999 / .958\\
\bottomrule
\end{tabular}

\smallskip
\begin{minipage}{\textwidth}\footnotesize\emph{Notes:} Independent estimates with known standard errors, $\alpha=0.05$, $\gamma=0.005$. Both rules have nominal level $0.95$: PoLeCe uses $q_{0.95,\Pi}$, and the screened rule spends $\gamma=0.005$ on the first step and uses level $0.955$ in the second. In every example both rules select the first policy. Coverage is the simulated probability that the welfare of the selected policy is at least the reported bound, when true welfare equals the displayed estimates.\end{minipage}
\end{table}

In Example~1, the leading policy is precisely estimated and the other two are far below it. Both are screened out, and the penalty falls to $z_{0.955}=1.695$, the critical value at level $1-\alpha+\gamma$ for a menu that contains only the leading policy. The reported bound rises, and the simulated coverage falls from a conservative 0.982 to close to the nominal level. The gain in the bound is small in absolute terms, however, because the selected policy is precisely estimated, so the penalty is multiplied by a small standard error.

Example~2 replaces one of the inferior policies with a noisy rival whose estimate is close to the leader's. The rival is not screened, which is the right outcome: it is the policy that could mislead the decision maker. Only the third policy drops out, and the penalty falls to 2.000, the critical value at level $1-\alpha+\gamma$ for a menu of two policies. Because the penalty is lower, the rule is slightly more willing to choose the noisy rival, and its simulated mean regret rises slightly, from 0.0009 to 0.0013.

In Example~3, the two inferior policies are noisy. Their upper confidence bounds are high, so neither is screened, even though their estimates are far below the leader's. The penalty then rises to $\bar q=2.164$, and the reported bound falls slightly. This is the cost of the first step: both rules have nominal level $0.95$, but the screened rule has spent $\gamma=0.005$ on a first step that removes nothing. As we explain below, this is the case in our application.

Example~4 shows where screening pays off. With 99 clearly inferior policies, the PoLeCe penalty is 3.283, and the guarantee is very conservative, with simulated coverage of 0.999. Screening removes all 99 policies, the penalty returns to 1.695, and the reported bound rises by about 1.6 standard errors of the leading policy.

The examples also show how far below the benchmark a policy must be to be screened. When all standard errors are equal to $s$, a policy is shifted only if its estimate is more than $(3k_1+\bar q)s$ below the leader's, about 11.6 standard errors in the three-policy examples. For instance, with estimates $1.00$, $0.80$, and $0.70$ and standard errors of 0.02, only the third policy is shifted, and $\widehat k=2.000$. This requirement reflects the conservative benchmark needed for a guarantee on the selected policy. It becomes less demanding as the sample grows, since the gaps measured in standard errors grow at the rate $\sqrt n$.

\subsection{The MVPF application}\label{sec:screen-application}
In the application, $\Pi$ is the whole simplex and $\widehat\Omega$ is diagonal. Since both levels exceed $1/2$, the quantiles $q_{1-\alpha,\Pi}$ and $\bar q$ are then quantiles of $\|N_+\|$, and $k_1$ is the $(1-\gamma)$-quantile of $\max(\|N_+\|,\|N_-\|)$ (Online Appendix~\ref{sec:full simplex}). Conditionally on the number of positive coordinates of $N$, the two norms are independent chi variables, so both quantiles can be computed exactly. For a given allocation, the shift $\widehat G(\pi)$ is positive only when $\pi'\widehat R+k_1\widehat{s}(\pi)$ falls below the benchmark $L(2k_1+\bar q)$, so it suffices to compare the smallest value of $\pi'\widehat R+k_1\widehat{s}(\pi)$ over the simplex with the benchmark. Both are convex programs.

For programs targeting adults, $q_{0.95,\Pi}=3.186$, $\bar q=3.230$, and $k_1=4.210$. The benchmark is $L(2k_1+\bar q)=0.717$, while the smallest value of $\pi'\widehat R+k_1\widehat{s}(\pi)$ over the simplex is 0.719. No allocation is shifted, although the margin is small. Hence $\widehat k=\bar q=3.230$. The selected allocation is unchanged (97 percent OHIE and about 3 percent JTPA), and the reported bound falls from 1.038 to 1.036. For programs targeting youth, $q_{0.95,\Pi}=2.899$, $\bar q=2.944$, and $k_1=3.936$. The benchmark is 0.142, far below the smallest value of $\pi'\widehat R+k_1\widehat{s}(\pi)$, which is 0.497. Again no allocation is shifted, so $\widehat k=2.944$. The allocation is essentially unchanged (51 percent Head Start and 49 percent Wisconsin Scholars Grant), and the reported bound falls from 0.759 to 0.746. These figures use the exact critical values of Online Appendix~\ref{sec:full simplex}. Both bounds have the same nominal level, $0.95$. The screened bounds are lower because the screened rule spends $\gamma=0.005$ of the error budget on a first step that, here, removes nothing.

Screening has no effect here for the reasons illustrated by Examples~3 and~1. Several of the programs with low estimates are imprecisely estimated, such as the Seattle--Denver income maintenance experiment, the Padua pilot, Job Corps, and JTPA for youth, so their upper confidence bounds are high. The precisely estimated programs with low estimates, such as Housing Vouchers and Year Up, are not screened either: the estimate of Housing Vouchers lies below the adult benchmark, but not far enough below, and that of Year Up lies above the youth benchmark. For adults, the upper confidence bound of Housing Vouchers, $0.65+4.21\times0.02=0.734$, exceeds the benchmark $0.717$ by less than $0.02$. For youth, the benchmark is low because the two leading programs are themselves imprecisely estimated. Even the more aggressive screen for a lower bound on $V_{\max}$ alone, with benchmark $L(k_1)=0.997$ for adults, would shift only allocations concentrated on Housing Vouchers among the adult programs, and no allocation of the youth programs. Because no allocation is shifted here, the screened rule, whose nominal level is $1-\alpha$ by Corollary~\ref{cor:screen-allocation}, selects the allocation and reports the bound that PoLeCe would give at the higher level $1-\alpha+\gamma$. At the common nominal level $1-\alpha$, PoLeCe therefore does slightly better, and we report it in the main text. Screening is worth applying when the menu is large and contains many precisely estimated policies that are clearly inferior, for example treatment rules defined on many covariate cells.

\clearpage
\renewcommand{\refname}{References for the Online Appendix}
\putbib
\end{bibunit}
\end{document}